\documentclass[pra,twocolumn,nofootinbib,superscriptaddress,10pt]{revtex4-2}

\usepackage{mathdots}
\usepackage{amsfonts,amssymb,amsmath}
\usepackage[]{graphics,graphicx,epsfig}
\usepackage{amsthm}
\usepackage{csquotes}
\usepackage{mathrsfs}
\MakeOuterQuote{"}
\usepackage{epstopdf}

\def\identity{\leavevmode\hbox{\small1\kern-3.8pt\normalsize1}}

\newtheorem{theorem}{Theorem}

\newtheorem{lemma}{Lemma}

\newtheorem{proposition}{Proposition}

\newcommand{\tr}{\operatorname{tr}}
\newcommand{\1}{\operatorname{\uppercase\expandafter{\romannumeral1}}}
\newcommand{\2}{\operatorname{\uppercase\expandafter{\romannumeral2}}}
\newcommand{\3}{\operatorname{\uppercase\expandafter{\romannumeral3}}}
\newcommand{\4}{\operatorname{\uppercase\expandafter{\romannumeral4}}}
\newcommand{\5}{\operatorname{\uppercase\expandafter{\romannumeral5}}}
\newcommand{\6}{\operatorname{\uppercase\expandafter{\romannumeral6}}}

\newcommand{\iid}{\mathrm{iid}}

\newcommand{\rmi}{\mathrm{i}}

\newcommand{\rme}{\operatorname{e}}

\newcommand{\caH}{\mathcal{H}}

\newcommand{\bbZ}{\mathbb{Z}}

\def\eqref#1{\textup{(\ref{#1})}}
\newcommand{\eref}[1]{Eq.~\textup{(\ref{#1})}}
\newcommand{\lref}[1]{Lemma~\ref{#1}}

\usepackage{algpseudocode}
\usepackage{algorithm}

\def\<{\langle}  
\def\>{\rangle}

\newcommand{\rcite}[1]{Ref.~\cite{#1}}

\usepackage[bookmarks=false]{hyperref}
\hypersetup{colorlinks=true,citecolor=blue,linkcolor=red,urlcolor=blue,pdfstartview=FitH,bookmarksopen=true}

\begin{document}
	\title{Robust and efficient verification of graph states in blind measurement-based quantum computation}
	
	\author{Zihao Li}
	\affiliation{State Key Laboratory of Surface Physics and Department of Physics, Fudan University, Shanghai 200433, China}
	\affiliation{Institute for Nanoelectronic Devices and Quantum Computing, Fudan University, Shanghai 200433, China}
	\affiliation{Center for Field Theory and Particle Physics, Fudan University, Shanghai 200433, China}
	
	\author{Huangjun Zhu}
	\email{zhuhuangjun@fudan.edu.cn}
	\affiliation{State Key Laboratory of Surface Physics and Department of Physics, Fudan University, Shanghai 200433, China}
	\affiliation{Institute for Nanoelectronic Devices and Quantum Computing, Fudan University, Shanghai 200433, China}
	\affiliation{Center for Field Theory and Particle Physics, Fudan University, Shanghai 200433, China}
	
	\author{Masahito Hayashi}
	\email{hmasahito@cuhk.edu.cn}
	\affiliation{School of Data Science, The Chinese University of Hong Kong, Longgang District, Shenzhen, 518172, China}
	\affiliation{International Quantum Academy (SIQA), Futian District, Shenzhen 518048, China}
	\affiliation{Graduate School of Mathematics, Nagoya University, Nagoya 464-8602, Japan}

	\begin{abstract}
		Blind quantum computation (BQC) is a secure quantum computation method that protects the privacy of clients. Measurement-based quantum computation (MBQC) is a promising approach for realizing BQC. To obtain reliable results in blind MBQC,  it is crucial to verify whether the resource graph states are accurately prepared in the adversarial scenario. However, previous verification protocols for this task are too resource consuming or noise susceptible to be applied in practice. Here, we propose a robust and efficient  protocol for verifying arbitrary graph states with any prime local dimension in the adversarial scenario, which leads to a robust and efficient  protocol for verifying the resource state in blind MBQC. Our protocol requires only local Pauli measurements and is thus easy to realize with current technologies. Nevertheless, it can achieve  the optimal scaling behaviors with respect to the system size and the target precision as quantified by the infidelity and significance level, which has never been achieved before. Notably, our protocol can exponentially enhance the scaling behavior with the significance level.
	\end{abstract}

	\date{\today}
	\maketitle

	\begin{center}
		\textbf{INTRODUCTION}
	\end{center}
	\vspace{1em}
	
	Quantum computation offers the promise of exponential speedups over classical computation on a number of important problems \cite{shor1994,NielC00,Presk18}. However, it is very challenging to realize practical quantum computation in the near future, especially for clients with limited quantum computational power. Blind quantum computation (BQC)~\cite{Fitzs17npj} is an effective method that enables such a client to delegate his (her)  computation to a server, who is capable to perform quantum computation, without leaking any information about the computation task.
	So far, various protocols of BQC have been proposed in theory \cite{Broadbent09,RUV13,Mantri13,Morimae13} and demonstrated in experiments \cite{Greganti16,JiangWei19,BarzScience12,BarzNP13}. Many of these protocols build on the model of measurement-based quantum computation (MBQC) \cite{RaussBrie01,RaussBrowne03,BriegelDur09}, in which graph states are used as resources and local projective measurements on qudits are used to drive the computation.

	To realize BQC successfully, it is crucial to protect the privacy of the client and verify the correctness of the computation results.
	The latter task, known as verification of BQC, has been studied in various models as explained in the Methods section, among which MBQC in the 
	receive-and-measure setting is particularly convenient \cite{Fujii17,HayaM15,TMMMF19,Xu21,MTH17,HayaHajdu18}. However, it is extremely challenging to construct robust and efficient verification protocols, especially for noisy, intermediate-scale quantum (NISQ) devices \cite{Presk18,Arute19,Zhong20}.  Actually, this problem  lies at the heart of the active research field of quantum characterization, verification, and validation (QCVV) \cite{Gheorg19,Eisert20,SupicBow20,Carra21,Kliesch21,YSG21}.

	\begin{table*}
		\caption{\label{tab:compare}
			Comparison of various protocols for verifying the resource states of blind MBQC in the adversarial scenario.
			Here $n$ is the qubit (qudit) number of the resource graph state;  $\epsilon$ and $\delta$ denote the target infidelity and significance level, respectively. The optimal scaling behaviors of the test number $N$ in $n$, $\epsilon$, and  $\delta$ are $O(1)$, $O(\epsilon^{-1})$, and $O(\ln \delta^{-1})$, respectively. By 'robust to noise' we mean  the verifier Alice can accept with a high probability if the state prepared has a sufficiently high fidelity. The robustness achieved in Ref.~\cite{Fujii17} is different from the current definition. 
			The scaling behaviors with respect to $\epsilon$ and $\delta$ are not clear for protocols in Refs.~\cite{TMMMF19,Xu21,Takeuchi18,MTH17}.  See Supplementary Note~3 for details.}
		\begin{math} 
			\begin{array}{c|cccccc}
				\hline\hline
				\mbox{Protocol}
				& \text{This paper}
				& \text{Ref.~\cite{ZhuEVQPSlong19} }
				& \text{Ref.~\cite{ZH3} }
				& \text{Ref.~\cite{Fujii17} }
				& \text{Ref.~\cite{HayaM15} }
				& \text{Refs.~\cite{TMMMF19,Xu21,Takeuchi18,MTH17}}
				\\[0.5ex]
				\hline
				\text{Is the scaling optimal in $n$?}
				&\text{Yes}   &\text{Yes}  &\text{No} &\text{Yes} &\text{Yes} &\text{No}  \\[0.5ex]
				\hline
				\text{Is the scaling optimal in $\epsilon$?}
				&\text{Yes}   &\text{Yes}  &\text{Yes} &\text{Yes} &\text{Yes} &\text{The choice of $\epsilon$ is restricted}   \\[0.5ex]
				\hline
				\text{Is the scaling optimal in $\delta$?}
				&\text{Yes}   &\text{Yes}  &\text{Yes} &\text{No}  &\text{No}  &\text{The choice of $\delta$ is restricted}   \\[0.5ex]
				\hline
				\text{Is it robust to noise?}
				& \text{Yes}  &\text{No}   &\text{No} &\text{Yes*} &\text{No}  &\text{No}   \\[0.5ex]
				\hline\hline
			\end{array}	
		\end{math}
	\end{table*}
	
	In this work, we  focus on the problem of verifying the resource graph states in the following adversarial scenario \cite{ZhuEVQPSshort19,HayaM15,ZhuEVQPSlong19}, which is crucial to the verification of blind MBQC in the receive-and-measure setting \cite{Morimae13,Fujii17,HayaM15,TMMMF19,Xu21,MTH17,HayaHajdu18}: Alice is a client (verifier) who can only perform single-qudit projective measurements with a trusted measurement device, and Bob is a server (prover) who can prepare arbitrary quantum states. 
	In order to perform MBQC, Alice delegates the preparation of the $n$-qudit graph state $|G\>\in \caH$ to Bob, who then prepares a quantum state $\rho$ on the whole space $\caH^{\otimes (N+1)}$ and sends it to Alice qudit by qudit. If Bob is honest, then he is supposed to prepare $N+1$ copies of $|G\>$; while if he is malicious, then he can mess up the computation of Alice by generating an arbitrary correlated or even entangled state $\rho$. To obtain reliable computation results, Alice needs to verify the resource state prepared by Bob with suitable tests on $N$ systems, where each test is a binary measurement on a single-copy system. If the test results satisfy certain conditions, then the conditional reduced state on  the remaining system is close to the target state $|G\>$ and can be used  for MBQC; otherwise, the state is  rejected. Since there is no  communication from Alice to Bob after the preparation of the state $\rho$, the information-theoretic blindness is guaranteed by the no-signaling principle \cite{Morimae13}.

	The assumption that the client can perform reliable local projective measurements can be justified as follows. First, the measurement devices are controlled by Alice in her laboratory and are not affected by the adversary. So it is reasonable to assume that the measurement devices are trustworthy. Second, in practice, Alice can calibrate and verify her measurement devices before performing  blind MBQC, and the resource costs of these operations are independent of the complexity of the quantum computation and the qudit number of the resource graph state. If high quality measurements can be certified after  calibration and verification, then Alice can safely use them to verify the graph state and perform blind MBQC.

	As pointed out above, the verification of the resource graph state in the adversarial scenario \cite{ZhuEVQPSshort19,HayaM15,ZhuEVQPSlong19} is a crucial and challenging part in the verification of blind MBQC. A valid verification protocol in the adversarial scenario has to meet the basic requirements of completeness and soundness \cite{TMMMF19,HayaM15,ZhuEVQPSlong19}. The completeness means Alice does not reject the ideal graph state $|G\>$. 
	Intuitively, the verification protocol is sound if Alice does not mistakenly accept any bad state that is far from the ideal state $|G\>$. 
	Concretely, the soundness means the following: once accepting, Alice needs to ensure with a high confidence level $1-\delta$ 
	that the reduced state for MBQC has a sufficiently high fidelity (at least $1-\epsilon$) with $|G\>$. 
	Here $0<\delta\leq1$ is called the significance level and the threshold $0<\epsilon<1$ is called the target infidelity. 
	The two parameters specify the target verification precision. The efficiency  of a protocol is characterized by the number $N$ of tests needed to achieve a given precision. Under the requirements of completeness and soundness, the optimal scaling behaviors of $N$ with respect to $\epsilon$, $\delta$, and the qudit number $n$ of $|G\>$ are $O(\epsilon^{-1})$, $O(\ln\delta^{-1})$, and $O(1)$, respectively, as explained in the Results section.
	However, it is highly nontrivial  to construct efficient verification protocols in the adversarial scenario. 
	Although various  protocols have been proposed \cite{ZhuEVQPSlong19,ZH3,HayaM15,TMMMF19,Xu21,Takeuchi18,MTH17}, most protocols known so far
	are too resource consuming. Even without considering noise robustness,  only the protocol of Refs.~\cite{ZhuEVQPSshort19,ZhuEVQPSlong19} achieves the optimal scaling behaviors with $n$, $\epsilon$, and $\delta$ (see Table~\ref{tab:compare}).

	Moreover, most protocols are not robust to experimental noise: the state prepared by Bob may be rejected with a high probability even if it has a  small deviation  from the ideal resource state.   
	However, in practice, it is extremely difficult to prepare quantum states with genuine multipartite entanglement perfectly. 
	So it is unrealistic to ask honest Bob to generate the perfect resource state. On the other hand, if the  deviation  from the ideal state is small enough, then it is still useful for MBQC \cite{TMMMF19,Takeuchi18}. 
	Therefore, a practical and robust protocol should accept nearly ideal states with a sufficiently high probability; otherwise, Alice needs to repeat the verification protocol many times to perform MBQC, which substantially increase the sample complexity. 
	Unfortunately, no protocol known in the literature can achieve this goal.

	Recently, a fault-tolerant protocol was proposed for verifying MBQC based on two-colorable graph states \cite{Fujii17}. 
	With this protocol, Alice can detect whether or not the given state belongs to a set of error-correctable states; then she can perform fault-tolerant MBQC on the accepted state. Although this protocol is noise-resilient to some extent, it is not very efficient  (see Table~\ref{tab:compare}),  
	and is difficult to realize in the current era of NISQ devices \cite{Presk18,Arute19,Zhong20} because too many physical qubits are required to encode the  logical qubits. In addition, this protocol is robust only to certain correctable errors since it is based on a given  error correcting code. 
	If the actual error is not correctable, then the probability of acceptance will decrease exponentially with the number of tests, which substantially increases the actual sample complexity.

	In this work, we propose a robust and efficient protocol for verifying general qudit graph states with a prime local dimension in the adversarial scenario, which plays a crucial role in robust and efficient verification of blind MBQC. 
	Our protocol is appealing to practical applications because it only requires stabilizer tests based on local Pauli measurements, which are easy to implement with current technologies. It is robust against arbitrary types of noise in state preparation, as long as the fidelity is sufficiently high. 
	Moreover, our protocol can achieve optimal scaling behaviors with respect to the system size and target precision $\epsilon,\delta$, and the sample cost is comparable to the counterpart in the nonadversarial scenario as clarified in the Methods section. 
	As far as we know, such a high efficiency  has never been achieved before when robustness is taken into account.
	In addition to qudit graph states, our protocol can also be applied to verifying many other important quantum states in the adversarial scenario, as explained in the Discussion section. Furthermore, many technical results developed in the course of our work are also useful to studying random sampling without replacement, as discussed in the companion paper \cite{Classical22} (cf. the Methods section).

	\bigskip
	\begin{center}
		\textbf{RESULTS}
	\end{center}
	\vspace{1em}
	\noindent\textbf{Qudit graph states}\\
	To establish our results, first we review the definition of qudit graph states as a preliminary, where the local dimension $d$ is a prime. 
	Mathematically, a graph $G=(V,E,m_E)$ is characterized by a set of $n$ vertices  $V=\{1,2,\dots,n\}$ and a set of edges $E$
	together with multiplicities specified by $m_E=(m_e)_{e\in E}$, where $m_e\in \bbZ_d$ and $\bbZ_d$ is the ring of integers modulo $d$,
	which is also a field given that $d$ is a prime. Two distinct vertices $i,j$ of $G$ are adjacent if they are connected by an edge. The generalized Pauli operators $X$ and $Z$ for a qudit read
	\begin{equation}
		Z|j\>=\omega^j|j\>,\quad\  X|j\>=|j+1\>,\quad\  \omega=\rme^{2\pi\rmi/d},
	\end{equation}
	where $j\in \bbZ_d$.

	Given a graph $G=(V,E,m_E)$ with $n$ vertices, we can construct an $n$-qudit graph state $|G\>\in\caH$ as follows \cite{ZH3,Keet10}: 
	first, prepare the state $|+\>:=\sum_{j\in \bbZ_d}\!|j\>/\sqrt{d}$ for each vertex; then, for each edge $e\in E$,  apply $m_{e}$ times the generalized controlled-$Z$ operation $\mathrm{CZ}_{e}$ on the vertices of $e$, where $\mathrm{CZ}_{e}=\sum_{k\in \bbZ_d}|k\>\<k |_i \otimes Z_j^k$ if $e=(i,j)$. 
	The resulting graph state has the form 
	\begin{equation}\label{eq:defGstate}
		|G\>=\left( \prod_{e\in E} \mathrm{CZ}_e ^{m_e} \right)  |+\>^{\otimes n}. 
	\end{equation}
	This graph state is also uniquely determined by its stabilizer group $S$ generated by the $n$ commuting operators $S_i:=X_i \bigotimes_{j\in V_i} Z_j^{m_{(i,j)}}$ for $i=1,2,\ldots ,n$, where  $V_i$ is the set of vertices adjacent to  vertex $i$.
	Each stabilizer operator  in $S$ can be written as
	\begin{equation}
		g_{\mathbf{k}}=\prod_{i=1}^{n} S_i^{k_i}=\bigotimes_{i=1}^n (g_{\mathbf{k}})_i,
	\end{equation} 
	where $\mathbf{k}:=(k_{1},\ldots, k_{n}) \in \mathbb{Z}_{d}^{n}$, and $(g_{\mathbf{k}})_i$ denotes the local generalized Pauli operator for the $i$th qudit.

	\bigskip
	\noindent\textbf{Strategy for testing qudit graph states}\\
	Recently, a homogeneous strategy \cite{ZhuEVQPSshort19,ZhuEVQPSlong19} for testing qubit stabilizer states based on stabilizer tests was proposed in Ref.~\cite{PLM18} and generalized to the qudit case with a prime local dimension in Sec.~X~E of Ref.~\cite{ZhuEVQPSlong19}. Here we use a variant  strategy for testing qudit graph states, which serves as an important subroutine of our verification protocol.
	Let $S$ be the stabilizer group of $|G\>\in\caH$ and $\mathcal{D}(\caH)$ be the set of all density operators on $\caH$. 
	For any operator $g_{\mathbf{k}}\in S$, the corresponding stabilizer test  is constructed as follows: party $i$ measures the local generalized Pauli operator $(g_{\mathbf{k}})_i$ for $i=1,2,\dots,n$, and records the outcome by an integer $o_{i} \in \mathbb{Z}_{d}$, which corresponds to the eigenvalue $\omega^{o_{i}}$ of $(g_{\mathbf{k}})_i$; then the test is passed if and only if the outcomes satisfy $\sum_{i} o_{i}=0 \bmod d$.
	By construction, the test can be represented by a two-outcome measurement $\{P_{\mathbf{k}},\openone-\!P_{\mathbf{k}}\}$. 
	Here $\openone$ is the identity operator on $\caH$; 
	\begin{equation}\label{eq:DefPk}
		P_{\mathbf{k}}= \frac{1}{d}\sum_{j\in \bbZ_d}g_{\mathbf{k}}^j
	\end{equation}
	is the projector onto the eigenspace of $g_{\mathbf{k}}$ with eigenvalue~1 and  corresponds to passing the test, while $\openone-P_{\mathbf{k}}$ corresponds to the failure. It is easy to check that $P_{\mathbf{k}}|G\>=|G\>$, which means $|G\>$ can always pass the test. The stabilizer test corresponding to the operator $\openone\in S$ is called the `trivial test' since all states can pass the test with certainty.
	
	To construct a verification strategy for $|G\>$, we  perform all distinct tests $P_{\mathbf{k}}$ for $\mathbf{k}\in \bbZ_d^n$  randomly each with probability $d^{-n}$. The resulting strategy is characterized by a two-outcome measurement $\{\tilde\Omega,\openone-\tilde\Omega\}$, which is determined by the verification operator 
	\begin{align}\label{eq:PLMstrategy}
		\tilde\Omega &=\frac{1}{d^{n}} \sum_{\mathbf{k} \in \mathbb{Z}_{d}^{n}} P_{\mathbf{k}}
		=|G\>\< G|+\frac{1}{d}(\openone-|G\>\<G|). 
	\end{align}
	For $1/d\leq \lambda<1$, if one performs  $\tilde\Omega$ and the trivial test with probabilities $p=\frac{d(1-\lambda)}{d-1}$ and $1-p$, respectively, then another strategy can be constructed as  \cite{ZhuEVQPSshort19,ZhuEVQPSlong19}
	\begin{align}\label{eq:strategy}
		\Omega
		=p\,\tilde\Omega+(1-p)\openone
		=|G\>\< G|+\lambda(\openone-|G\>\<G|). 
	\end{align}
	We denote by $\nu:=1-\lambda$ the spectral gap of $\Omega$ from the largest eigenvalue. 
	This strategy plays a key role in our verification protocol introduced in the next subsection. 
	
	As shown in Supplementary Note~6~A, the second equality in \eref{eq:PLMstrategy} holds whenever $d$ is a prime, but may fail if $d$ is not a prime. In the latter case, our strategy is no longer homogeneous in general, and many results in this work may not hold since they are based on homogeneous strategies. This is why we restrict our attention to the case of prime local dimensions.

	\bigskip
	\noindent\textbf{Verification of graph states in blind MBQC}\\
	Suppose Alice intends to perform quantum computation with single-qudit projective measurements on the $n$-qudit graph state $|G\>$ generated by Bob.  
	As shown in Fig.~\ref{fig:FigQSVadv}, our protocol for verifying $|G\>$ in the adversarial scenario runs as follows.

	\begin{figure}
		\begin{center}
			\includegraphics[width=6.5cm]{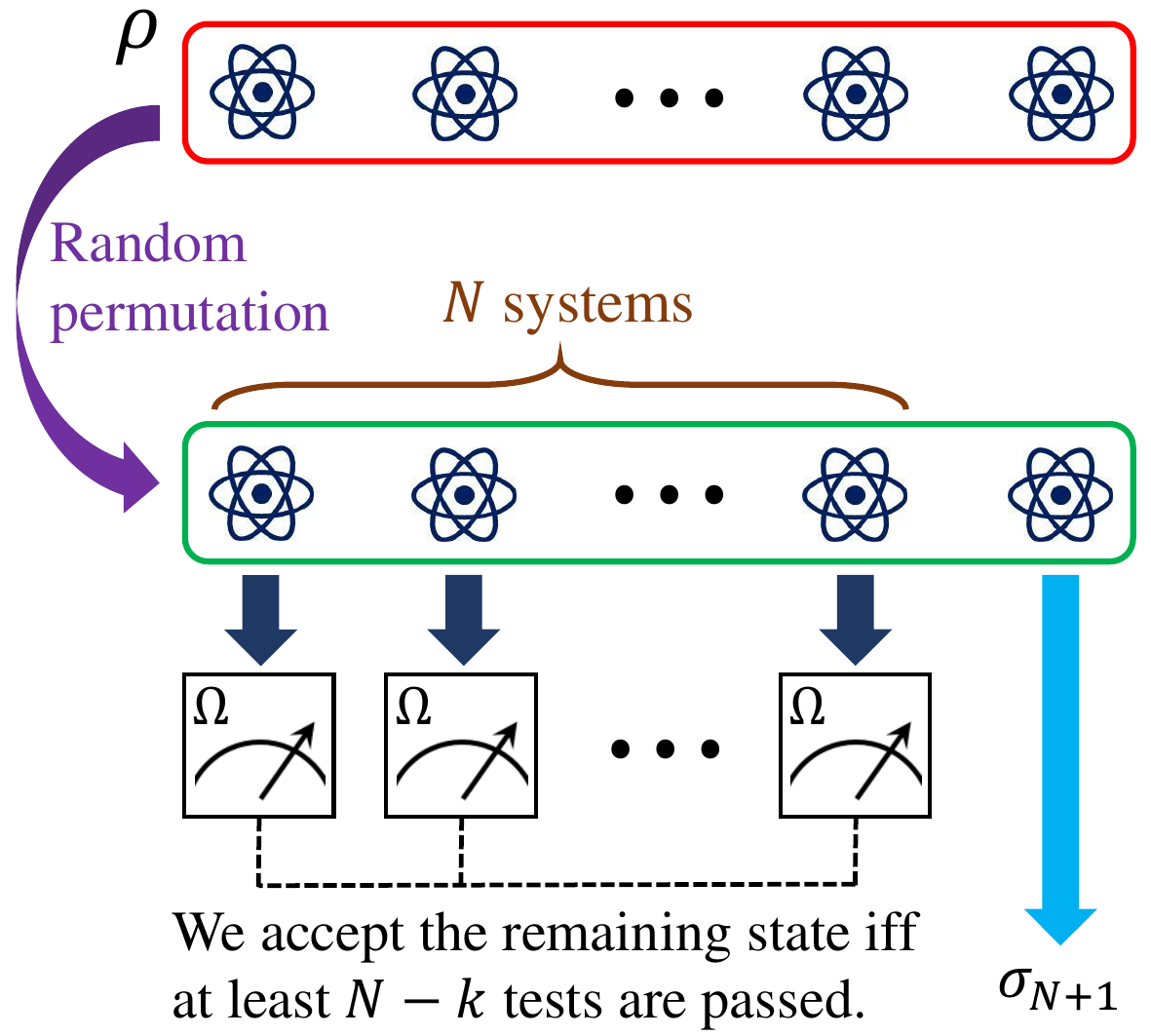}
			\caption{\label{fig:FigQSVadv}
				Schematic view of our verification protocol. Here the state $\rho$ generated by Bob might be arbitrarily correlated or entangled on the whole space $\caH^{\otimes (N+1)}$. To verify the target state,  Alice first randomly permutes all $N+1$ systems, and then uses a strategy $\Omega$ to test each of the first $N$ systems. Finally, she accepts the reduced state $\sigma_{N+1}$ on the remaining system iff at least $N-k$ tests are passed.
			}
		\end{center}
	\end{figure}

	\begin{enumerate}
		\item[1.] Bob produces a state $\rho$ on the whole space $\caH^{\otimes (N+1)}$ with $N\geq1$ and sends it to Alice. 
		
		\item[2.] After receiving the state, Alice randomly permutes the $N+1$ systems of $\rho$ 	(due to this procedure, we can assume that $\rho$ is permutation invariant without loss of generality) and applies the strategy $\Omega$ defined in \eref{eq:strategy} to the first $N$ systems.
		
		\item[3.] Alice chooses an integer $0\leq k\leq N-1$, called the number of allowed failures.
		If at most $k$ failures are observed among the $N$ tests, Alice accepts the reduced state $\sigma_{N+1}$ on the remaining system and uses it for MBQC; otherwise, she rejects it.
	\end{enumerate}

	With this verification protocol, Alice aims to achieve three goals: completeness, soundness, and robustness. 
	Recall that $|G\>$ can always pass each test, so the completeness is automatically  guaranteed. 
	The soundness is characterized by the target infidelity $\epsilon$ and significance level $\delta$ as explained in the introduction.  
	For verification protocols working in the nonadversarial scenario, where the source only produces independent states with no correlation or entanglement among different runs, the optimal scaling behaviors of the test number $N$ with respect to $\epsilon$, $\delta$, and $n$ are $O(\epsilon^{-1})$, $O(\ln\delta^{-1})$, and $O(1)$, respectively \cite{PLM18,ZhuEVQPSlong19}. The adversarial scenario studied in this work has a weaker assumption on the source \cite{PLM18,ZhuEVQPSlong19}, so the scaling behaviors in $\epsilon$, $\delta$, and $n$ cannot be better. Although the condition of soundness looks quite simple, it is highly nontrivial to determine the degree of soundness.
	Even in the special case $k=0$, this problem was resolved only very recently after quite a lengthy analysis \cite{ZhuEVQPSshort19,ZhuEVQPSlong19}. Unfortunately, the robustness of this protocol is poor in this special case, as we shall see later. 
	So we need to tackle this challenge in the general case.

	Most previous works did not consider the problem of robustness at all, because it is already very difficult to detect the bad case without considering robustness. To characterize the robustness of a protocol, we need to consider the case in which honest Bob prepares an independent and identically distributed (i.i.d.) quantum state, that is, $\rho$ is a tensor power of the form $\rho=\tau^{\otimes (N+1)}$ with $\tau\in \mathcal{D}(\caH)$. 
	Due to inevitable noise, $\tau$ may not equal the ideal state $|G\>\<G|$.
	Nevertheless, if the infidelity $\epsilon_\tau:=1-\<G|\tau|G\>$ is smaller than the target infidelity, that is, $\epsilon_\tau<\epsilon$, then  $\tau$ is still useful for quantum computing. For a robust verification protocol, such a state should be accepted with a high probability.

	In the i.i.d.\! case, the probability that Alice accepts $\tau$ reads 
	\begin{align}\label{eq:SuccessProbpk}
		p^{\iid}_{N,k}(\tau)
		=B_{N,k}\big(1-\tr(\Omega \tau) \big)
		=B_{N,k}(\nu\epsilon_\tau),
	\end{align}
	where $N$ is the number of tests, $k$ is the number of allowed failures, and $B_{N,k}(p)\!:=\!\sum_{j=0}^k \binom{N}{j} p^j (1-p)^{N-j}$ is the 
	binomial cumulative distribution function. To construct a robust verification protocol, it is preferable to choose a large value of $k$, so that $p^{\iid}_{N,k}(\tau)$ is sufficiently high. Unfortunately, most previous verification protocols can reach a meaningful conclusion only when  $k=0$ \cite{ZhuEVQPSshort19,ZhuEVQPSlong19,ZH3,HayaM15,MTH17}, in which case the probability 
	\begin{align}\label{eq:pracck=0}
		p^{\iid}_{N,k=0}(\tau)
		=(1-\nu\epsilon_\tau)^N
	\end{align}
	decreases exponentially  with the test number $N$, which is not satisfactory. These protocols need a large number of tests to guarantee the soundness, so it is difficult to get accepted even if Bob is honest. Hence,  previous protocols with the choice $k = 0$ are not robust to noise in state preparation. Since the acceptance probability is small, Alice needs to repeat the verification protocol many times to ensure that she accepts the state $\tau$  at least once, which substantially increases the actual sample cost.

	\begin{figure}
		\begin{center}
			\includegraphics[width=6.9cm]{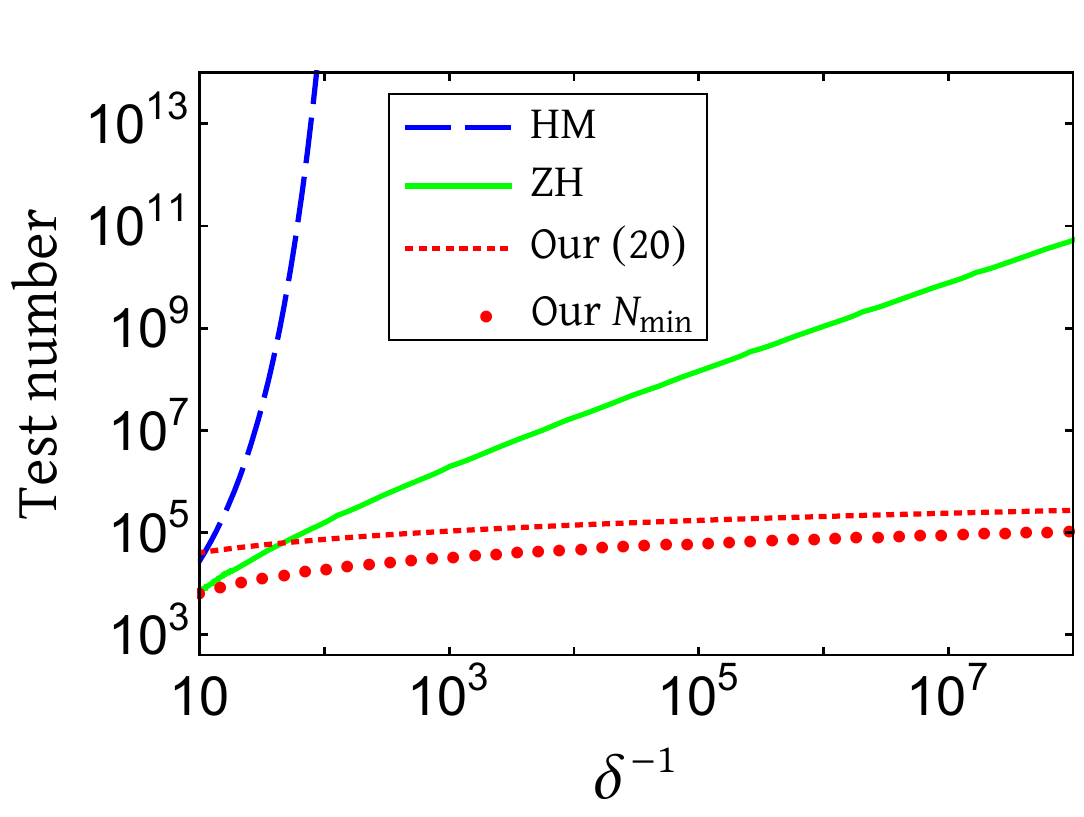}
			\caption{\label{fig:TotalNumTest}
				Number of tests required to verify a general qudit graph state in the adversarial scenario within infidelity $\epsilon=0.01$, significance level $\delta$, and robustness $r=1/2$.  The red dots correspond to $N_{\rm min}(\epsilon,\delta,\lambda,r)$ in \eref{eq:OptNadvDef} with $\lambda=1/2$,  and 
				the red dashed curve corresponds to the RHS of \eref{eq:ngeqiidHighProb}, which is an upper bound for $N_{\rm min}(\epsilon,\delta,\lambda,r)$.  	The blue dashed curve corresponds to the HM protocol \cite{HayaM15}, and the green solid curve corresponds to the ZH protocol \cite{ZhuEVQPSlong19} with $\lambda=1/2$. The performances of the TMMMF protocol \cite{TMMMF19} and  TM protocol \cite{Takeuchi18} are not shown  because the numbers of tests required are too large (see Supplementary Note~3). 
			}
		\end{center}
	\end{figure}

	When  $\epsilon_\tau= \frac{\epsilon}{2}$ for example, the number of repetitions required is at least $\Theta\big(\exp[\frac{1}{4\delta}]\big)$ for the HM protocol in  Ref.~\cite{HayaM15} and  $\Theta\big(\frac{1}{\sqrt{\delta}}\big)$  for the ZH protocol in Refs.~\cite{ZhuEVQPSshort19,ZhuEVQPSlong19} 
	(see Supplementary Note~3 for details). As a consequence, the total number of required tests is at least $\Theta\big(\frac{1}{\delta}\exp[\frac{1}{4\delta}]\big)$ for the HM protocol and $\Theta\big( \frac{\ln\delta^{-1}}{\sqrt{\delta}}\big)$ for the ZH protocol, as illustrated in Fig.~\ref{fig:TotalNumTest}. Therefore, although some protocols known in the literature are reasonably efficient in detecting the bad case, they are not useful in verifying the resource state of blind MBQC in a realistic scenario.

	\bigskip
	\noindent\textbf{Guaranteed infidelity}\\	
	Suppose $\rho$ is  permutation invariant. Then the probability that Alice accepts $\rho$ reads
	\begin{align}\label{eq:pnmkrho}
		p_k(\rho)=
		\sum_{i=0}^k \binom{N}{i}
		\tr \! \big(\big[\Omega^{\otimes (N-i)}\otimes \overline{\Omega}^{\otimes i} \otimes \openone\big] \rho\big),
	\end{align}
	where $\overline{\Omega}:=\openone-\Omega$. 
	Denote by $\sigma_{N+1}$ the reduced state on the remaining system when at most $k$ failures are observed. 
	The fidelity between $\sigma_{N+1}\!$ and the ideal state $|G\>$ reads $F_k(\rho) = f_k(\rho)/p_k(\rho)$ $[{\rm assuming}\  p_k(\rho)>0]$, where
	\begin{align}\label{eq:fnmkrho}
		\begin{split}
			f_k(\rho) =  
			\sum_{i=0}^k \binom{N}{i} \tr \! \big(\big[
			\Omega^{\otimes (N-i)} \otimes \overline{\Omega}^{\otimes i} 
			\otimes |G\>\<G|\big] \rho \big).
		\end{split}
	\end{align}
	The actual verification precision can be characterized by the following figure of merit with $0<\delta\leq1$,
	\begin{align}
		\bar{\epsilon}_{\lambda}(k,N,\delta):=
		1-\min_{\rho} \left\{F_k(\rho) \,|\, p_k(\rho)\geq \delta  \right\},
		\label{eq:hatzeta2}
	\end{align}
	where $\lambda$ is determined by \eref{eq:strategy}, and the minimization is taken over permutation-invariant states $\rho$ on ${\cal H}^{\otimes (N+1)}$.

	If Alice accepts the state prepared by Bob, then she can guarantee (with significance level $\delta$) that the reduced state $\sigma_{N+1}$ has
	infidelity at most $\bar{\epsilon}_{\lambda}(k,N,\delta)$ with the ideal state $|G\>$. 
	Consequently, according to the relation between the fidelity and trace norm, Alice can ensure  the condition \cite{HayaM15}
	\begin{align}\label{eq:relaFidTr}
		\left|\operatorname{tr}(E \sigma_{N+1})-\langle G|E| G\rangle\right| \leq \sqrt{\bar{\epsilon}_\lambda(k, N, \delta)}
	\end{align}
	for any POVM element $0\leq E\leq \openone$; that is, the deviation of any measurement outcome probability from the ideal value is not larger than $\sqrt{\bar{\epsilon}_{\lambda}(k,N,\delta)}$.

	In view of  the above discussions, the computation of $\bar{\epsilon}_{\lambda}(k,N,\delta)$ given in \eref{eq:hatzeta2} is of central importance 
	to analyzing the soundness of our protocol. 
	Thanks to the analysis presented in the Methods section, this quantum optimization problem can actually be reduced to a 
	classical sampling problem studied in the companion paper \cite{Classical22}. 
	Using the results derived in Ref.~\cite{Classical22}, we can deduce many useful properties of $\bar{\epsilon}_{\lambda}(k,N,\delta)$ 
	as well as its analytical formula, which are presented in Supplementary Note~1. Here it suffices to clarify the monotonicity properties of  $\overline{\epsilon}_\lambda(k,N,\delta)$ as stated  in Proposition~\ref{prop:epsMonoton} below, which follows from Proposition~6.5 in Ref.~\cite{Classical22}. Let $\bbZ^{\geq j}$ be the set of integers larger than or equal to $j$.
	
	\begin{proposition}\label{prop:epsMonoton}	Suppose $0\leq\lambda<1$, $0<\delta\leq1$, $k\in\bbZ^{\geq 0}$, and $N\in \bbZ^{\geq k+1}$. 
		Then $\overline{\epsilon}_\lambda(k,N,\delta)$ is nonincreasing in $\delta$ and $N$, but nondecreasing in $k$.
	\end{proposition}

	\bigskip
	\noindent\textbf{Verification with a fixed error rate}\\
	If  the number $k$ of allowed failures is sublinear in $N$, that is, $k=o(N)$, then the acceptance probability $p^{\iid}_{N,k}(\tau)$  in \eref{eq:SuccessProbpk} for the i.i.d.\! case approaches 0 as the number of tests $N$ increases, which is not satisfactory.
	To achieve robust verification, here we set the number $k$ to be proportional to the number of tests, that is,
	$k=\lfloor s\nu N \rfloor$, where $0\leq s<1$ is the error rate, and $\nu=1-\lambda$ is the spectral gap of the strategy $\Omega$. 
	In this case, when Bob prepares i.i.d.\! states $\tau\in\mathcal{D}(\caH)$ with $\epsilon_\tau<s$, 
	the acceptance probability $p^{\iid}_{N,k}(\tau)$ approaches one as $N$ increases. 
	In addition, we can deduce the following theorem, which is proved in Supplementary Note~6~B. 
	
	\begin{figure}
		\begin{center}
			\includegraphics[width=6.7cm]{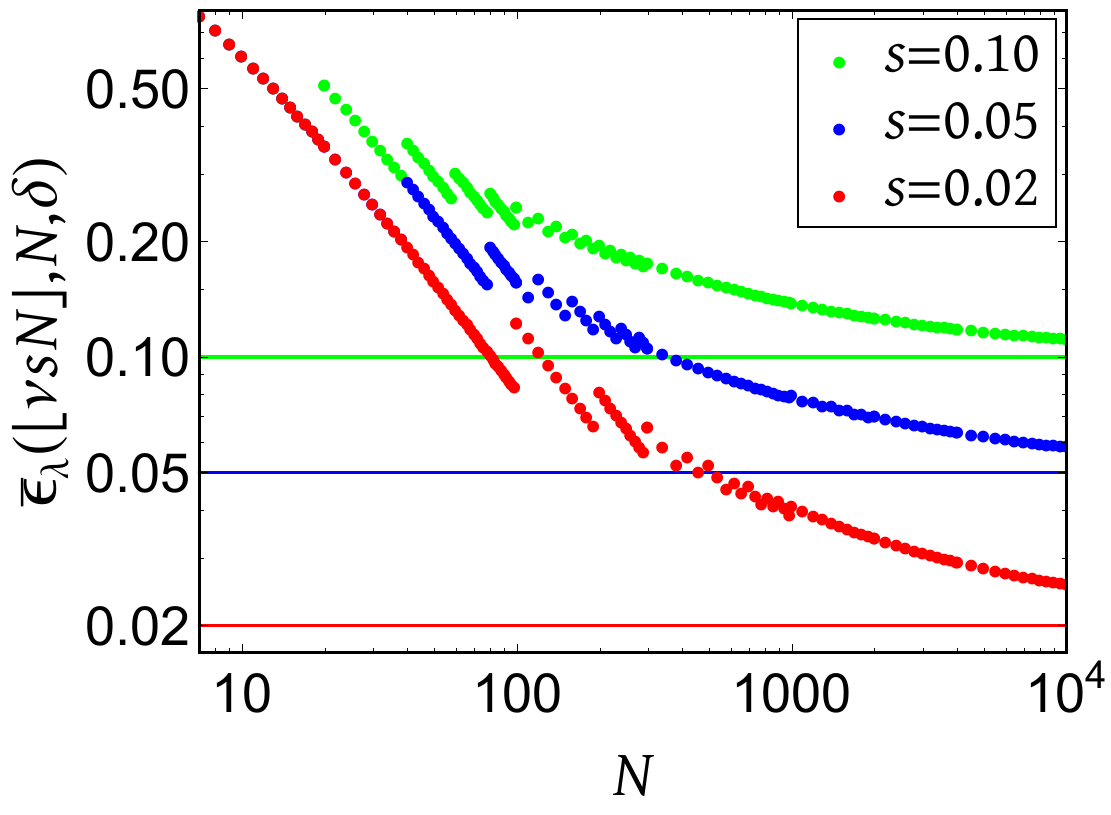}
			\caption{\label{fig:Logplot_barEps}
				Variations of $\bar{\epsilon}_{\lambda}(\lfloor\nu s N\rfloor, N, \delta)$ with the number $N$ of tests and error rate $s$ [by Eq.~(6) in Supplementary Note~1]. Here $\lambda=1/2$ and significance level $\delta=0.05$.  Each horizontal line represents an error rate. As the test number $N$ increases, $\bar{\epsilon}_{\lambda}(\lfloor\nu s N\rfloor, N, \delta)$ approaches~$s$.
			}
		\end{center}
	\end{figure}

	\begin{theorem}\label{thm:Boundeps}
		Suppose $0<s,\lambda<1$, $0<\delta\leq1/4$, and $N\in\bbZ^{\geq 1}$. Then
		\begin{align}\label{eq:infidelityLB}
			s- \frac{1}{\nu N}
			&<\bar{\epsilon}_{\lambda}(\lfloor\nu s N\rfloor, N, \delta) \nonumber\\
			&\leq s+ \frac{1}{\nu \lambda}
			\sqrt{\frac{s \ln\delta^{-1}}{N}} + \frac{\ln\delta^{-1}}{2\nu^2\lambda N}+\frac{2}{\lambda N}.
		\end{align}
	\end{theorem}
	Theorem~\ref{thm:Boundeps} implies that  $\bar{\epsilon}_{\lambda}(\lfloor\nu s N\rfloor, N, \delta)$ converges to the error rate 
	$s$ when the number $N$ of tests  gets large, as illustrated in Fig.~\ref{fig:Logplot_barEps}. 
	To achieve a given infidelity $\epsilon$ and significance level $\delta$, which means $\bar{\epsilon}_{\lambda}(\lfloor\nu s N\rfloor,N,\delta)\leq\epsilon$, it suffices to  set $s< \epsilon$ and choose a sufficiently  large $N$. By virtue of Theorem~\ref{thm:Boundeps} we can derive the  following theorem as proved in Supplementary Note~6~C.

	\begin{theorem}\label{thm:UBtestsNumber}
		Suppose  $0<\lambda<1$, $0\leq s<\epsilon<1$, and $0<\delta\leq1/2$.
		If the number  $N$ of tests satisfies
		\begin{align}
			N\geq
			\frac{\epsilon}{\left[ \lambda\nu (\epsilon-s) \right]^2}
			\left( \ln\delta^{-1} +4\lambda\nu^2 \right) ,
		\end{align}
		then  $\bar{\epsilon}_{\lambda}(\lfloor\nu s N\rfloor,N, \delta)\leq\epsilon$.
	\end{theorem}
	Notably, if the ratio $s/\epsilon$ is a constant, then the sample cost is only $O(\epsilon^{-1}\ln \delta^{-1})$. 
	The scaling behaviors in $\epsilon$ and $\delta$ are the same as the counterparts in the nonadversarial scenario, and are thus optimal.

	\bigskip
	\noindent\textbf{The number of allowed failures}\\
	Next, we consider the case in which the number $N$ of tests is given. 
	To  construct a concrete verification protocol, we need to specify the number $k$ of allowed failures such that the conditions of soundness and robustness are satisfied simultaneously. According to Proposition~\ref{prop:epsMonoton}, a small $k$ is preferred to guarantee soundness, while a larger $k$ is preferred to guarantee robustness. To construct a robust and efficient verification protocol, we need to find a good balance between the two conflicting requirements. The following proposition provides a suitable interval for the number $k$ of allowed failures that can guarantee soundness; see Supplementary Note~6~E for a proof.  
	
	\begin{proposition}\label{prop:kbounds}
		Suppose $0<\lambda,\epsilon <1$, $0<\delta\leq1/4$, and $N,k\in\bbZ^{\geq 0}$. 
		If $\nu \epsilon N\leq k\leq N-1$, then $\bar{\epsilon}_{\lambda}(k, N, \delta)>\epsilon$.
		If $k\leq l(\lambda, N, \epsilon,\delta)$, then $\bar{\epsilon}_{\lambda}(k, N, \delta)\leq\epsilon$. Here
		\begin{align}\label{eq:definel}
			l(\lambda, N, \epsilon,\delta):=\!
			\bigg\lfloor\nu \epsilon N -\! \frac{\sqrt{  N \epsilon \ln\delta^{-1}}}{\lambda}
			-\! \frac{\ln\delta^{-1}}{2\lambda\nu}-\! \frac{2\nu}{\lambda} \bigg\rfloor.\!
		\end{align}
	\end{proposition}

	\begin{figure}
		\begin{center}
			\includegraphics[width=6.5cm]{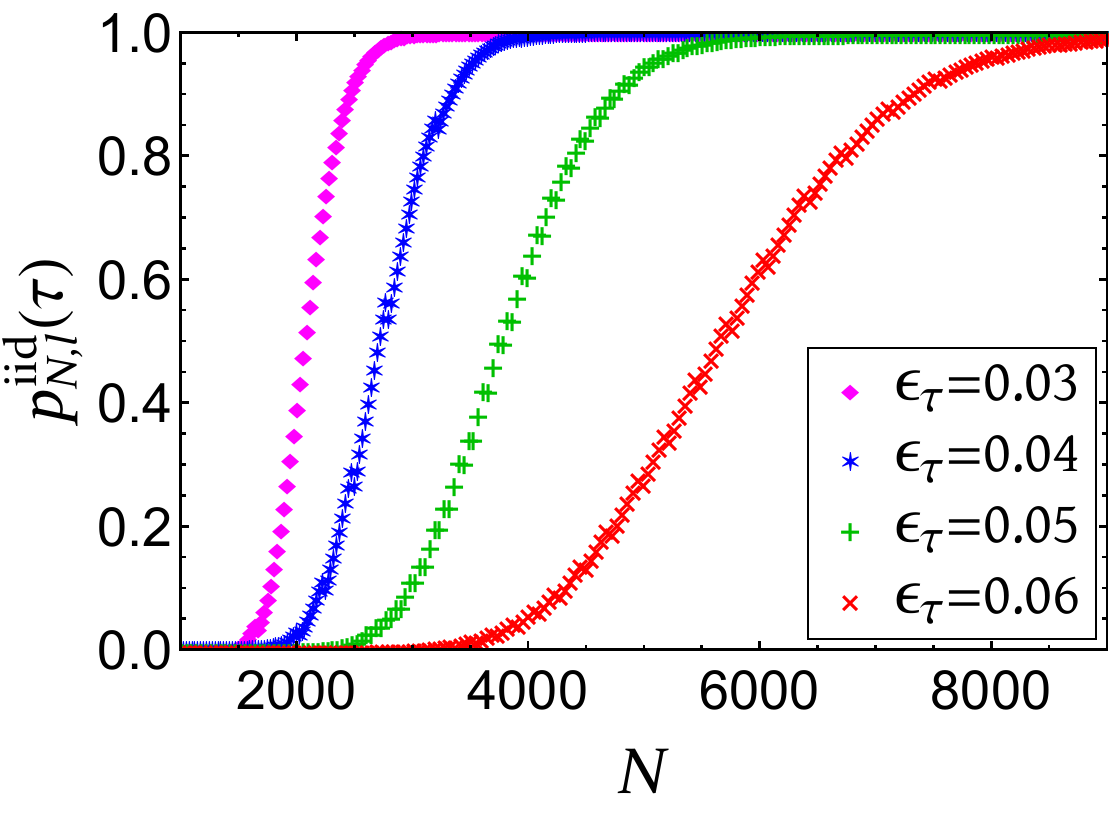}
			\caption{\label{fig:plot_pNk}
				The probability $p^{\iid}_{N,l(\lambda, N, \epsilon,\delta)}(\tau)$ that Alice accepts i.i.d.\! quantum states $\tau\in \mathcal{D}(\caH)$. 
				Here $\lambda=1/2$, infidelity $\epsilon=0.1$, and significance level $\delta=0.01$; $\epsilon_\tau$ is the infidelity between $\tau$ and the target state $|G\>$; and  $l(\lambda, N, \epsilon,\delta)$ is  the number of allowed failures defined in \eref{eq:definel}.
			}
		\end{center}
	\end{figure}
	
	Next, we turn to the condition of robustness.
	When honest Bob prepares i.i.d.\! quantum states $\tau\in \mathcal{D}(\caH)$ with infidelity $0<\epsilon_\tau<\epsilon$, the probability that Alice accepts $\tau$ is $p^{\iid}_{N,k}(\tau)$ given in \eref{eq:SuccessProbpk}, which is strictly increasing in $k$ according to Lemma~S4 in Supplementary Note~2. Suppose we set $k=l(\lambda, N, \epsilon,\delta)$. As the number of tests $N$ increases, the acceptance probability has the following asymptotic behavior if $0<\epsilon_\tau<\epsilon$ (see Supplementary Note~6~F for a proof),
	\begin{align}\label{eq:pkLimit}
		p^{\iid}_{N,l}(\tau)=
		1-\exp\bigl[ - D( \nu\epsilon \| \nu\epsilon_\tau ) N +O(\sqrt{N}\,) \bigr],
	\end{align}
	where $D(p\|q):=p\ln\frac{p}{q}+(1-p)\ln\frac{1-p}{1-q}$ is the relative entropy between 
	two binary probability vectors $(p,1-p)$ and $(q,1-q)$, and $l$ is a shorthand for $l(\lambda, N, \epsilon,\delta)$. 
	Therefore, the probability of acceptance is arbitrarily close to one  as long as  $N$ is sufficiently large, as illustrated in Fig.~\ref{fig:plot_pNk}.
	Hence,  our verification protocol is able to reach any degree of robustness.

	\bigskip
	\noindent\textbf{Sample complexity of robust verification}\\	
	Now we consider the resource cost required by our protocol to reach given verification precision and robustness. 
	Let $\rho$ be the state on $\caH^{\otimes (N+1)}$ prepared by Bob and $\sigma_{N+1}\!$ be the reduced state after Alice performs suitable tests and accepts the state $\rho$. To verify the target state within infidelity $\epsilon$, significance level $\delta$, and robustness $r$ (with $0\leq r<1$) entails the following two conditions. 
	\begin{enumerate}
		\item[1.] (Soundness) If the infidelity of  $\sigma_{N+1}\!$ with the target state is larger than $\epsilon$, then
		the probability that Alice accepts $\rho$ is less than $\delta$.
		
		\item[2.] (Robustness) If $\rho=\tau^{\otimes (N+1)}$ with $\tau\in\mathcal{D}(\caH)$  and  $\epsilon_\tau\leq r\epsilon$, then
		the probability that Alice accepts $\rho$ is at least  $1-\delta$.
	\end{enumerate}
	The tensor power $\rho$ in Condition 2 can be replaced by the tensor product of $N+1$ independent quantum states $\tau_1,\tau_2,\dots,\tau_{N+1}\in\mathcal{D}(\caH)$ that have infidelities at most $r\epsilon$. All our conclusions do not change under this modification.

	To achieve the conditions of soundness and robustness, we need to choose the test number $N$ and the number $k$ of allowed failures properly. To determine the resource cost, we define  $N_{\rm min}(\epsilon,\delta,\lambda,r)$ as the minimum number of tests required for robust verification, that is, the minimum positive integer $N$ such that there exists an integer $0\leq k\leq N-1$ which together with $N$ achieves the above two conditions. 
	Note that the conditions of soundness and robustness can be expressed as 
	\begin{align}
		&\bar{\epsilon}_{\lambda}(k, N, \delta)\leq\epsilon,\qquad  B_{N,k}(\nu r\epsilon)\geq1-\delta.  \label{eq:robustCondition}
	\end{align}
	So $N_{\rm min}(\epsilon,\delta,\lambda,r)$ can be expressed as 
	\begin{align}\label{eq:OptNadvDef}
		N_{\rm min}(\epsilon,\delta,\lambda,r) &:= 
		\min_{N,k} \big\{ N \,\big|\,  k\in\bbZ^{\geq 0}, N\in \bbZ^{\geq k+1}, \nonumber\\
		&\bar{\epsilon}_{\lambda}(k, N, \delta)\leq\epsilon, B_{N,k}(\nu r\epsilon)\geq1-\delta \big\}.  
	\end{align}

	\begin{figure}
		\begin{algorithm}[H]
			{\small
				\hspace{-98pt}\textbf{Input:}  $\lambda,\epsilon,\delta\in(0,1)$ and $r\in[0,1)$.\\
				\hspace{-71pt} \textbf{Output:} $k_{\min}(\epsilon,\delta,\lambda,r)$ and $N_{\rm min}(\epsilon,\delta,\lambda,r)$.
				
				\begin{algorithmic}[1]
					\caption{{\small Minimum test number for robust verification}}
					\label{alg:NoptAdv}
					
					\If{$r=0$} 
					
					\State{$k_{\min}\leftarrow0$}
					
					\Else \For{$k=0,1,2,\dots$}
					
					\State{Find the largest integer $M$ such that $B_{M,k}(\nu r\epsilon)\geq1-\delta$.}
					
					\If{$M\geq k+1$ and $\bar{\epsilon}_{\lambda}(k, M, \delta)\leq\epsilon$,} 
					
					\State{stop}
					
					\EndIf
					
					\EndFor 
					
					\State{$k_{\min}\leftarrow k$}
					
					\EndIf
					
					\State{Find the smallest integer $N$ that satisfies $N\geq k_{\min}+1$ 
						and $\bar{\epsilon}_{\lambda}(k_{\min}, N, \delta)\leq\epsilon$.}
					
					\State{$N_{\min}\leftarrow N$ }
					
					\State{\textbf{return} $k_{\min}$ and $N_{\min}$}.
					
				\end{algorithmic}
			}
		\end{algorithm}
	\end{figure}
	
	Next, we propose a simple algorithm, Algorithm~\ref{alg:NoptAdv},  for computing $N_{\rm min}(\epsilon,\delta,\lambda,r)$, which is very useful to practical applications. In addition to $N_{\rm min}(\epsilon,\delta,\lambda,r)$, this algorithm determines the corresponding number of allowed failures, 
	which is denoted by $k_{\rm min}(\epsilon,\delta,\lambda,r)$. 
	In Supplementary Note~7~C we explain why Algorithm~\ref{alg:NoptAdv} works. Algorithm~\ref{alg:NoptAdv} is particularly useful to studying  the variations of $N_{\rm min}(\epsilon,\delta,\lambda,r)$ with the four parameters $\epsilon$,  $\delta$,  $\lambda$, $r$ as illustrated in  Fig.~\ref{fig:optN}. 
	When $\delta$ and $r$ are fixed,  $N_{\rm min}(\epsilon,\delta,\lambda,r)$ is inversely proportional to $\epsilon$;  
	when $\epsilon,r$ are fixed and $\delta$ approaches 0, $N_{\rm min}(\epsilon,\delta,\lambda,r)$ is proportional to $\ln\delta^{-1}$. 
	In addition, Fig.~\ref{fig:optN}~(d) indicates that a strategy $\Omega$ with small or large $\lambda$ is not very efficient for robust verification, while any  choice satisfying $0.3\leq\lambda\leq0.5$ is nearly optimal.

	The following theorem provides an informative upper bound for $N_{\rm min}(\epsilon,\delta,\lambda,r)$ and clarifies the sample complexity of robust verification; see Supplementary Note~6~D for a proof. 
	\begin{theorem}\label{thm:iidHighProb}
		Suppose $0<\lambda,\epsilon<1$, $0<\delta\leq1/2$, and $0\leq r<1$. Then the conditions of soundness and robustness in \eref{eq:robustCondition} hold as long as
		\begin{align}
			k&
			= \bigg\lfloor  \bigg(\frac{\lambda\sqrt{2\nu}+r}{\lambda\sqrt{2\nu}+1}\bigg)\nu \epsilon N\bigg\rfloor, \\
			N&
			\geq
			\left\lceil \bigg[ \frac{\lambda\sqrt{2\nu}+1}{\lambda \nu (1-r)} \bigg]^2 \;
			\frac{\ln\delta^{-1} +4\lambda\nu^2}{ \epsilon} \right\rceil.
			\label{eq:ngeqiidHighProb}
		\end{align}
	\end{theorem}

	\begin{figure}
		\includegraphics[width=8.63cm]{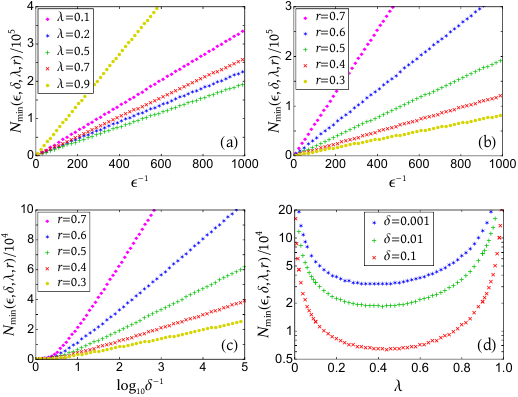}
		\caption{\label{fig:optN}
			Minimum number of tests required for robust verification (by Algorithm~\ref{alg:NoptAdv}). 
			(a) Variations of $N_{\rm min}(\epsilon,\delta,\lambda,r)$ with $\epsilon^{-1}$ and $\lambda$, where robustness $r=1/2$ and significance level $\delta=0.01$. 
			(b) Variations of $N_{\rm min}(\epsilon,\delta,\lambda,r)$ with $\epsilon^{-1}$ and $r$,  where $\delta=0.01$ and $\lambda=1/2$. 
			(c) Variations of $N_{\rm min}(\epsilon,\delta,\lambda,r)$ with $\log_{10}\delta^{-1}$ and $r$, where $\lambda=1/2$ and infidelity $\epsilon=0.01$.
			(d) Variations of $N_{\rm min}(\epsilon,\delta,\lambda,r)$ with $\lambda$ and $\delta$, where $\epsilon=0.01$ and $r=1/2$. 
		}
	\end{figure}

	For  given $\lambda$ and  $r$,  the minimum number of tests is only $O(\epsilon^{-1}\ln \delta^{-1})$, which is independent of the qudit number $n$ of $|G\>$ and achieves the optimal scaling behaviors with respect to the infidelity $\epsilon$ and significance level $\delta$. The coefficient is large when $\lambda$ is close to 0 or 1, while it is around the minimum for any value of $\lambda$ in the interval $[0.3, 0.5]$. Numerical calculation based on Algorithm~\ref{alg:NoptAdv} shows that the upper bound for $N_{\rm min}(\epsilon,\delta,\lambda,r)$ provided in Theorem~\ref{thm:iidHighProb} is  a bit conservative, especially when $r$ is small. In other words, the actual sample cost is smaller than what can be proved rigorously. Nevertheless, the bound is quite informative about the general trends. If we choose $r=\lambda=1/2$ for example, then Theorem~\ref{thm:iidHighProb} implies that 
	\begin{align}
		N_{\rm min}(\epsilon,\delta,\lambda,r)\leq\lceil 144 \,\epsilon^{-1} (\ln\delta^{-1}+0.5)\rceil,
	\end{align} 
	while numerical calculation shows that $N_{\rm min}(\epsilon,\delta,\lambda,r)\leq 67 \,\epsilon^{-1}\ln\delta^{-1}$. Compared with previous works \cite{HayaM15, ZhuEVQPSshort19, ZhuEVQPSlong19}, our protocol improves the scaling behavior with respect to the significance level $\delta$ exponentially and even doubly exponentially,  as illustrated in Fig.~\ref{fig:TotalNumTest}. 
	
	\bigskip
	\begin{center}
		\textbf{DISCUSSION}
	\end{center}
	\vspace{1em}
	Verification of resource graph states in the adversarial scenario is a crucial step in the verification of blind MBQC. 
	We have proposed a highly robust and efficient  protocol for achieving this task, which applies to any qudit graph state with a prime local dimension. To implement this protocol, it suffices to perform simple stabilizer tests based on local Pauli measurements, which is quite appealing to NISQ devices. For any given degree of robustness, to verify the target graph state within infidelity $\epsilon$ and significance level $\delta$,  only $O(\epsilon^{-1}\ln \delta^{-1})$ tests are required, which achieves the optimal sample complexity with respect to the system size, infidelity, and significance level.
	Compared with previous protocols, our protocol can reduce the sample cost dramatically in a realistic scenario; notably, 
	the scaling behavior in the significance level can be  improved exponentially.

	So far we have focused on the verification of resource graph states with trustworthy and ideal local projective measurements.
	According to \eref{eq:relaFidTr}, if the blind MBQC is performed with ideal measurements after Alice accepts the state prepared by Bob, 
	then the precision of the computation results is guaranteed by the precision of the graph state. 
	However, in practice, it is unrealistic to assume that the measurement devices are perfect. So we need additional operations to guarantee the precision of the computation results when verifying blind MBQC in the receive-and-measure setting.  As mentioned in the introduction, the client can calibrate her measurement devices before performing  blind MBQC with a small overhead. In addition, we can convert the noise in measurements to noise in  state preparation. To apply this method, we need the assumption that any  measurement used in MBQC and graph state verification can be expressed as a composition of	a measurement-independent noise process and  the noiseless measurement. 
	The detail of this conversion method is presented in Supplementary Note~4. 
	When the noise process depends on the specific measurement, the situation is more complicated, and further study is required to deal with such noise.

	After obtaining a reliable resource graph state accepted by  the verification protocol, Alice  can use it to perform MBQC. 
	In this procedure, she needs to adaptively select local projective measurements to drive the computation.  
	Nevertheless, these operations can be completed by using a classical computer, and the classical computation complexity scales linearly with the size of the original quantum computation \cite{RaussBrie01}. Therefore, the most challenging part in the verification of blind MBQC is the verification of the resource graph state, which is the focus of this work.

	In the above discussion, we assume that the measurement devices are controlled by the client and are trustworthy. 
	It is also desirable to construct robust and efficient protocols for verifying blind MBQC when the measurement devices are not trustworthy. 
	To this  end, a device-independent (DI) verification protocol was proposed in Ref.~\cite{GKW15}.
	However, this protocol has a quantum communication complexity of the order $O(\tilde{n}^c)$, where $\tilde{n}$ is the size of the delegated quantum computation and $c>2048$, which is too prohibitive for any practical implementation. 
	By combining the CHSH inequality and stabilizer tests applied to a qubit graph state, Ref.~\cite{HayaHajdu18} 
	proposed a protocol for self-testing MBQC in the receive-and-measure setting. 
	This protocol requires $O(n^4\log n)$ samples with $n$ being the qubit number of the resource graph state, which is much more efficient than previous protocols, but is still far from the optimal scaling achieved in this work. In addition, it does not consider the  problem of robustness. 
	To further reduce the overhead and improve the robustness, it might be helpful to combine our approach with DI quantum state certification (DI QSC) developed recently \cite{Aleks22}. See Supplementary Note~5 for details.

	In addition to graph states, our protocol can also be used to verify many other pure quantum states in the adversarial scenario, where the 
	state preparation is controlled by a potentially malicious adversary Bob, who can produce an arbitrary correlated or entangled state $\rho$ on the whole system $\caH^{\otimes(N+1)}$. Let $|\Psi\>\in\caH$ be the target pure state to be verified. 
	Then a verification strategy $\Omega$ for $|\Psi\>$  is called homogeneous \cite{ZhuEVQPSshort19,ZhuEVQPSlong19} if it has the form 
	\begin{align}\label{eq:homo}
		\Omega=|\Psi\>\< \Psi|+\lambda(\openone-|\Psi\>\<\Psi|),
		\qquad 0\leq\lambda<1. 
	\end{align}
	Efficient homogeneous strategies based on local projective measurements have been constructed for many important quantum states \cite{ZH4,Ha09-2,HMT,LHZ19,Wang19,LiGHZ19,Li21,PLM18,ZhuEVQPSlong19,LLSZ22}.

	If  a homogeneous strategy $\Omega$ given in \eref{eq:homo} can be constructed, then the target state $|\Psi\>$ can be verified in the adversarial scenario by virtue of our protocol:  Alice first randomly permutes all systems of $\rho$  and applies the strategy $\Omega$ to the first $N$ systems, then she accepts the remaining unmeasured system if at most $k$ failures are observed among these tests. Most results (including Theorems~\ref{thm:Boundeps}, \ref{thm:UBtestsNumber}, \ref{thm:iidHighProb}, Algorithm~\ref{alg:NoptAdv}, and Propositions~\ref{prop:epsMonoton}, \ref{prop:kbounds}) in this paper are still applicable  if the target graph state $|G\>$ is replaced by $|\Psi\>$. Therefore, our verification protocol is of interest not only to blind MBQC, but also to many other tasks in quantum information processing  that entail high-security. More results on quantum state verification (QSV) in the adversarial scenario are presented in Supplementary Note~7.

	Up to now we have focused  on robust QSV in the adversarial scenario, in which the prepared state $\rho$ can be arbitrarily correlated or entangled, which is pertinent to blind MBQC. On the other hand, robust QSV in the i.i.d.\! scenario is also important to  many applications. 
	Although this scenario is much simpler than the adversarial scenario, the sample complexity of robust QSV  has not been clarified before. In the Methods section and Supplementary Note~8 we will discuss this issue in detail and clarify the sample complexity of robust QSV in the i.i.d.\! scenario in comparison with the adversarial scenario. Not surprisingly, most of our results on the adversarial scenario have analog for the i.i.d. scenario.

\bigskip	

\begin{center}
	\textbf{METHODS}
\end{center}
\vspace{1em}

\noindent\textbf{Protocols for realizing verifiable BQC}\\
To put our work into context, here we briefly review existing protocols for realizing verifiable BQC, which  can be broadly  divided into four classes \cite{Gheorg19}. Many protocols in the four classes build on the model of MBQC due to its convenience and flexibility. 

The first class of protocols work in the multi-prover setting \cite{Colada19,RUV13,GKW15,HPDF15}. 
These protocols can achieve a classical client (verifier), but a trade-off is the requirement of multiple non-communicating servers (provers) that share entanglement with each other, which is very difficult to realize in practice.

The second and third classes of protocols need only a single server, but assume that the client has  limited quantum computational power. The second class of protocols work in the prepare-and-send setting \cite{BarzNP13,Fitzsimons17,Aharonov10,Broadb15}, in which the client has a trusted preparation device and the ability to send single-qudit quantum states to the server. This class includes the protocol based on quantum authentication \cite{Aharonov10}, protocol based on repeating indistinguishable runs of tests and computations \cite{Broadb15}, and protocol based on trap qubits \cite{Fitzsimons17}, which has been  demonstrated experimentally  \cite{BarzNP13}.  The third class of protocols work in the receive-and-measure setting \cite{GKW15,HayaM15,Fujii17,MTH17,TMMMF19,Xu21}, in which the client receives quantum states from the server and has the ability to perform reliable local projective measurements. This class includes the protocol based on CHSH games \cite{GKW15}, protocols based on QSV in the adversarial scenario \cite{HayaM15,Fujii17,MTH17,TMMMF19,Xu21}, and our protocol. Notably, the above three classes of protocols are all information-theoretically secure \cite{Gheorg19}.

Recently, the forth class of protocols based on computational assumptions have been developed \cite{Mahadev18,GheorVidick19,Bartusek22,JZhang22}, which elegantly enables a classical client to hide and verify the quantum computation of a single server. However, these schemes are no longer information-theoretically secure, and their overheads are too prohibitive for any sort of practical implementation in the near future.

\vspace{1em}
\noindent\textbf{Simplifying the calculation of $\overline{\epsilon}_\lambda(k,N,\delta)$}\\
Here we show how to simplify the calculation of the guaranteed infidelity 
$\overline{\epsilon}_\lambda(k,N,\delta)$ given in \eref{eq:hatzeta2} by virtue of results derived in 
the companion paper \cite{Classical22}. 

Recall that $\Omega$ is a homogeneous strategy for the target state $|G\>\in\caH$ as shown in \eref{eq:strategy}. It
has the following spectral decomposition,
\begin{align}\label{eq:homoDecompose}
	\Omega
	= |G\>\< G|+\lambda(\openone-|G\>\<G|)
	= \Pi_1 +\lambda \sum_{j=2}^{D} \Pi_j ,
\end{align}
where $D$ is the dimension of $\caH$, and $\Pi_j$ are mutually orthogonal rank-1 projectors with $\Pi_1=|G\>\< G|$. In addition, $\rho$ is a permutation-invariant state on $\caH^{\otimes (N+1)}$. 
Note that $p_k(\rho)$ defined in \eref{eq:pnmkrho} and $f_k(\rho)$ defined in \eref{eq:fnmkrho} only depend on the
diagonal elements of $\rho$ in the product
basis constructed from the eigenbasis of $\Omega$ (as determined by $\Pi_j$).
Hence, we may assume that $\rho$ is diagonal in this basis without loss of generality.  In other words, $\rho$ can be expressed as a mixture of tensor products of $\Pi_j$. 
For $i=1,2,\dots,N+1$, we can associate the $i$th system of $\rho$ with a $\{0,1\}$-valued variable $Y_i$: we define $Y_i=0$ (1) if the state on the $i$th system is $\Pi_1$ ($\Pi_{j\ne1}$). 
Since the state $\rho$ is  permutation invariant, the variables $Y_1,\dots,Y_{N+1}$ are subject to 
a permutation-invariant joint distribution $P_{Y_1, \ldots, Y_{N+1}}$ on $[N+1]:=\{1,2,\dots,N+1\}$. 
Conversely, for any permutation-invariant joint distribution on $[N+1]$, we can always find a diagonal state $\rho$, 
whose corresponding variables $Y_1,\dots,Y_{N+1}$ are subject to this distribution. 

Next, we  define a $\{0,1\}$-valued random variable $U_i$ to  express the test outcome on the $i$th system, 
where 0 corresponds to passing the test and 1 corresponds to the failure.
If $Y_i=0$, which means the state on the $i$th system is $\Pi_1$, then the $i$th system must pass the test;  
if $Y_i=1$, which means the state on the $i$th system is $\Pi_{j\ne1}$, then the $i$th system passes the test with probability $\lambda$, 
and fails with probability $1-\lambda$. So  we have the following conditional distribution: 
\begin{align}\label{eq:deflambda}
	\begin{split}
		P_{U_i|Y_i}(0|0)=1, \qquad &P_{U_i|Y_i}(1|0)=0 , \\
		P_{U_i|Y_i}(0|1)=\lambda, \qquad &P_{U_i|Y_i}(1|1)=1-\lambda.
	\end{split}
\end{align}
Note that $U_i$ is determined by the random  variable $Y_i$ and the parameter $\lambda$ in \eref{eq:strategy}. 
Let $K$ be the random variable that counts the number of 1, that is, the number of failures, among $U_1,U_2,\ldots,U_N$. Then the probability that Alice accepts is
\begin{align}
	p_k(\rho)=\Pr(K \le k),
\end{align} 
given that Alice  accepts
if at most $k$ failures are observed among the $N$ tests. This probability only depends on the joint distribution $P_{Y_1, \ldots, Y_{N+1}}$. 
If at most $k$ failures are observed, then the fidelity of the state on the $(N+1)$th system can be expressed 
as the conditional probability 
\begin{align}
	F_k(\rho)=\Pr(Y_{N+1}\!=0|K \le k),
\end{align}  
which also only depends on $P_{Y_1, \ldots, Y_{N+1}}$. 
Hence, the guaranteed infidelity defined in \eref{eq:hatzeta2} can be expressed as 
\begin{align}\label{eq:epslamDef}
	&\overline{\epsilon}_\lambda(k,N,\delta) \nonumber\\
	&= 1-\min\, \{\Pr(Y_{N+1}=0|K \le k)\,|\Pr(K \le k) \geq \delta \} \nonumber\\
	&=   \max\, \{\Pr(Y_{N+1}=1|K \le k)\,|\Pr(K \le k) \geq \delta \},
\end{align}
where the optimization is taken over all permutation-invariant joint distributions $P_{Y_1, \ldots, Y_{N+1}}$.  

Equation~\eqref{eq:epslamDef} reduces the computation of $\overline{\epsilon}_\lambda(k,N,\delta)$ to the computation of a maximum conditional probability. The latter problem was studied in detail in our companion paper \cite{Classical22}, in which $\overline{\epsilon}_\lambda(k,N,\delta)$
is called the upper confidence limit. Hence, all properties of $\overline{\epsilon}_\lambda(k,N,\delta)$ derived in Ref.~\cite{Classical22} also hold in the current context. Notably, several results in  this paper are simple corollaries of the counterparts in Ref.~\cite{Classical22}.
To be specific, Proposition~1 follows from Proposition 6.5 in Ref.~\cite{Classical22}; 
Theorem~S1 in Supplementary Note~1 follows from Theorem 6.4 in Ref.~\cite{Classical22};
Lemma~S6 in Supplementary Note~2 follows from Lemma 6.7 in Ref.~\cite{Classical22}; 
Lemma~S7 in Supplementary Note~2 follows from Lemma 2.2 in Ref.~\cite{Classical22}; 
Proposition~S7 in Supplementary Note~7 follows from Lemma 5.4 and Eq.~(89) in Ref.~\cite{Classical22}.

Although  this paper and the companion paper \cite{Classical22} study essentially the same quantity $\overline{\epsilon}_\lambda(k,N,\delta)$, they have different focuses. In Ref.~\cite{Classical22}, we mainly focus on  asymptotic behaviors of $\overline{\epsilon}_\lambda(k,N,\delta)$
and its related quantities, which are of interest to the theory of statistical sampling and hypothesis testing. 
The main goal of Ref.~\cite{Classical22} is to show that the randomized test with parameter $\lambda>0$ can substantially improve the significance level over the deterministic test with $\lambda=0$. In this paper, by contrast, we focus on  finite bounds for $\overline{\epsilon}_\lambda(k,N,\delta)$ and its related quantities, which are important to  practical applications. In addition, the key result on robust verification, Theorem~\ref{thm:iidHighProb}, has no analog in the companion paper. The main goal of this paper is to provide a robust and efficient protocol for verifying the resource graph state in blind MBQC and clarify the sample complexity. So the two papers are complementary to each other.

It is worth pointing out that the `randomized test' considered in Ref.~\cite{Classical22} has a different meaning from the `quantum test' in this paper because of different conventions in the two communities. The `randomized test' in Ref.~\cite{Classical22} means the whole procedure that one observes the $N$ variables $U_1, U_2,\ldots, U_{N}$ and makes a decision based on the number of failures observed;  
while a `quantum test' in this paper means Alice performs a two-outcome  measurement on one system of the state $\rho$,  in which one outcome corresponds to passing the test, and the other outcome  corresponds to a failure.

\bigskip
\noindent\textbf{Robust and efficient verification of quantum states in the i.i.d.\! scenario}\\
Up to now we have focused  on QSV in the adversarial scenario, in which the server Bob can prepare an arbitrary state $\rho$ on the whole space $\caH^{\otimes (N+1)}$. In this section, we turn to the i.i.d.\! scenario, in which the prepared state is a tensor power of the form $\rho=\sigma^{\otimes (N+1)}$ with  $\sigma\in \mathcal{D}(\caH)$. This verification problem was originally  studied in  Refs.~\cite{Ha09-2,HMT} and later more systematically in Ref.~\cite{PLM18}. So far, efficient verification strategies based on  local operations and classical communication (LOCC) have been constructed for various classes of pure states, including bipartite pure states \cite{LHZ19,Wang19,Yu19}, stabilizer states (including graph states) \cite{HayaM15,PLM18,ZH3,ZhuEVQPSlong19,Dangn20}, hypergraph states \cite{ZH3}, weighted graph states \cite{HayaTake19}, Dicke states \cite{Liu19,Li21}, 
ground states of local Hamiltonians \cite{ZLC22,CLZ22}, and certain continuous-variable states \cite{LiuYC21},
see Refs.~\cite{Kliesch21,YSG21} for overviews. Verification protocols based on local collective measurements have also been constructed for Bell states \cite{Ha09-2,PhysRevLett.129.190504}. However, most previous works did not consider the problem of robustness. Consequently, most protocols known so far are not robust, and the sample cost may increase substantially if robustness is taken into account, see Supplementary Note~8~A for explanation. Only recently,  several  works considered the problem of robustness \cite{YSG21,WHZhang20,WHZhang20npj,Jiang20npj,Xia22NJP}; however, the degree of robustness of  verification protocols has not been analyzed, and the sample complexity of robust verification has not been clarified, although this problem is apparently much simpler than the counterpart in the adversarial scenario.

In this section, we propose a general approach for constructing robust and efficient verification protocols in the i.i.d.\! scenario and clarify the  sample complexity of robust verification. The results presented here can serve as a benchmark for understanding QSV in the adversarial scenario. 
To streamline the presentation, the proofs of these results [including Propositions~\ref{prop:epsiidMonoton}--\ref{prop:iidHighProbiid} and \eref{eq:iidNoptUB}] are relegated to Supplementary Note~8.

Consider a quantum device that is expected to produce the target state $|\Psi\>\in\caH$, but actually produces the  states $\sigma_1,\sigma_2,\dots,\sigma_N$  in $N$ runs. In the i.i.d.\! scenario, all these states are identical to the state $\sigma$, and the goal of Alice is to verify whether $\sigma$ is sufficiently close to the target state $|\Psi\>$. If a strategy $\Omega$ of the form in \eref{eq:homo} can be constructed for $|\Psi\>$, then our verification protocol runs as follows: Alice applies the strategy $\Omega$ to each of the $N$ states, and counts the number of failures.   
If at most $k$ failures are observed among the $N$ tests, then Alice accepts the states prepared; otherwise, she rejects.
Here $0\leq k\leq N-1$ is called the number of allowed failures.
The completeness of this protocol is guaranteed because the target state $|\Psi\>$ can never be mistakenly rejected.

Most previous works did not consider the problem of robustness and can reach a meaningful conclusion only when $k=0$ \cite{ZH4,Ha09-2,ZH3,HMT,LHZ19,Wang19,LiGHZ19,PLM18,ZhuEVQPSlong19,Li21}, i.e., Alice accepts iff all $N$ tests are passed. 
However, the requirement of passing all tests is too demanding in a realistic scenario and leads to  poor robustness, as clarified in Supplementary Note~8. To remedy this problem, several recent works considered modifications that allow some failures \cite{YSG21,WHZhang20,WHZhang20npj,Jiang20npj,Xia22NJP}. However, the robustness of such verification protocols has not been analyzed, and the sample complexity of robust verification has not been clarified.

Here we consider robust verification in which at most $k$ failures are allowed. Then the probability of acceptance is given by 
\begin{align}\label{eq:iidPrAcpt}
	p^{\iid}_{N,k}(\sigma)
	&=\sum_{j=0}^k \binom{N}{j} [1-\tr(\Omega \sigma)]^j \tr(\Omega \sigma)^{N-j} \nonumber\\
	&=B_{N,k}\big(1-\tr(\Omega \sigma) \big)
	=B_{N,k}(\nu\epsilon_\sigma), 
\end{align}
where $\epsilon_\sigma:=1-\<\Psi|\sigma|\Psi\>$ is the infidelity between $\sigma$ and the target state.
Similar to \eref{eq:hatzeta2}, for $0<\delta\leq1$ we define the guaranteed infidelity in the i.i.d.\! scenario as 
\begin{align}\label{eq:hatzetaiid}
	\bar{\epsilon}^{\,\iid}_{\lambda}(k,N,\delta)
	&:=  \max_{\sigma}   \left\{\epsilon_\sigma \,|\, p^{\iid}_{N,k}(\sigma)\geq \delta  \right\} \nonumber\\
	& =  \max_{\epsilon} \left\{0\leq \epsilon\leq1 \,|\, B_{N,k}(\nu\epsilon) \geq \delta  \right\},
\end{align}
where  the first maximization is taken over all states $\sigma$ on $\cal H$, and the second equality follows from \eref{eq:iidPrAcpt}. 
By definition, if Alice accepts the state $\sigma$, then she can ensure (with significance level $\delta$) that $\sigma$ has
infidelity at most $\bar{\epsilon}^{\,\iid}_{\lambda}(k,N,\delta)$ with the target state (soundness). 
Hence, $\bar{\epsilon}^{\,\iid}_{\lambda}(k,N,\delta)$ characterizes the verification precision in the i.i.d.\! scenario. Since the i.i.d.\! scenario has a stronger constraint than the full adversarial scenario, the guaranteed infidelity for the former scenario cannot be larger than that for the later scenario, that is, 
\begin{align}
	\bar{\epsilon}^{\,\iid}_{\lambda}(k,N,\delta) 
	\leq \bar{\epsilon}_{\lambda}(k,N,\delta),
\end{align}
as illustrated in Fig.~\ref{fig:Logplot_barEpsIID}. 

The following proposition clarifies the monotonicities of $\bar{\epsilon}^{\,\iid}_{\lambda}(k,N,\delta)$. 
It is the counterpart of Proposition~\ref{prop:epsMonoton}.

\begin{proposition}\label{prop:epsiidMonoton} 
	Suppose $0\leq\lambda<1$, $0<\delta\leq1$, $k\in \bbZ^{\geq 0}$, and $N\in \bbZ^{\geq k+1}$. 
	Then $\overline{\epsilon}^{\,\iid}_\lambda(k,N,\delta)$ is strictly decreasing in $\delta$ and $N$, but strictly increasing in $k$.
\end{proposition}

\begin{figure}
	\begin{center}
		\includegraphics[width=7.2cm]{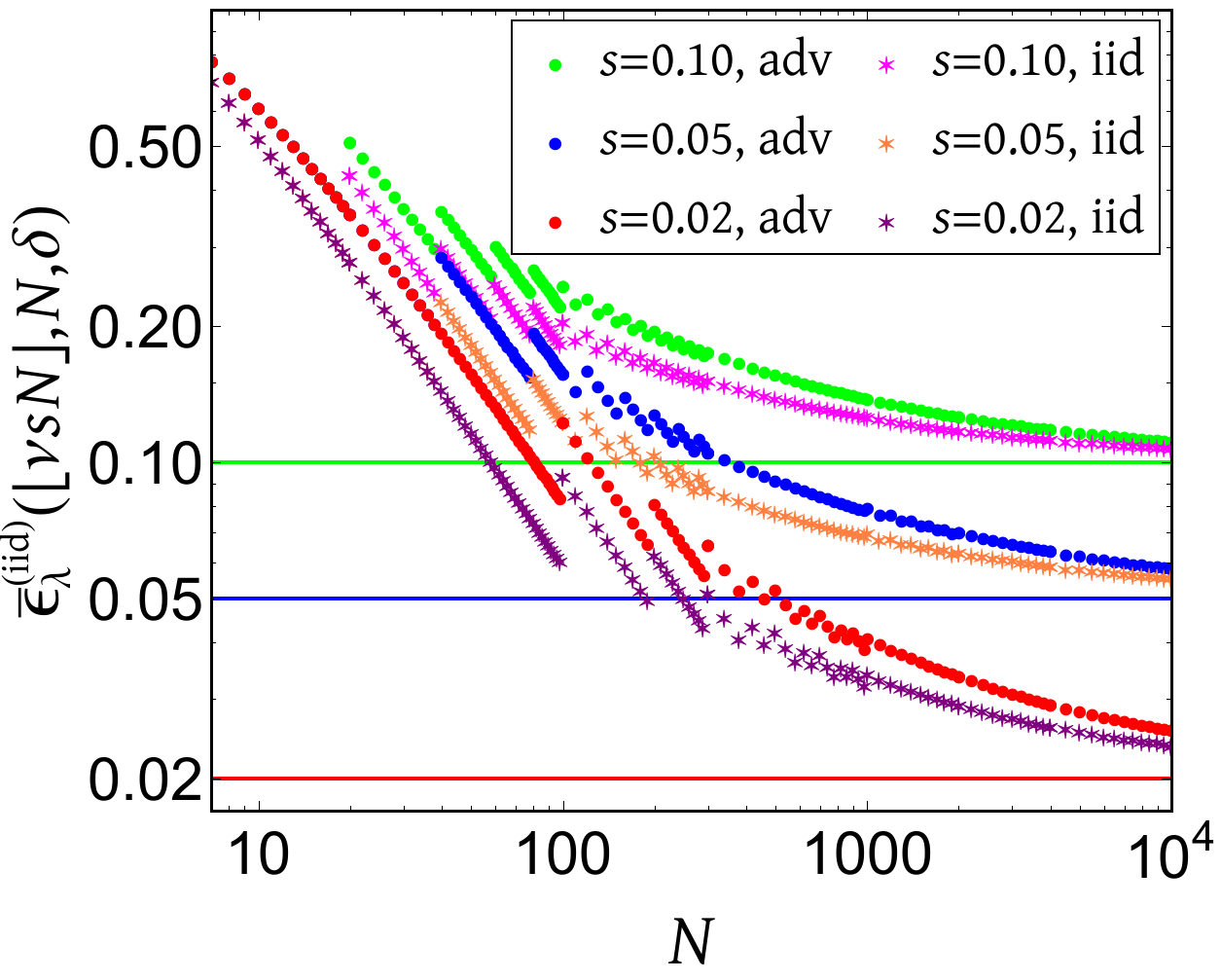}
		\caption{\label{fig:Logplot_barEpsIID}
			Guaranteed infidelities in the i.i.d.\! scenario and adversarial scenario. 
			Here $\lambda=1/2$ and $\delta=0.05$; the green, blue, and red dots represent $\bar{\epsilon}_{\lambda}(\lfloor\nu s N\rfloor, N, \delta)$
			[by Eq.~(6) in Supplementary Note 1] for the adversarial scenario, while the magenta, orange, and purple stars represent $\bar{\epsilon}^{\,\iid}_{\lambda}(\lfloor\nu s N\rfloor, N, \delta)$ [by \eref{eq:hatzetaiid}] for the i.i.d.\! scenario; each horizontal line represents an error rate $s$.
		}
	\end{center}
\end{figure}

Next, we consider the verification with a fixed error rate in the i.i.d.\! scenario. 
Concretely, we set the number of allowed failures $k$ to be proportional to the number of tests, i.e., $k=\lfloor s\nu N \rfloor$, where $0\leq s<1$ is the error rate, and $\nu=1-\lambda$ is the spectral gap of the strategy $\Omega$. The following proposition  provides informative bounds for 
$\bar{\epsilon}^{\,\iid}_{\lambda}(\lfloor\nu s N\rfloor,N,\delta)$. It is the counterpart of Theorem \ref{thm:Boundeps}.

\begin{proposition}\label{prop:Boundepsiid}
	Suppose $0<s,\lambda<1$, $0<\delta\leq 1/2$, and $N\in\bbZ^{\geq 1}$; then
	\begin{align}\label{eq:iidBDeps}
		\begin{split}
			s-\! \frac{1}{\nu N}
			<\bar{\epsilon}^{\,\iid}_{\lambda}(\lfloor\nu s N\rfloor,N,\delta) 
			\leq s+\! 
			\sqrt{\frac{2s \ln\delta^{-1}}{\nu N}}\! + \frac{2\ln\delta^{-1}}{\nu N}.
		\end{split}
	\end{align}
\end{proposition}

Similar to the behavior of $\bar{\epsilon}_{\lambda}(\lfloor\nu s N\rfloor,N,\delta)$, the guaranteed infidelity $\bar{\epsilon}^{\,\iid}_{\lambda}(\lfloor\nu s N\rfloor,N,\delta)$ for the i.i.d.\! scenario converges to the error rate 
$s$ as the number $N$ gets large, as illustrated in Fig.~\ref{fig:Logplot_barEpsIID}. 
To achieve a given infidelity $\epsilon$ and significance level $\delta$, which means $\bar{\epsilon}^{\,\iid}_{\lambda}(\lfloor\nu s N\rfloor,N,\delta)\leq\epsilon$, it suffices to  set $s< \epsilon$ and choose a sufficiently  large $N$. 
By virtue of Proposition~\ref{prop:Boundepsiid} we can derive the  following proposition,  which is the counterpart of  Theorem \ref{thm:UBtestsNumber}.

\begin{proposition}\label{prop:iidUBtestsNumber}
	Suppose  $0\leq \lambda<1$, $0\leq s<\epsilon<1$, and $0<\delta<1$.
	If the number of tests $N$ satisfies
	\begin{align}
		N\geq \frac{\ln\delta^{-1}}{D(\nu s\|\nu\epsilon)} ,
	\end{align}
	then  $\bar{\epsilon}^{\,\iid}_{\lambda}(\lfloor\nu s N\rfloor,N, \delta)\leq\epsilon$.
\end{proposition}

In the rest of this section, we turn to study the sample complexity of robust verification in the i.i.d.\! scenario. To verify the target state within infidelity $\epsilon$, significance level $\delta$, and robustness $r$ (with $0\leq r<1$) entails the following conditions,  
\begin{enumerate}
	\item[1.] (Soundness) If the device prepares i.i.d.\! states $\tau\in\mathcal{D}(\caH)$  with infidelity $\epsilon_\tau>\epsilon$, then the
	probability that Alice accepts $\tau$ is smaller than $\delta$.
	
	\item[2.] (Robustness) If the device prepares i.i.d.\! states $\tau\in\mathcal{D}(\caH)$  with infidelity $\epsilon_\tau\leq r\epsilon$, then
	the probability that Alice accepts $\tau$ is at least $1-\delta$.		
\end{enumerate}
Here the condition of robustness is the same as the counterpart in the adversarial scenario, while the condition of soundness is  different.  
In the adversarial scenario, once accepting, only the reduced state on the remaining unmeasured system can be used for application,   
so the condition of soundness only focuses on the fidelity of this state. In the i.i.d.\! scenario, by contrast, the prepared states are identical and independent, so the condition of soundness focuses on the fidelity of each state.

Given the total number $N$ of tests and the number  $k$ of allowed failures,  then the conditions of soundness and robustness can be expressed as 
\begin{align}\label{eq:robustConditionIID}
	B_{N,k}(\nu\epsilon)\leq\delta, \qquad B_{N,k}(\nu r\epsilon)\geq1-\delta. 
\end{align}
Let $N_{\rm min}^{\rm iid}(\epsilon,\delta,\lambda,r)$ be the minimum number of tests required for robust verification in the i.i.d.\! scenario. Then $N_{\rm min}^{\rm iid}(\epsilon,\delta,\lambda,r)$ is the minimum positive integer $N$ such that \eref{eq:robustConditionIID} holds for some  $0\leq k\leq N-1$, namely, 
\begin{align}\label{eq:iidOptNadvDef2}
	N_{\rm min}^{\rm iid}(\epsilon,\delta,\lambda,r)
	&:=
	\min_{N,k} \big\{ N \,\big|\,  k\in \bbZ^{\geq 0}, N\in\bbZ^{\geq k+1}, \nonumber\\
	&\ \ 
	B_{N,k}(\nu\epsilon)\leq\delta, B_{N,k}(\nu r\epsilon)\geq1-\delta \big\} . 
\end{align}
It is determined by $\nu\epsilon$, $\delta$, $r$,   
and is the counterpart of  $N_{\rm min}(\epsilon,\delta,\lambda,r)$ in the adversarial scenario.

\begin{figure}
	\begin{algorithm}[H]
		{\small
			\hspace{-98pt}\textbf{Input:}  $\lambda,\epsilon,\delta\in(0,1)$ and $r\in[0,1)$.\\
			\hspace{-71pt} \textbf{Output:} $k_{\min}^{\rm iid}(\epsilon,\delta,\lambda,r)$ and $N_{\rm min}^{\rm iid}(\epsilon,\delta,\lambda,r)$.
			
			\begin{algorithmic}[1]
				\caption{{\small Minimum test number for robust verification in the i.i.d.\! scenario}}
				\label{alg:iidNoptAdv}
				
				\If{$r=0$} 
				
				\State{$k_{\min}^{\rm iid}\leftarrow0$}
				
				\Else \For{$k=0,1,2,\dots$}
				
				\State{Find the largest integer $M$ such that $B_{M,k}(\nu r\epsilon)\geq1-\delta$.}
				
				\If{$M\geq k+1$ and $B_{M,k}(\nu\epsilon)\leq\delta$} 
				
				\State{stop}
				
				\EndIf
				
				\EndFor 
				
				\State{$k_{\min}^{\rm iid}\leftarrow k$}
				
				\EndIf
				
				\State{Find the smallest integer $N$ that satisfies $N\geq k_{\min}^{\rm iid}+1$ 
					and $B_{N,k_{\min}^{\rm iid}}(\nu\epsilon)\leq\delta$.}
				
				\State{$N_{\min}^{\rm iid}\leftarrow N$ }
				
				\State{\textbf{return} $k_{\min}^{\rm iid}$ and $N_{\min}^{\rm iid}$}
				
			\end{algorithmic}
		}
	\end{algorithm}
\end{figure}

\begin{figure}
	\includegraphics[width=8.63cm]{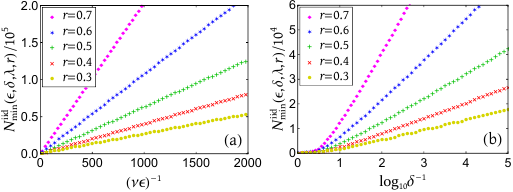}
	\caption{\label{fig:optNiid}
		Minimum number of tests required for robust verification in the i.i.d.\! scenario (by Algorithm~\ref{alg:iidNoptAdv}). (a)~Variations of $N_{\rm min}^{\rm iid}(\epsilon,\delta,\lambda,r)$ with $(\nu\epsilon)^{-1}$ and $r$,  where $\delta=0.01$. (b) Variations of $N_{\rm min}^{\rm iid}(\epsilon,\delta,\lambda,r)$ with $\log_{10}\delta^{-1}$ and $r$,  where $\nu\epsilon=0.005$.
	}
\end{figure}

Next, we propose a simple algorithm, Algorithm~\ref{alg:iidNoptAdv}, for computing $N_{\rm min}^{\rm iid}(\epsilon,\delta,\lambda,r)$, which is very useful to practical applications. This algorithm is the counterpart of  Algorithm~\ref{alg:NoptAdv} for computing $N_{\rm min}(\epsilon,\delta,\lambda,r)$. In addition to the number of tests, Algorithm~\ref{alg:iidNoptAdv}
also determines the corresponding number of allowed failures, which is denoted by $k_{\rm min}^{\rm iid}(\epsilon,\delta,\lambda,r)$. 
In Supplementary Note~8~F we explain why Algorithm~\ref{alg:iidNoptAdv} works.

Algorithm \ref{alg:iidNoptAdv} is quite useful to studying the variations of $N_{\rm min}^{\rm iid}(\epsilon,\delta,\lambda,r)$ with 
$\lambda$, $\delta$, $\epsilon$, and $r$ as illustrated in Fig.~\ref{fig:optNiid}.
When $\epsilon,r$ are fixed and $\delta$ approaches 0, $N_{\rm min}^{\rm iid}(\epsilon,\delta,\lambda,r)$ is proportional to $\ln\delta^{-1}$. 
When $\delta$ and $r$ are fixed, $N_{\rm min}^{\rm iid}(\epsilon,\delta,\lambda,r)$ is inversely proportional to $\nu\epsilon$. 
This fact shows that strategies with larger spectral gaps are more efficient, in sharp contrast with the adversarial scenario. 

At this point it is instructive  to compare the minimum number of tests for robust verification in the adversarial scenario 
with the counterpart in the i.i.d.\! scenario. Numerical calculation shows that the ratio of $N_{\rm min}(\epsilon,\delta,\lambda,r)$ over $N_{\rm min}^{\rm iid}(\epsilon,\delta,\lambda,r)$ is decreasing in $\lambda$,  as reflected in Fig.~\ref{fig:RatioNNiid}. For a typical value of $\lambda$, say $\lambda=1/2$, this ratio is smaller than 2, so the sample complexity in the adversarial scenario is comparable to the counterpart in the i.i.d.\! scenario. When $\lambda$ is small, one can construct another strategy with a larger $\lambda$ by adding the trivial test 
[see \eref{eq:strategy}], which can achieve a higher efficiency in the adversarial scenario. 
Due to this reason, the ratio of $N_{\rm min}(\epsilon,\delta,\lambda,r)$ over $N_{\rm min}^{\rm iid}(\epsilon,\delta,\lambda,r)$ is not so important when $\lambda \leq 0.3$.

\begin{figure}
	\includegraphics[width=6.5cm]{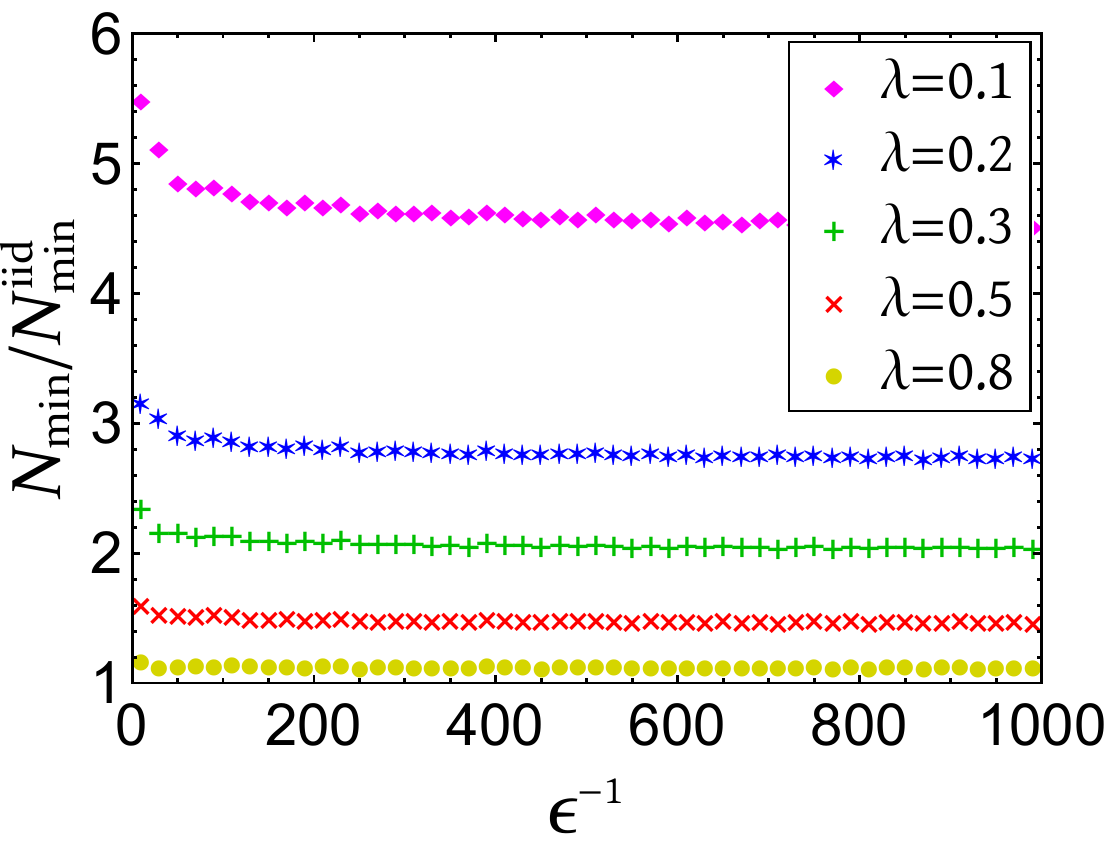}
	\caption{\label{fig:RatioNNiid}
		The ratio of $N_{\rm min}(\epsilon,\delta,\lambda,r)$ over $N_{\rm min}^{\rm iid}(\epsilon,\delta,\lambda,r)$
		with $\epsilon=\delta$ and $r=1/2$. 
		Here $N_{\rm min}(\epsilon,\delta,\lambda,r)$ and $N_{\rm min}^{\rm iid}(\epsilon,\delta,\lambda,r)$ 
		are the minimum numbers of tests required for robust verification in the adversarial scenario and i.i.d.\! scenario, respectively.  
	}
\end{figure}

The following proposition provides a guideline for choosing appropriate parameters $N$ and $k$ for achieving a given verification precision and robustness.

\begin{proposition}\label{prop:iidHighProbiid}
	Suppose $0<\delta,\epsilon,r<1$ and $0\leq~\!\!\!\lambda<~\!\!\!1$. Then the conditions of soundness and robustness in \eref{eq:robustConditionIID} hold as long as 
	$s\in (r\epsilon,\epsilon)$, $k=\left\lfloor\nu s N\right\rfloor$, and 
	\begin{align}\label{eq:choosesknIID}
		N \geq
		\left\lceil \frac{ \ln\delta^{-1}}{ \min\{D(\nu s\|\nu r\epsilon), D(\nu s\|\nu\epsilon)\} } \right\rceil. 
	\end{align}
\end{proposition}

For $0<p,r<1$ we define functions
\begin{align}
	\zeta(r,p)&:=
	p \left[ D\!\left( \frac{\ln\!\big(\frac{1-p}{1-rp}\big) }{\ln r + \ln\!\big(\frac{1-p}{1-rp}\big)} \Bigg\|p\right)\right] ^{-1}, \label{eq:zetarpxi}
	\\
	\xi(r) &:=
	\lim_{p \to 0}\zeta(r,p) \nonumber\\
	&\;= 
	\left[ \frac{r-1}{\ln r} \ln\!\left( \frac{r-1}{\ln r}\right) 
	+\left( 1- \frac{r-1}{\ln r}\right)\right]^{-1} . \label{eq:xir}
\end{align}
By virtue of Proposition~\ref{prop:iidHighProbiid} we can derive the following informative bounds (for $0<\delta,\epsilon,r<1$), 
\begin{align}\label{eq:iidNoptUB}
	N_{\rm min}^{\rm iid}(\epsilon,\delta,\lambda,r)
	\leq \! \left\lceil \frac{\ln\delta^{-1}}{\nu\epsilon} \zeta(r,\nu\epsilon) \right\rceil 
	\leq \! \left\lceil \frac{ \ln\delta^{-1}}{ \nu\epsilon} \xi(r) \right\rceil.  
\end{align} 
These bounds become tighter when the significance level $\delta$ approaches 0, as shown in Supplementary Figure~4. 
The coefficient $\xi(r)$ in the second bound is plotted in Supplementary Figure~5. 
When $r=\lambda=1/2$ for instance, the second bound in \eref{eq:iidNoptUB} implies that 
\begin{align}
	N_{\rm min}^{\rm iid}(\epsilon,\delta,\lambda,r)
	\leq \left\lceil \frac{2\,\xi(1/2)\ln\delta^{-1}}{\epsilon} \right\rceil 
	\leq \left\lceil \frac{46.5 \ln\delta^{-1}}{\epsilon} \right\rceil,
\end{align}   
while numerical calculation shows that $N_{\rm min}^{\rm iid}(\epsilon,\delta,\lambda,r)$ is smaller than 
$41 \,\epsilon^{-1}\ln\delta^{-1}$ for $\delta\geq 10^{-10}$ and approaches $2\,\xi(1/2)\,\epsilon^{-1}\ln\delta^{-1}$ when $\delta,\epsilon\to 0$. 
Therefore, our protocol can enable robust and efficient  verification of quantum states in the  i.i.d.\! scenario.

Finally, it is instructive to clarify the relation between QSV in the i.i.d.\! scenario, nonadversarial scenario, and adversarial scenario. In the i.i.d.\! scenario, the assumptions on the source are the strongest, so QSV  is the easiest, and the sample cost is the smallest. In the adversarial scenario, by contrast,  the assumptions on the source are the weakest, so QSV  is the most difficult, and the sample cost is the largest. For graph states with a prime local dimension, the sample cost in the adversarial scenario is comparable to the counterpart in the i.i.d.\! scenario thanks to our analysis above, which means the sample costs in all three scenarios are comparable.

\let\oldaddcontentsline\addcontentsline
\renewcommand{\addcontentsline}[3]{}

\let\addcontentsline\oldaddcontentsline

\bigskip
\noindent\textbf{ACKNOWLEDGMENTS}\\	
The work at Fudan is supported by the National Natural Science Foundation of China (Grants No.~92165109 and No.~11875110),  
National Key Research and Development Program of China (Grant No.~2022YFA1404204), and Shanghai Municipal Science and Technology Major Project (Grant No.~2019SHZDZX01). MH is supported in part by the National Natural Science Foundation of China (Grants No.~62171212 and No.~11875110)
and Guangdong Provincial Key Laboratory (Grant No.~2019B121203002).

\newpage
\onecolumngrid
\newpage
\vspace{2em}
\begin{center}
	\textbf{\large Robust and efficient verification of graph states in blind measurement-based\\ \vspace{0.5ex} quantum computation: Supplementary Information}
	\\ \vspace{3ex}
	\text{Zihao Li,\; Huangjun Zhu,\; and Masahito Hayashi}
\end{center}

\setcounter{page}{1}

\renewcommand{\figurename}{Supplementary Figure}
\renewcommand{\theequation}{\arabic{equation}}
\renewcommand{\thetable}{S\arabic{table}}
\renewcommand{\thetheorem}{S\arabic{theorem}}
\renewcommand{\thelemma}{S\arabic{lemma}} 
\renewcommand{\theproposition}{S\arabic{proposition}} 

\renewcommand{\thesection}{Supplementary Note~\arabic{section}}
\renewcommand{\thesubsection}{\Alph{subsection}}
\renewcommand{\thesubsubsection}{\alph{subsubsection}}

\setcounter{section}{0}
\setcounter{equation}{0}
\setcounter{figure}{0}
\setcounter{lemma}{0}
\setcounter{theorem}{0}
\setcounter{proposition}{0}	

\def\eqref#1{\textup{(\ref{#1})}}

\tableofcontents

\section{Analytical formula of $\overline{\epsilon}_\lambda(k,N,\delta)$}\label{app:epsanalytical}
In this section, we  provide an analytical formula for the guaranteed infidelity $\overline{\epsilon}_\lambda(k,N,\delta)$ defined in \eref{eq:hatzeta2} in the main text. 

For $0\leq p \leq 1$ and $z,k\in\bbZ^{\geq 0}$, define 
\begin{align}\label{eq:binomCFD}
	B_{z,k}(p):= \sum_{j=0}^k \binom{z}{j} p^j (1-p)^{z-j}.
\end{align}
Here it is understood that $x^0=1$ even if $x=0$.
For integer $0\leq z\leq N+1$, define
\begin{align}
	h_z(k,N,\lambda)&:=
	\begin{cases}
		1 & 0\leq z \le k, \\
		\frac{(N-z+1) B_{z,k}(\nu) +z B_{z-1,k}(\nu) }{N+1} & k+1 \leq z \leq N+1,
	\end{cases}
	\label{eq:hzHomo} \\
	g_z(k,N,\lambda)&:=
	\begin{cases}
		\frac{N-z+1}{N+1} &  0\leq z \le k, \\
		\frac{(N-z+1)B_{z,k}(\nu)}{N+1}   &  k+1 \leq z \leq N+1. 
	\end{cases}
	\label{eq:gzHomo}
\end{align}

\begin{lemma}[Lemma 6.2, \cite{Classical22}]\label{lem:gzhzMono}
	Suppose $0<\lambda<1$, $k,z,N\in\bbZ^{\geq 0}$, and $N \geq k+1$.
	Then 
	$h_z(k,N,\lambda)$ strictly decreases with $z$ for $k\leq z\leq N+1$, and
	$g_z(k,N,\lambda)$ strictly decreases with $z$ for $0\leq z\leq N+1$.
\end{lemma}

By this lemma, we have $B_{N,k}(\nu)=h_{N+1}(k,N,\lambda)\leq h_z(k,N,\lambda)\leq 1$ for $0\leq z\leq N+1$.
Hence, for $B_{N,k}(\nu)<\delta\leq 1$, we can define $\hat{z}$ as the largest integer $z$ such that
$h_z(k,N,\lambda)\geq\delta$.
For $z\in\{k,k+1,\dots,N\}$, define
\begin{align}
	\kappa_z(k,N,\delta,\lambda)
	:=\;& \frac{\delta-h_{z+1}(k,N,\lambda)}{h_z(k,N,\lambda)-h_{z+1}(k,N,\lambda)}, \label{eq:kappaz} \\
	\tilde{\zeta}_\lambda(k,N,\delta,z)
	:=\;& [1-\kappa_z(k,N,\delta,\lambda)]g_{z+1}(k,N,\lambda)
	+\kappa_z(k,N,\delta,\lambda)g_z(k,N,\lambda).  \label{eq:tildezeta1}
\end{align}
The exact value of $\overline{\epsilon}_\lambda(k,N,\delta)$ is determined by the following theorem, which follows from 
Theorem 6.4 in the companion paper \cite{Classical22} according to the discussions in the Methods section. 

\begin{theorem}[Theorem 6.4, \cite{Classical22}]\label{thm:AdvFidelity}
	Suppose $0<\lambda<1$, $0<\delta\leq1$,
	$k\in\bbZ^{\geq 0}$, and $N\in\bbZ^{\geq k+1}$. Then we have
	\begin{align}
		&\overline{\epsilon}_\lambda(k,N,\delta)
		=
		\begin{cases}
			1                             & \delta \leq B_{N,k}(\nu), \\
			1-\tilde{\zeta}_\lambda(k,N,\delta,\hat{z})/\delta >0 \  & \delta > B_{N,k}(\nu).
		\end{cases}
		\label{eq:True-plot}
	\end{align}
\end{theorem}

\section{Auxiliary lemmas}\label{app:UsefulLemma}
To establish our main results
here we prepare several auxiliary lemmas. 

\begin{lemma}\label{lem:ln(1-x)}
	Suppose $0<x<1$; then 
	\begin{align}\label{eq:ln(1-x)}
		\frac{1}{x}-1< -\frac{1}{\ln(1-x)}< \frac{1}{x}.
	\end{align}
\end{lemma}

\begin{proof}[Proof of \lref{lem:ln(1-x)}]
	The upper bound follows because $\ln(1-x)<-x$ for $0<x<1$. 
	To prove the lower bound, we define the function $\mu(x):=x/(1-x)+\ln(1-x)$, which 
	is strictly increasing in $x$ because 
	\begin{align}
		\frac{{\rm d } \mu(x)}{{\rm d} x} 
		=  \frac{x}{(1-x)^2} >0 \qquad \forall\, 0<x<1.  
	\end{align}
	This fact and the equality $\mu(0)=0$ imply that $\mu(x)>0$ for $0<x<1$, which implies the lower bound in \eref{eq:ln(1-x)}. 
\end{proof}

For $0< p,q<1$,
the relative entropy between two binary probability vectors $(p,1-p)$ and $(q,1-q)$ reads
\begin{align}\label{eq:RelEntropy}
	D(p\|q):=p\ln\frac{p}{q}+(1-p)\ln\frac{1-p}{1-q}.
\end{align}
\begin{lemma}\label{lem:DpqMonoton}
	When $0< p< q< 1$, $D(p\|q)$ is strictly increasing in $q$ and strictly decreasing in $p$;  
	when $0< q< p< 1$, $D(p\|q)$ is strictly decreasing in $q$ and strictly increasing in $p$. 
\end{lemma}

\begin{proof}[Proof of \lref{lem:DpqMonoton}]
	We have
	\begin{align}
		\frac{\partial D(p\|q)}{\partial q} &= \frac{q-p}{q(1-q)} , 
	\end{align}		
	which is positive when $0< p< q< 1$, and is negative when $0< q< p< 1$. 
	In addition, we have 
	\begin{align}
		\frac{\partial D(p\|q)}{\partial p} &= \ln \frac{p}{q} + \ln \frac{1-q}{1-p} ,
	\end{align}
	which is negative when $0< p< q< 1$, and is positive when $0< q< p< 1$. 
	These observations complete the proof. 
\end{proof}

\begin{lemma}[Lemma 3.2, \cite{Binomial22}]\label{lem:Bzkmono}
	Suppose $k,z\in\bbZ^{\geq 0}$, $0 \leq k \leq z$ and $0<p<1$. Then $B_{z, k}(p)$ is strictly increasing in $k$, strictly decreasing in $z$, and
	nonincreasing in $p$.
	In addition, $B_{z, k}(p)$ is strictly decreasing in $p$ when $k<z$.
\end{lemma}

For $0<p<1$ and $k,z\in\bbZ^{\geq 0}$, the Chernoff bound states that 
\begin{align}
	B_{z,k}(p)\leq \rme^{-z D(\frac{k}{z}\| p )}\qquad \forall k\leq p z.  \label{eq:ChernoffB} 
\end{align}
The following lemma provides a reverse Chernoff bound for $B_{z,k}(p)$.

\begin{lemma}[Proposition 5.4, \cite{Binomial22}]\label{lem:ChernoffRev}
	Suppose $k,z$ are positive integers that satisfy $k\leq z-1$ and $0<p<1$; then
	\begin{align}
		B_{z,k}(p) \geq \frac{1}{\rme \sqrt{k}} \rme^{-zD(\frac{k}{z}\|p)}	.  \label{eq:ChernoffRev}
	\end{align}
\end{lemma}

For $0<\delta\leq 1$, define
\begin{align}\label{eq:definez}
	z^*(k,\delta,\lambda):=\min\left\{z\in\bbZ^{\geq k}|B_{z,k}(\nu)\leq \delta\right\}, \qquad
	z_*(k,\delta,\lambda):=z^*(k,\delta,\lambda)-1.
\end{align}
The following two lemmas follow from results in our companion paper \cite{Classical22} according to the discussions in the Methods section. Note that $\nu=1-\lambda$.

\begin{lemma}[Lemma 6.7, \cite{Classical22}]\label{lem:zeta2Bound2}
	Suppose $0<\lambda<1$, $0<\delta\leq1/2$, $k\in\bbZ^{\geq 0}$, and $N\in\bbZ^{\geq k+1}$; then we have
	\begin{align}
		\min \left\{1, \frac{z_*-k}{\lambda(N+1)+\nu z_*-k }\right\}
		\leq 
		\bar{\epsilon}_{\lambda}(k,N,\delta)
		\leq
		\frac{z^*-k+1+\sqrt{\lambda k}}{\lambda(N-z^*)+z^*-k+1+\sqrt{\lambda k} }
		\leq 
		\frac{z^*-k+1+\sqrt{\lambda k}}{\lambda N }, 
		\label{eq:zeta2LB1-frac}
	\end{align}
	where $z^*$ and $z_*$ are shorthands for $z^*(k,\delta,\lambda)$ and $z_*(k,\delta,\lambda)$, respectively.
\end{lemma}

\begin{lemma}[Lemma 2.2, \cite{Classical22}]\label{lem:LLP}
	Suppose  $0< \delta \leq 1/2$, $0\leq \lambda<1$, $k\in\bbZ^{\geq 0}$, and $N\in\bbZ^{\geq k+1}$; then 
	we have
	\begin{align}\label{XMA6C}
		\overline{\epsilon}_\lambda(k,N,\delta)\ge
		\overline{\epsilon}_\lambda(k,N,1/2)
		\left\{
		\begin{array}{ll}
			=1 &\quad
			\hbox{ when \ }
			k\geq \nu N,\\
			> \frac{k}{\nu N}&\quad
			\hbox{ when \ }
			k<\nu N.
		\end{array}
		\right.
	\end{align} 
\end{lemma}

\begin{lemma}\label{lem:z*UB3}
	Suppose $0<\lambda<1$, $0<\delta<1$, and $k\in\bbZ^{\geq 0}$. Then
	\begin{align}
		z^*(k,\delta,\lambda) \leq \frac{k}{\nu}+\frac{\sqrt{2\nu k\ln\delta^{-1}}}{2\nu^2}+\frac{\ln\delta^{-1}}{2\nu^2}+1.
	\end{align}
\end{lemma}
\begin{proof}[Proof of \lref{lem:z*UB3}]
	According to \eref{eq:definez}	we have
	\begin{align}
		z^*(k,\delta,\lambda)
		&=\min\bigl\{z\in\bbZ^{\geq k} \, | B_{z,k}(\nu)\leq \delta \bigr\}
		\stackrel{(a)}{\leq} \min\left\{z \geq \frac{k}{\nu} \,\bigg|\,\rme^{-2z (\frac{k}{z}-\nu)^2}\leq \delta \right\} \nonumber \\
		&=\left\lceil \frac{k}{\nu}+\frac{\sqrt{(\ln\delta^{-1})^2+8\nu k\ln\delta^{-1}}+\ln\delta^{-1}}{4\nu^2} \right\rceil \nonumber \\
		&\stackrel{(b)}{\leq} \left\lceil \frac{k}{\nu}+\frac{\ln\delta^{-1}+\sqrt{8\nu k\ln\delta^{-1}}+\ln\delta^{-1}}{4\nu^2} \right\rceil
		\leq\frac{k}{\nu}+\frac{\sqrt{2\nu k\ln\delta^{-1}}}{2\nu^2}+\frac{\ln\delta^{-1}}{2\nu^2}+1,
	\end{align}
	where $(a)$ follows from the Hoeffding's bound, and $(b)$ follows from the relation $\sqrt{x+y}\leq\sqrt{x}+\sqrt{y}$ for $x,y\geq0$.
\end{proof}

\begin{lemma}\label{lem:z*UB3smallnu}
	Suppose $1/2\leq\lambda<1$, $0<\delta<1$, and $k\in\bbZ^{\geq 0}$. Then
	\begin{align}
		z^*(k,\delta,\lambda) \leq \frac{k+2\lambda\ln\delta^{-1}+\sqrt{2\lambda k\ln\delta^{-1}}}{\nu}+1. 
	\end{align}
\end{lemma}

\begin{proof}[Proof of \lref{lem:z*UB3smallnu}]
	According to \eref{eq:definez}	we have
	\begin{align}
		z^*(k,\delta,\lambda)
		&=\min\bigl\{z\in\bbZ^{\geq k} \, | B_{z,k}(\nu)\leq \delta \bigr\}
		\stackrel{(a)}{\leq} \min\left\{z \geq \frac{k}{\nu} \,\bigg|\rme^{-z D(\frac{k}{z}\| \nu )}\leq \delta \right\}  \nonumber\\
		&\stackrel{(b)}{\leq} \min\left\{z \geq \frac{k}{\nu} \,\bigg|\rme^{-\frac{z}{2\nu\lambda} (\frac{k}{z}- \nu )^2}\leq \delta \right\}
		=\left\lceil \frac{k+\lambda\ln\delta^{-1}+\sqrt{(\lambda\ln\delta^{-1})^2+ 2\lambda k\ln\delta^{-1}}}{\nu} \right\rceil  \nonumber\\
		&\stackrel{(c)}{\leq} \left\lceil \frac{k+2\lambda\ln\delta^{-1}+\sqrt{2\lambda k\ln\delta^{-1}}}{\nu} \right\rceil
		\leq \frac{k+2\lambda\ln\delta^{-1}+\sqrt{2\lambda k\ln\delta^{-1}}}{\nu}+1,
	\end{align}
	where $(a)$ follows from the Chernoff bound \eqref{eq:ChernoffB}, 
	$(b)$ follows from the relation $D(x\|y)\geq \frac{(x-y)^2}{2y(1-y)}$ for $0\leq x\leq y\leq 1/2$, 
	and $(c)$ follows from the relation $\sqrt{x+y}\leq\sqrt{x}+\sqrt{y}$ for $x,y\geq0$.
\end{proof}

\begin{lemma}\label{lem:1.098}
	Suppose $0<\lambda<1$ and $0<\delta\leq1/4$; then 
	\begin{align}\label{eq:lem1.098}
		\frac{\sqrt{2 \ln\delta^{-1}}}{2\nu}+\sqrt{\lambda \nu }
		\leq 
		\frac{\sqrt{\ln\delta^{-1}}}{\nu } . 
	\end{align}
\end{lemma}

\begin{proof}[Proof of \lref{lem:1.098}]
	Equation \eqref{eq:lem1.098} is equivalent to the following inequality, 
	\begin{align}\label{eq:condition1.098}
		\left( 1- \frac{1}{\sqrt{2}}\right)\sqrt{\ln\delta^{-1}} \geq  \nu\sqrt{\lambda \nu},
	\end{align}
	which in turn follows from the inequalities $\nu\sqrt{\lambda \nu}\leq \sqrt{27/256} \leq 0.33$ and $\sqrt{\ln\delta^{-1}} \geq \sqrt{\ln4} \geq 1.17$.
\end{proof}

\section{Performances of previous protocols for verifying the resource states in blind MBQC}\label{app:PreviousWork}

To illustrate the advantage of our verification protocol, here we provide more details on the performances of the previous protocols summarized in Table~1 in the main text, including 
protocols in Refs.~\cite{ZhuEVQPSlong19, HayaM15, TMMMF19, Takeuchi18, Fujii17}. It turns out none of these protocols can verify the resource states of blind MBQC in a robust and efficient way. Actually, most previous works did not consider the problem of robustness at all, because it is already very difficult to detect the bad case without considering robustness.

Suppose Bob is honest and prepares an i.i.d. state of the form $\rho=\tau^{\otimes (N+1)}$, where $\tau\in\mathcal{D}(\caH)$ has a high fidelity with the target state and is useful for MBQC. For an ideal verification protocol, such a state should be accepted with a high probability. However, the acceptance probability is very small for most protocols known in the literature. This is not surprising given that a large number of tests are required by these protocols to detect the bad case, which means it is difficult to get accepted even if Bob is honest. 
So many repetitions are necessary to ensure that Alice can accept the state  preparation at least once, which may substantially increase the actual sample cost.

To be concrete, suppose Alice repeats the verification protocol $M$ times, and the 
acceptance probability of each run is $p_{\text{acc}}$ (depending on $\tau$). 
Then the probability that she accepts the state $\tau$ at least once
reads $1-(1-p_{\text{acc}})^M$. To ensure that this probability is at least
$1-\tilde\delta$ with $0<\tilde\delta<1$, the minimum of $M$ reads
\begin{align}
	M
	= \left\lceil \frac{\ln\tilde\delta}{\ln\left( 1-p_{\text{acc}}\right) } \right\rceil
	\geq  \frac{\ln\tilde\delta^{-1}}{p_{\text{acc}}}-\ln\tilde\delta^{-1}
	\approx \frac{\ln\tilde\delta^{-1}}{p_{\text{acc}}}, \label{eq:repetitionNum}
\end{align}
where the inequality follows from \lref{lem:ln(1-x)} and the approximation is applicable when $p_{\text{acc}}\ll 1$. To make a fair comparison with our protocol presented in the main text, we can choose $\tilde\delta=\delta$. When robustness is taken into account, therefore, the actual sample complexity 
will be increased by a factor of $M\approx (\ln\delta^{-1})/p_{\text{acc}}$, which is a very large overhead for most previous protocols.

\subsection{\label{sec:HM}HM protocol \cite{HayaM15}}
In Ref.~\cite{HayaM15}, Hayashi and Morimae (HM) introduced a protocol for verifying two-colorable graph states 
in the  adversarial scenario; the state $\rho$ prepared by Bob is accepted only if all tests are passed. To verify an $n$-qubit two-colorable graph state $|G\>\in\caH$ within target infidelity $\epsilon$ and significance level $\delta$, 
this protocol requires 
\begin{align}\label{eq:NumHM}
	N_{\rm HM}=\left\lceil \frac{1}{\delta\epsilon}-1\right\rceil=\Theta(\epsilon^{-1}\delta^{-1})
\end{align} 
tests. This scaling is optimal with respect to $\epsilon$ and $n$, but not optimal with respect to $\delta$.

Next, we  analyze the robustness of the HM protocol, which is not considered in the original paper \cite{HayaM15}.
When Bob generates i.i.d.\! quantum states $\tau\in\mathcal{D}(\caH)$ with infidelity $\epsilon_\tau$, 
the average probability that $\tau$ passes each test is at most $1-\epsilon_\tau/2$. 
So the probability $p_{\text{acc}}^{\rm HM}$ that Alice accepts $\tau$ satisfies 
\begin{align}\label{eq:HM15accp}
	p_{\text{acc}}^{\rm HM}
	\leq \left( 1-\frac{\epsilon_\tau}{2}\right) ^{N_{\rm HM}}.
\end{align}
For high precision verification, we have $\epsilon,\delta\ll1$ and  $N_{\rm HM}\gg1$, 
so the RHS of \eref{eq:HM15accp} is $\Theta(1)$ only if $\epsilon_\tau=O(N_{\rm HM}^{-1})=O(\delta\epsilon)$. 
Note that $\delta\epsilon$ is much smaller than  the target infidelity $\epsilon$, so it is in general very difficult to achieve such a low infidelity even if the target infidelity $\epsilon$ is accessible. Therefore, the robustness of the HM protocol is very poor.

\begin{figure}[b]
	\begin{center}
		\includegraphics[width=7.8cm]{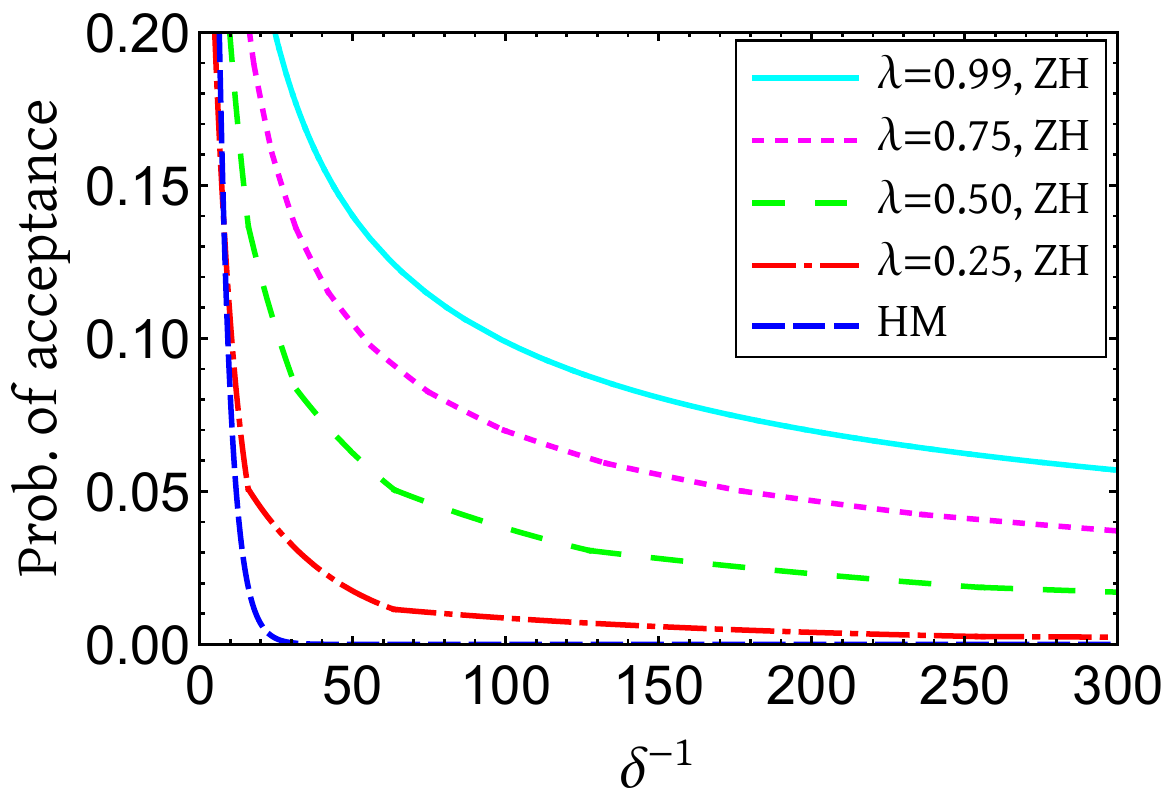}
		\caption{\label{fig:ProbAcceptPrevious}
			The maximum probability that Alice accepts i.i.d.\! quantum states $\tau\in \mathcal{D}(\caH)$ with previous verification protocols. 
			Here $\delta$ is the significance level, $\epsilon=0.01$ is the  target infidelity, and 
			$\tau$ has infidelity $\epsilon_\tau=\epsilon/2$ with the target state $|G\>$. 
			The blue curve corresponds to the HM protocol \cite{HayaM15} and is based on \eref{eq:HM15accp};    
			the other four curves correspond to the ZH protocol \cite{ZhuEVQPSlong19} and are based on \eref{eq:piidapp1}. 
		}
	\end{center}
\end{figure}

If the ratio $\epsilon_\tau/\epsilon$ is a constant, then the probability $p_{\text{acc}}^{\rm HM}$ decreases rapidly as $\delta$ decreases. 
When $\epsilon_\tau=\epsilon/2$ and $\delta\leq1/8$ for example, this probability is upper bounded by 
\begin{align}\label{eq:HMpassPrIIDUB}
	p_{\text{acc}}^{\rm HM}
	\leq \left( 1-\frac{\epsilon}{4}\right)^{\lceil 1/(\delta\epsilon)-1\rceil}\leq \exp\!\left(- \frac{1}{4\delta}  \right),
\end{align}
as illustrated in Supplementary Figure~\ref{fig:ProbAcceptPrevious}.
Here the first inequality follows from Eqs.~\eqref{eq:NumHM} and \eqref{eq:HM15accp}, and the second inequality is proved in \ref{sec:probHMproof}. This probability decreases exponentially with $\delta^{-1}$ and is already extremely small when $\epsilon=\delta=0.01$, in which case $p_{\text{acc}}^{\rm HM}
< 1.35 \times 10^{-11}$. According to \eref{eq:repetitionNum}, to ensure that Alice accepts the state prepared by Bob at least once with confidence level $1-\delta$, the number of repetitions is given by
\begin{align}
	M_{\rm HM}
	= \left\lceil \frac{\ln\delta}{\ln\left( 1-p_{\text{acc}}^{\rm HM}\right) } \right\rceil
	\geq  \frac{\ln\delta^{-1}}{p_{\text{acc}}^{\rm HM}}-\ln\delta^{-1}
	\geq\exp\left(\frac{1}{4\delta}  \right) \ln\delta^{-1}-\ln\delta^{-1}\approx \exp\left(\frac{1}{4\delta}  \right) \ln\delta^{-1}.
\end{align}
Therefore, the total sample cost reads
\begin{align}\label{eq:TotalHM}
	M_{\rm HM}N_{\rm HM}
	= \left\lceil \frac{1}{\delta\epsilon}-1 \right\rceil  
	\left\lceil \frac{\ln\delta}{\ln\left( 1-p_{\text{acc}}^{\rm HM}\right) } \right\rceil
	\geq \Theta\left( \frac{1}{\delta\epsilon}\exp\left(\frac{1}{4\delta}  \right)\ln\delta^{-1}\right),
\end{align}
which increases exponentially with $1/\delta$.
By contrast, the sample cost of our
protocol is only $O(\epsilon^{-1}\ln \delta^{-1})$, which improves the scaling behavior with $\delta$ doubly exponentially.

\subsection{ZH protocol \cite{ZhuEVQPSlong19}}
In Ref.~\cite{ZhuEVQPSlong19}, Zhu and Hayashi (ZH) introduced a protocol for verifying 
qudit stabilizer states (including graph states) in the  adversarial scenario. 
In this protocol, Alice uses the strategy in \eref{eq:strategy} in the main text to test $N$ systems of the state $\rho$ prepared by Bob, and 
she accepts the state on the remaining system iff all $N$ tests are passed. 
Hence, the ZH protocol can be viewed as a special case of our protocol with $k=0$.  
To verify an $n$-qudit graph state $|G\>\in\caH$ within target infidelity $\epsilon$ and significance level $\delta$, 
the ZH protocol requires $N_{\rm ZH}=\Theta(\epsilon^{-1}\ln\delta^{-1})$ tests, which is optimal with respect to all $n$, $\epsilon$ and $\delta$ if robustness is not taken into account. 
The analytical formula of  $N_{\rm ZH}$ is provided in Theorem 2 of Ref.~\cite{ZhuEVQPSlong19}.

The robustness of the ZH protocol is higher than the HM protocol, but is still not satisfactory. 
When Bob generates i.i.d.\! quantum states $\tau\in\mathcal{D}(\caH)$ with infidelity $\epsilon_\tau$, 
the average probability that $\tau$ passes each test in the ZH protocol is $1-\nu\epsilon_\tau$, 
where $\nu=1-\lambda$ is the spectral gap of the strategy in \eref{eq:strategy} in the main text.
The probability that Alice accepts $\tau$ reads [cf. \eref{eq:pracck=0} in the main text] 
\begin{align}\label{eq:piidapp1}
	p_{\text{acc}}^{\rm ZH}
	=(1-\nu\epsilon_\tau)^{N_{\rm ZH}}.
\end{align}
For high precision verification, we have $\epsilon,\delta\ll1$ and  $N_{\rm ZH}\gg1$, 
so $p_{\text{acc}}^{\rm ZH}$ is $\Theta(1)$ only if $\epsilon_\tau=O(N_{\rm ZH}^{-1})=O(\epsilon/(\ln\delta^{-1}))$, which is much more demanding than achieving the target infidelity $\epsilon$. So the robustness of the ZH protocol is not satisfactory. 

In most cases of practical interest, we have $0<\delta\leq\lambda<1$. If in addition $\epsilon_\tau=\epsilon/2\leq 1/4$, then 
\begin{align}\label{eq:piidapp2}
	p_{\text{acc}}^{\rm ZH}
	=\left( 1-\frac{\nu\epsilon}{2}\right) ^{N_{\rm ZH}}\leq \frac{3}{2}\sqrt{\delta},
\end{align}
as illustrated in Supplementary Figure~\ref{fig:ProbAcceptPrevious}. 
Here the inequality is proved in \ref{sec:probZHproof}. According to \eref{eq:repetitionNum}, to ensure that Alice accepts the state prepared by Bob at least once with confidence level $1-\delta$, the number of repetitions reads
\begin{align}
	M_{\rm ZH}
	= \left\lceil \frac{\ln\delta}{\ln\left( 1-p_{\text{acc}}^{\rm ZH}\right) } \right\rceil
	\geq \frac{\ln\delta^{-1}}{p_{\text{acc}}^{\rm ZH}} -\ln\delta^{-1}\geq \frac{2\ln\delta^{-1}}{3\sqrt{\delta}} -\ln\delta^{-1}
	=\Theta\left(\frac{\ln\delta^{-1}}{\sqrt{\delta}}\right).
\end{align}
So the total sample cost reads
\begin{align}
	M_{\rm ZH}N_{\rm ZH}
	\geq \Theta\left(\frac{(\ln\delta)^2}{\epsilon\sqrt{\delta}}\right).
\end{align}
By contrast, the sample cost of our
protocol is only $O(\epsilon^{-1}\ln \delta^{-1})$, which improves the scaling behavior with $\delta$ exponentially. 

When $\lambda=1/2$ and $\epsilon=\delta=0.01$
for example, calculation shows that $N_{\rm ZH}=1307$ and $p_{\text{acc}}^{\rm ZH}\approx 3.79\%$. 

\subsection{TMMMF protocol \cite{TMMMF19}}
In Ref.~\cite{TMMMF19}, Takeuchi, Mantri, Morimae, Mizutani, and Fitzsimons (TMMMF) 
introduced a protocol for verifying qudit graph states in the  adversarial scenario.
Let $|G\>\in\caH$ be the target graph state of $n$ qudits. 
In the TMMMF protocol, the total number of copies required is $N_{\rm TMMMF}=2n\lceil (5 n^{4} \ln n )/ 32\rceil$, and
the  number of tests required is $N_{\text {test}}\!=n\lceil (5 n^{4} \ln n )/ 32\rceil$. 
If at least $[1-(2n^2)^{-1}]N_{\text {test}}$ tests are passed, then the verifier Alice can guarantee that 
the reduced state $\sigma$ on one of the remaining systems satisfies
\begin{align}
	\langle G|\sigma| G\rangle \geq 1-\frac{2 \sqrt{c}+1}{n}
\end{align}
with significance level $n^{1-5 c / 64}$,
where $c$ is a constant that satisfies $64/5<c<(n-1)^{2} / 4$. The number of required copies $N_{\rm TMMMF}$ grows rapidly with $n$, 
which makes the TMMMF protocol hardly practical.

By contrast, to verify an $n$-qudit graph state $|G\>$ with infidelity $\epsilon=(2 \sqrt{c}+1)/n$ 
and significance level $\delta=n^{1-5 c / 64}$, the number of copies required by our protocol is 
\begin{align}
	N
	=O\left( \epsilon^{-1}\ln\delta^{-1}\right)
	=O\left( \frac{n}{2 \sqrt{c}+1}\left(\frac{5c}{64}-1\right)\ln n \right)
	\leq O(n^2\ln n) .
\end{align}
If the parameter $c$ is a constant independent of $n$, then
only $N=O(n\ln n)$ copies are needed.
So  our protocol is much more efficient than  the TMMMF protocol. In addition, in our protocol, it is easy to adjust the verification precision as quantified by the infidelity and significance level. By contrast, the TMMMF protocol does not have this flexibility because  the target
infidelity and significance level are intertwined with the qudit number.

To analyze the robustness of the TMMMF protocol, suppose Bob generates i.i.d.\! states $\tau\in\mathcal{D}(\caH)$ with infidelity $\epsilon_\tau=\epsilon/2$. 
In this case, the following proposition proved in \ref{sec:probTMMMFproof} shows that the probability of acceptance decreases exponentially fast with the qudit number $n$
if Alice applies the TMMMF protocol.
Therefore, the robustness of the TMMMF protocol is very  poor.

\begin{proposition}\label{prop:TMMMF19passPrIIDUB}
	Suppose Alice applies the TMMMF protocol to verify the $n$-qudit graph state $|G\>$
	within infidelity $\epsilon=(2 \sqrt{c}+1)/n$ and significance level  $\delta=n^{1-5 c/64}$, where $64/5<c<(n-1)^{2}/4$. 
	When Bob generates i.i.d.\! states $\tau\in\mathcal{D}(\caH)$ with infidelity $\epsilon_\tau=\epsilon/2$, 
	the probability that Alice accepts satisfies $p_{\rm acc}^{\rm TMMMF}\leq\exp(-0.245\, n^3 \ln n )$. 
\end{proposition}

For example, in order to reach infidelity $\epsilon=0.01$, 
the qudit number $n$ should satisfy $n\geq \lceil(2 \sqrt{c}+1)/0.01\rceil\geq816$, which means $N_{\rm TMMMF}\geq7\times 10^{14}$. 
If Bob prepares i.i.d.\! states $\tau\in\mathcal{D}(\caH)$ with infidelity $\epsilon_\tau=\epsilon/2=0.005$, then 
Proposition~\ref{prop:TMMMF19passPrIIDUB} implies that the acceptance probability satisfies 
\begin{align}
	p_{\rm acc}^{\rm TMMMF}
	\leq \exp(-0.245\, n^3 \ln n)
	\leq 10^{-387598171} .
\end{align} 
According to \eref{eq:repetitionNum}, to ensure that Alice accepts the state prepared by Bob at least once with confidence level $1-\delta$, the number of repetitions reads
\begin{align}
	M_{\rm TMMMF}
	= \left\lceil \frac{\ln\delta}{\ln\left( 1-p_{\text{acc}}^{\rm TMMMF}\right) } \right\rceil
	\geq 10^{3\times 10^8} \ln\delta^{-1}.
\end{align}
Consequently, the total number of copies consumed by Alice is  
\begin{align}\label{eq:TotalTMMMF}
	M_{\rm TMMMF}N_{\rm TMMMF} 
	\geq 7\times 10^{14} \times 10^{3\times 10^8} \ln\delta^{-1} 
	\geq 10^{3\times 10^8} \ln\delta^{-1}, 
\end{align}
which is astronomical. 

\subsection{TM protocol \cite{Takeuchi18}}
In Ref.~\cite{Takeuchi18}, Takeuchi and Morimae (TM) introduced a protocol for verifying 
qubit hypergraph states (including graph states) in the  adversarial scenario.
Let  $k\geq (4n)^7$ and $m\geq (2\ln 2)n^3k^{18/7}$  be positive integers.
To verify an $n$-qubit graph state within infidelity $\epsilon=k^{-1/7}$ and significance level $\delta=k^{-1/7}$. The number of required tests is $N_{\text {test}}=nk$, and the total number of samples is 
\begin{align}\label{eq:TakeNum}
	N_{\rm TM} = m+nk
	\geq (2\ln 2)n^3 k^{18/7} + nk 
	\geq (2\ln 2)n^3 (4n)^{18} + n(4n)^7
	\geq 9.5\times 10^{10} n^{21},
\end{align}
which is astronomical and too prohibitive for any practical application. By contrast, the sample number required by our protocol to achieve the same precision is only 
\begin{align}
	N
	=O\bigl( \epsilon^{-1}\ln\delta^{-1}\bigr)
	=O\bigl( k^{1/7} \ln k \bigr),  
\end{align}
which is much smaller than $N_{\rm TM}$. 
Furthermore, in the TM protocol, the choices of the target infidelity $\epsilon$ and significance level $\delta$ are restricted, 
while our protocol is applicable for all valid choices of $\epsilon$ and $\delta$. Since the TM protocol is impossible to realize in practice, we do not  
analyze its robustness further.

\subsection{FH protocol \cite{Fujii17}}
In Ref.~\cite{Fujii17}, Fujii and Hayashi (FH) introduced a protocol for verifying fault-tolerant MBQC based on two-colorable graph states. 
Let $|G\>\in\caH$ be the ideal two-colorable resource graph state of $n$ qubits, $\caH_S$ be the subspace of states that  are error-correctable, and $\Pi_S$ be the projector onto $\caH_S$. 
With the FH protocol, by performing 
$N_{\rm FH}=\lceil 1/(\delta\epsilon)-1\rceil$ tests, Alice can ensure that 
the reduced state $\sigma$ on the remaining system 
satisfies $\tr(\sigma\Pi_S)\geq 1-\epsilon$ with significance level $\delta$. 
That is, Alice can verify whether or not the given state belongs to the class of error-correctable states.  
Once accepting, she can safely use the reduced state $\sigma$ to perform fault-tolerant MBQC. 
In this sense, the FH protocol is fault tolerant, which offers a  kind of robustness different from our protocol. The scaling of  $N_{\rm FH}$ is optimal with respect to both $\epsilon$ and $n$, but not optimal with respect to $\delta$. 

Since the FH protocol relies on a given quantum error correcting code, it is difficult to realize for NISQ devices because too many physical qubits are required to encode logical qubits. In addition, the FH protocol is  robust only to certain  correctable errors. If the actual error is not correctable (even if the error probability is small), then the probability of acceptance will  decrease exponentially  with $\delta^{-1}$ (cf. the discussion in \ref{sec:HM}). By contrast, our protocol is robust against arbitrary
error as long as the  error probability is small.

The core idea of the FH protocol lies in quantum subspace verification. 
Later the idea of subspace verification was studied more systematically in Refs.~\cite{ZLC22,CLZ22}. 
In principle, this idea can be combined with our protocol to construct  verification protocols that are robust to both correctable errors and noncorrectable errors. This line of research deserves further exploration in the future.

\subsection{\label{sec:probHMproof}Proof of \eref{eq:HMpassPrIIDUB}}

When $0<\epsilon<1$ and $0<\delta\leq1/8$, we have 
\begin{align}\label{eq:pdelexp1HM15}
	\left( 1-\frac{\epsilon}{4}\right)^{\lceil 1/(\delta\epsilon)-1\rceil} 
	\leq \left( 1-\frac{\epsilon}{4}\right) ^{1/(\delta\epsilon)-1}
	=  \exp\left[ h_{\delta}(\epsilon)  \right]
	\stackrel{(a)}{\leq} \exp\left[ \lim_{\epsilon\to0} h_{\delta}(\epsilon)  \right]
	=  \exp\left(- \frac{1}{4\delta}  \right), 
\end{align}
which confirms \eref{eq:HMpassPrIIDUB}. 
Here the function $h_{\delta}(x)$ is defined as
\begin{align}\label{eq:defhdelx}
	h_{\delta}(x) := \left(\frac{1}{\delta x}-1 \right)  \ln\!\left( 1-\frac{x}{4} \right) , 
\end{align}	
and $(a)$ follows from \lref{lemma:defhdelx} below.

\begin{lemma}\label{lemma:defhdelx}
	When $0<\delta\leq1/8$, the function  $h_{\delta}(x)$ defined in \eref{eq:defhdelx}
	is strictly decreasing in $x$ for $0<x\leq 1$.  
\end{lemma}

\begin{proof}[Proof of \lref{lemma:defhdelx}]
	For $0\leq x\leq 1$ and $0<\delta\leq1/8$, we define the function 
	\begin{align}
		\hat{h}_{\delta}(x) := \ln\!\left( 1-\frac{x}{4} \right) + \frac{x(1-\delta x)}{4-x},  
	\end{align}	
	which is strictly increasing in $x$ for $0< x\leq 1$, because 
	\begin{align}
		\frac{{\rm d }\hat{h}_{\delta}(x)}{{\rm d} x} 
		= -\frac{1}{4-x} + \frac{4-8\delta x+\delta x^2}{(4-x)^2} 
		= \frac{x(1-8\delta +\delta x)}{(4-x)^2}  
		>  0
		\qquad \forall\,0< x\leq 1,\, 0<\delta\leq1/8. 
	\end{align}
	This fact and the relation $\hat{h}_{\delta}(0)=0$ together imply that $\hat{h}_{\delta}(x)>0$ for $0< x\leq 1$ and $0<\delta\leq1/8$. 
	Therefore, 
	\begin{align}
		\frac{{\rm d }h_{\delta}(x)}{{\rm d} x} 
		&= -\frac{1}{\delta x^2} \ln\!\left( 1-\frac{x}{4} \right)  -  \frac{1-\delta x}{\delta x (4-x)}
		= -\frac{\hat{h}_{\delta}(x)}{\delta x^2} 
		< 0 
		\qquad \forall\,0< x\leq 1,\, 0<\delta\leq1/8,  
	\end{align}
	which confirms \lref{lemma:defhdelx}. 
\end{proof}

\subsection{\label{sec:probZHproof}Proof of \eref{eq:piidapp2}}

Let $N_{\rm ZH}(\epsilon,\delta,\lambda)$ be the number of tests required to 
reach infidelity $\epsilon$ and significance level $\delta$. According to Theorem~3 in Ref.~\cite{ZhuEVQPSlong19} we have 
\begin{align}\label{eq:Nmindel0}
	N_{\rm ZH}(\epsilon,\delta,\lambda)
	\geq k_{-}+\left\lceil\frac{k_{-} (1-\epsilon)}{\lambda \epsilon}\right\rceil
	\geq k_{-}+ \frac{k_{-} (1-\epsilon)}{\lambda \epsilon} 
	=    \frac{(1-\nu\epsilon)k_{-}}{\lambda \epsilon}
	\geq \frac{1-\nu\epsilon}{\lambda \epsilon \ln \lambda} (\ln \delta - \ln \lambda),   
\end{align}
where $k_{-}:= \left\lfloor\log _\lambda \delta\right\rfloor$. 

In preparation for further discussion, for $0<x<1$ we define the function 
\begin{align}
	f(x) := \left( \frac{1}{2x} -1 \right) \ln (1-x),  
\end{align}
which is strictly increasing in $x$, because 
\begin{align}
	\frac{{\rm d }f(x)}{{\rm d} x} 
	= -\frac{\ln(1-x)}{2x^2} - \frac{1-2x}{2x(1-x)} 
	\geq  \frac{x}{2x^2} - \frac{1-2x}{2x(1-x)} 
	= \frac{1}{2(1-x)}>0. 
\end{align}

The following proof is composed of two steps.

\noindent \textbf{Step 1:} The aim of this step is to prove \eref{eq:piidapp2} in the case $0<\delta\leq\lambda^2$. 

In this case we have 
\begin{align}
	p_{\text{acc}}^{\rm ZH}
	&= \left( 1-\frac{\nu\epsilon}{2}\right) ^{N_{\rm ZH}(\epsilon,\delta,\lambda)}
	\stackrel{(a)}{\leq}  \exp\left[(\ln \delta - \ln \lambda)
	\frac{1-\nu\epsilon}{\lambda \epsilon \ln \lambda} \ln\!\left( 1-\frac{\nu\epsilon}{2}\right) \right]
	= 
	\exp \left[\frac{\nu(\ln \delta - \ln \lambda)}{\lambda \ln \lambda} f\!\left( \frac{\nu\epsilon}{2}\right) \right]\nonumber
	\\&\stackrel{(b)}{\leq} 
	\exp \left[\frac{\nu(\ln \delta - \ln \lambda)}{\lambda \ln \lambda} f\!\left( \frac{\nu}{4}\right) \right] 
	=
	\exp \left[(\ln \delta - \ln \lambda) \frac{1+\lambda}{\lambda \ln \lambda} \ln\!\Big( \frac{3+\lambda}{4}\Big) \right],   \label{eq:pdelexp1}
\end{align}
where $(a)$ follows from \eref{eq:Nmindel0},  
and $(b)$ follows from the monotonicity of $f(x)$ and the assumption $0<\epsilon\leq1/2$. 
Hence,  to prove \eref{eq:piidapp2}, it suffices to show 
that the last expression in \eqref{eq:pdelexp1} is not larger than $\leq\frac{3}{2}\sqrt{\delta}$, that is,  
\begin{align}\label{eq:F16geq0}
	\left[\frac{1+\lambda}{\lambda \ln \lambda} \ln\!\Big( \frac{3+\lambda}{4}\Big) - \frac{1}{2} \right] \ln\delta^{-1}
	+ \frac{1+\lambda}{\lambda} \ln\!\Big( \frac{3+\lambda}{4}\Big) + \ln\!\Big( \frac{3}{2}\Big) 
	\geq 0 .
\end{align}
This inequality can be proved as follows, 
\begin{align}	
	\text{ LHS of \eqref{eq:F16geq0}}
	&\stackrel{(a)}{\geq} 
	\left[\frac{1+\lambda}{\lambda \ln \lambda} \ln\!\Big( \frac{3+\lambda}{4}\Big) - \frac{1}{2} \right] \ln\lambda^{-2}
	+ \frac{1+\lambda}{\lambda} \ln\!\Big( \frac{3+\lambda}{4}\Big) + \ln\!\Big( \frac{3}{2}\Big) \nonumber
	\\
	&=  \ln \lambda -\frac{\lambda+1}{\lambda}\ln\!\Big( \frac{3+\lambda}{4}\Big) + \ln \!\Big( \frac{3}{2}\Big)
	\stackrel{(b)}{\geq} 0, 
\end{align}
which completes the proof of Step 1. 
Here $(a)$ follows because $0<\delta\leq\lambda^2$ and 
$\frac{1+\lambda}{\lambda \ln \lambda} \ln\!\big( \frac{3+\lambda}{4}\big) - \frac{1}{2}\geq0$ by \lref{lemma:defgx} below,  
and $(b)$ follows from \lref{lem:NumeriGeq0} below.

\noindent \textbf{Step 2:} The aim of this step is to prove \eref{eq:piidapp2} in the case $\lambda^2<\delta\leq\lambda$. 

In this case we have $k_{-}= \left\lfloor\log _\lambda \delta\right\rfloor=1$ and 
$N_{\rm ZH}(\epsilon,\delta,\lambda)\geq (1-\nu\epsilon)/(\lambda \epsilon)$ according to \eref{eq:Nmindel0}. 
So we have 
\begin{align}	
	p_{\text{acc}}^{\rm ZH}
	&= \left( 1-\frac{\nu\epsilon}{2}\right) ^{N_{\rm ZH}(\epsilon,\delta,\lambda)}
	\leq  \exp\left[
	\frac{1-\nu\epsilon}{\lambda \epsilon } \ln\!\left( 1-\frac{\nu\epsilon}{2}\right) \right] 
	= \exp\left[ \frac{\nu}{\lambda} f\!\left( \frac{\nu\epsilon}{2}\right) \right]  
	\stackrel{(a)}{\leq}
	\exp\left[ \frac{\nu}{\lambda} f\!\left( \frac{\nu}{4}\right) \right]  \nonumber
	\\
	&= \exp\left[ \frac{1+\lambda}{\lambda} \ln\!\Big( \frac{3+\lambda}{4}\Big) \right]  
	\stackrel{(b)}{\leq} 
	\exp\left[ \ln\!\Big( \frac{3}{2}\Big) + \ln \lambda \right]  
	= \frac{3\lambda}{2} 
	\stackrel{(c)}{\leq} \frac{3\sqrt{\delta}}{2},  
\end{align}
which completes the proof of Step 2. 
Here $(a)$ follows from the monotonicity of $f(x)$ and the assumption $0<\epsilon\leq1/2$, 
$(b)$ follows from \lref{lem:NumeriGeq0} below, 
and $(c)$ follows from the assumption $\lambda^2<\delta\leq\lambda$.

\begin{lemma}\label{lemma:defgx}
	Suppose  $0< x\leq 1$; then
	\begin{align}
		2(1+x)\ln\!\Big( \frac{3+x}{4}\Big) - x\ln x\leq 0. \label{eq:defgx}
	\end{align}
\end{lemma}

\begin{proof}[Proof of \lref{lemma:defgx}]
	Denote the LHS in \eref{eq:defgx} by $g_1(x)$; then we have
	\begin{align}
		\frac{{\rm d }g_1(x)}{{\rm d} x} 
		&= 2\ln\!\Big(\frac{3+x}{4} \Big) -\ln x +\frac{x-1}{x+3}, \quad \frac{{\rm d }g_1(x)}{{\rm d} x}\bigg|_{x=1}=0, \\
		\frac{{\rm d ^2}g_1(x)}{{\rm d} x^2} 
		&= \frac{2}{x+3} - \frac{1}{x} + \frac{4}{(x+3)^2} 
		= \frac{x^2+4x-9}{x(x+3)^2} 
		< 0 
		\qquad \forall\, 0< x<1. 
	\end{align}
	It follows that $\frac{{\rm d }g_1(x)}{{\rm d} x}\geq 0$ for $0< x\leq 1$. 
	Therefore, $g_1(x)$ is nondecreasing in $x$ for $0< x\leq 1$, which confirms \eref{eq:defgx} given that 
	$g_1(1)=0$. 
\end{proof}	

\begin{lemma}\label{lem:NumeriGeq0}
	Suppose $0< x<1$; then
	\begin{align}
		-\frac{x+1}{x}\ln\!\Big( \frac{3+x}{4}\Big) + \ln x + \ln \!\Big( \frac{3}{2}\Big)\geq 0.  \label{eq:NumeriGeq0}
	\end{align}
\end{lemma}

\begin{proof}[Proof of \lref{lem:NumeriGeq0}]
	Denote the LHS in \eref{eq:NumeriGeq0} by $g_2(x)$; then we have
	\begin{align}
		\frac{{\rm d }g_2(x)}{{\rm d} x} 
		= \frac{1}{x^2} \ln\!\Big(\frac{3+x}{4} \Big) - \frac{x+1}{x(3+x)} + \frac{1}{x}
		= \frac{\hat{g}_2(x)}{x^2} ,
	\end{align}
	where 
	\begin{align}
		\hat{g}_2(x):= \frac{2x}{3+x} + \ln\!\Big(\frac{3+x}{4} \Big)  
	\end{align}
	is strictly increasing in $x$, and satisfies 
	\begin{align}
		\hat{g}_2(0.31)= \frac{62}{331} - \ln\!\Big(\frac{400}{331} \Big)  \approx -0.00203 <0 ,
		\qquad 
		\hat{g}_2(0.315)= \frac{42}{221} - \ln\!\Big(\frac{800}{663}\Big)   \approx 0.00221 >0. 
		\qquad 
	\end{align}
	Therefore, $g_2(x)$ is strictly decreasing in $x$ for $0<x\leq 0.310$, 
	but strictly increasing in $x$ for $0.315\leq x<1$. 
	In addition, for $0.310\leq x\leq0.315$ we have
	\begin{align}
		g_2(x) \geq 
		-\frac{0.315+1}{0.315}\ln\!\Big( \frac{3+0.315}{4}\Big) + \ln 0.31 + \ln \!\Big( \frac{3}{2}\Big)
		\approx 0.0184 >0. 
	\end{align}
	These facts together confirm \lref{lem:NumeriGeq0}. 
\end{proof}	

\subsection{Proof of Proposition~\ref{prop:TMMMF19passPrIIDUB}}\label{sec:probTMMMFproof}
Let $S_1,\dots,S_n$ be the $n$ stabilizer generators of  $|G\>$ as  given below \eref{eq:defGstate} in the main text. 
Let $P_i=\frac{1}{d}\sum_{j\in \bbZ_d} \!S_i^j$ 
be the projector onto the eigenspace of $S_i$ with eigenvalue 1; then $|G\>\< G|=\prod_{i=1}^n P_i$. Accordingly, we have 
\begin{align}\label{eq:TMMMFstrategy}
	\sum_{i=1}^n P_i\leq
	\prod_{i=1}^n P_i+(n-1) \openone= |G\>\< G|+(n-1) \openone.  
\end{align}
Since the $n$ projectors $P_1,P_2,\dots,P_n$ commute with each other,  the above inequality can be proved by considering the common eigenspaces of these projectors. 
If all $P_i$ have eigenvalue 1, then the inequality holds because $n\leq 1+(n-1)$; 
if  $0\leq t<n$ projectors have eigenvalue 1, then the inequality holds because $t\leq 0+(n-1)$.

In the TMMMF protocol, the $N_{\text {test}}$ systems to be tested are divided uniformly into $n$ groups, 
and Alice performs the stabilizer test $\{P_i, \openone-P_i\}$ on each system in the $i$th group, where $P_i$ corresponds to passing the test. 
If at most $\lfloor N_{\text {test}}/(2n^2)\rfloor$ failures are observed among the  $N_{\text {test}}$ tests, Alice accepts Bob's state. 
When Bob generates i.i.d.\! quantum states $\tau\in\mathcal{D}(\caH)$ with infidelity $\epsilon_\tau$, 
the average probability that $\tau$ passes a test is given by 
\begin{align}\label{eq:TMMMFmaxpassPr}
	\frac{1}{n} \sum_{i=1}^n \tr(P_i\tau) 
	= \frac{1}{n} \tr\left( \tau\sum_{i=1}^n P_i \right) 
	\leq \frac{1}{n} \tr\big\{ \tau [ |G\>\< G|+(n-1) \openone ] \big\}
	= \frac{\< G|\tau|G\>}{n} + \frac{n-1}{n}
	= 1-\frac{\epsilon_\tau}{n},  
\end{align}
where the inequality follows from \eref{eq:TMMMFstrategy} above. 

For $j=1,2,\dots,N_{\text {test}}$, the outcome of the $j$th test  
can be associated with a $\{0,1\}$-valued variable $W_j$, where 0 corresponds to passing the test and 1 corresponds to failure.
Then \eref{eq:TMMMFmaxpassPr} implies that $\mathbb{E}\big[\sum_{j} W_j\big]/N_{\text {test}} \geq \epsilon_\tau/n$. 
The probability that Alice accepts reads 
\begin{align}
	p_{\text{acc}}^{\rm TMMMF}=\Pr\bigg(\sum_{j} W_j \leq \left\lfloor \frac{N_{\text {test}}}{2n^2} \right\rfloor \bigg). 
\end{align}
According to the Hoeffding's theorem on Poisson binomial distribution
(see Theorem 4 in Ref.~\cite{Hoeffd56}), we have
\begin{align}
	\Pr\bigg(\sum_{j} W_j \leq \left\lfloor \frac{N_{\text {test}}}{2n^2} \right\rfloor \bigg)
	\leq B_{N_{\text {test}},\left\lfloor \frac{N_{\text {test}}}{2n^2} \right\rfloor}
	\Bigg(  \frac{\mathbb{E}\big[\sum_{j} W_j\big]}{N_{\text {test}}} \Bigg) 
	\leq B_{N_{\text {test}},\left\lfloor \frac{N_{\text {test}}}{2n^2} \right\rfloor}\!\left( \frac{\epsilon_\tau}{n} \right) , 
\end{align}
where the second inequality follows from \lref{lem:Bzkmono}. Therefore, 
the probability of acceptance satisfies 
\begin{align}
	p_{\text{acc}}^{\rm TMMMF}
	&\leq B_{N_{\text {test}},\left\lfloor \frac{N_{\text {test}}}{2n^2} \right\rfloor}\!\left( \frac{\epsilon_\tau}{n} \right) 
	\stackrel{(a)}{\leq} \exp \! \left[-N_{\text {test}} D\!\left( \frac{\lfloor N_{\text {test}}/(2n^2)\rfloor}{N_{\text {test}}} \bigg\|                
	\frac{\epsilon_\tau}{n} \right)  \right] 
	\nonumber\\
	&\stackrel{(b)}{\leq} \exp \! \left(- \frac{1.569}{ n^2} N_{\text {test}} \right)
	= \exp \! \left( -\frac{1.569}{ n^2} n\left\lceil \frac{5 n^{4} \ln n }{32}\right\rceil \right)
	\leq \exp \!\left(-0.245\, n^3 \ln n \right), 
\end{align}
which confirms Proposition~\ref{prop:TMMMF19passPrIIDUB}. 
Here $(a)$ follows from the Chernoff bound \eqref{eq:ChernoffB}, 
and $(b)$ follows from the following relation, 
\begin{align}\label{eq:F28geq}
	D\!\left( \frac{\lfloor N_{\text {test}}/(2n^2)\rfloor}{N_{\text {test}}} \bigg\| \frac{\epsilon_\tau}{n} \right)
	\stackrel{(a)}{\geq}
	D\!\left( \frac{1}{2n^2} \bigg\| \frac{\epsilon_\tau}{n} \right)
	\stackrel{(b)}{\geq} 
	D\!\left( \frac{1}{2n^2} \bigg\| \frac{8.155}{2n^2} \right)
	\stackrel{(c)}{\geq} 
	\frac{1.569}{ n^2}. 
\end{align}
In \eref{eq:F28geq}, 
$(a)$ follows from \lref{lem:DpqMonoton}; 
$(b)$ follows from \lref{lem:DpqMonoton} and the relation  
\begin{align}
	\epsilon_\tau=\frac{\epsilon}{2}=\frac{2 \sqrt{c}+1}{2n} \geq \frac{2 \sqrt{64/5}+1}{2n} \geq \frac{8.155}{2n};  
\end{align}
and $(c)$ follows from the inequality $D(x\|y)\geq(x-y)^2/(2y)$ for $x\leq y$.

\section{Conversion of measurement noise}
In the main text, we proposed a robust and efficient  protocol for verifying the resource graph state $|G\>\in\caH$ with ideal local projective measurements.
To apply our protocol to verifying blind MBQC in a realistic setting, we need to  deal with the unknown noise processes 
that afflict the measurements of Alice.

Here we consider a noise model in which any actual measurement performed by Alice can be expressed as a composition of 
a noise process and the noiseless measurement. 
Suppose $\mathcal{M}_j=\big\{M_l^j\big\}_l$ is the measurement that Alice intends to perform on the space $\caH$.  
Then the actual noisy measurement that she performs has the form $\mathcal{M}'_j=\big\{{\cal E}^{\dag}_j(M_l^j)\big\}_l$, 
where ${\cal E}_j$ is a CPTP map that encodes the noise process
associated with $\mathcal{M}_j$, 
and ${\cal E}_j^{\dag}$ is the adjoint map of ${\cal E}_j$.  
To eliminate the impact of measurement noise on the verification and MBQC, one method is 
to convert the measurement noise to preparation noise of the state $\rho$ on $\caH^{\otimes(N+1)}$ prepared by Bob. 
To apply this method, we need to further assume that the noise process ${\cal E}_j$ is independent  of the specific measurement,  
that is, ${\cal E}_j=\cal E$ for all $j$. 

Under the above noise assumption, suppose Alice runs the verification protocol developed in the main text 
with her noisy measurement devices. 
In the protocol, for $i=1,2,\dots,N$, Alice performs a measurement $\mathcal{M}'_{j_i}$ on the $i$th system of $\rho$ for state verification; 
once accepting, she then performs a measurement $\mathcal{M}'_{j_{N+1}}$ on the last system for MBQC. 
Note that this procedure is equivalent to the following procedure: The initial state is ${\cal E}^{\otimes (N+1)}(\rho)$ instead of $\rho$; 
for $i=1,2,\dots,N$, Alice performs the noiseless measurement $\mathcal{M}_{j_i}$ on the $i$th system of ${\cal E}^{\otimes (N+1)}(\rho)$; 
once accepting, she then performs the noiseless measurement $\mathcal{M}_{j_{N+1}}$ on the remaining system for MBQC. 
Therefore, the  verification of $\rho$ with noisy measurements is equivalent to the  verification of ${\cal E}^{\otimes (N+1)}(\rho)$ with noiseless measurements. 
In this sense, the noise in measurements has been successfully converted to noise in the state $\rho$. 
With noisy measurements, if Alice accepts the state prepared by Bob, then she
can guarantee (with significance level $\delta$) that the reduced state $\sigma_{N+1}$ on the remaining system satisfies 
\begin{align}
	\<G|{\cal E}(\sigma_{N+1})|G\>\geq 1-\bar{\epsilon}_{\lambda}(k,N,\delta).
\end{align} 
Consequently, according to the relation between the fidelity and trace norm, Alice can ensure  the condition 
\begin{align}
	\left|\operatorname{tr}\!\left[ {\cal E}^{\dag}(E) \sigma_{N+1}\right] -\langle G|E| G\rangle\right| 
	\leq 
	\sqrt{\bar{\epsilon}_\lambda(k, N, \delta)}
\end{align}
for any POVM element $0\leq E\leq \openone$; that is, when performing a noisy measurement on $\sigma_{N+1}$, 
the deviation of any outcome probability from the ideal value is not larger than $\sqrt{\bar{\epsilon}_{\lambda}(k,N,\delta)}$.

When the noise processes ${\cal E}_j$ depend on the measurements 
and the concrete forms of ${\cal E}_j$ are not known, the above conversion method does not work, 
especially when some local measurements used in MBQC are different from those used in state verification. 
To see this, note that the verification of $\rho$ with noisy measurements can be viewed as the verification of
$\left( {\cal E}_{j_1}\otimes\dots\otimes{\cal E}_{j_{N+1}}\right)\!(\rho)$ with noiseless measurements. 
In this case, even if Alice accepts, the fidelity of the remaining system cannot be guaranteed 
since it is subject to a noise process ${\cal E}_{j_{N+1}}$ which might be different from the noise processes 
in the state verification procedure. 
To solve this problem, we need to use the same set of local measurements in MBQC and state verification. However, universal quantum computation cannot be achieved using graph states together with  Pauli measurements. 
By contrast, \rcite{TMH19} proposed a special qubit hypergraph state that can realize MBQC  with Pauli $X$ and $Z$ measurements. 
In addition,  this state can be verified with Pauli $X$ and $Z$ measurements. 
Although the same set of measurements are used in MBQC and state verification, there are still two problems to be addressed.  
First, it is not easy to construct a homogeneous strategy for this hypergraph state  with $X$ and $Z$ measurements.
So many results in this work cannot be applied directly since they are based on homogeneous strategies. 
Second, the noise conversion method summarized above cannot be applied if  $X$ and $Z$ measurements have different noises.

The first problem shows the importance of extending our work to the non-homogeneous case, which deserves further study. 
The second problem can be resolved when the noise in measurements has a classical form, which means each actual $X$ ($Z$) measurement can be regarded as the application of some classical noise (i.e., the application of a $2 \times 2$ probability transition matrix) to the outcomes of the noiseless $X$ ($Z$) measurement. 
In this case, the noise effect on a single system of $\rho$ is determined by the ratio of $X$ and $Z$ measurements used in the test or MBQC. 
To eliminate the impact of measurement noise, we can use two strategies $\Omega_A$ and $\Omega_B$ for testing the hypergraph state, which are applied to different systems of $\rho$.   
Here $\Omega_A$ has a larger ratio of $X$ measurements and $\Omega_B$ has a smaller ratio of $X$ measurements. 
When 
the MBQC performed on the last system has an intermediate ratio of $X$ measurements between the above two ratios and
the error rates of both strategies are smaller than a threshold to verify 
the hypergraph state, we accept it.
In this case, the noise effect is also an intermediate value between the noise effects of the two  strategies $\Omega_A$ and $\Omega_B$.
Therefore, we can conclude that the sum of the classical noise in MBQC and the error in state preparation
is smaller than a certain threshold, and can guarantee the precision of the computation result on the last system. 
In this way, the  conversion method is applicable under the assumption of classical noise. 
When this assumption does not hold, further study is required to devise a noise conversion method.

\section{Comparison with DI QSC \cite{Aleks22}}\label{app:comDIQSV}
Recently, by combining the idea of QSV and self-testing \cite{SupicBow20}, Go\v{c}anin, \v{S}upi\'{c}, and Daki\'{c} \cite{Aleks22} developed a 
systematic approach for verifying  quantum states in a device-independent (DI) way, and they call their approach DI quantum
state certification (DI QSC). 
In the framework of DI QSC, there is an untrusted quantum source which is supposed to produce the pure state $|\Psi\>\in\caH$, 
but may actually produce independent states $\sigma_1,\sigma_2,\dots,\sigma_N$ in $N$ runs.
Meanwhile, the verifier Alice has an untrusted measurement device. 
By measuring a fragment of the produced states, her goal is to verify whether the rest of the copies 
are close enough to the target state $|\Psi\>$ on average. 
In this DI setting,  the closeness between a state $\sigma\in\mathcal{D}(\caH)$ and the target state $|\Psi\>$ is quantified by the extractability, which is defined as 
\begin{equation}
	\Xi(\sigma, \Psi) := \max_{\Lambda} F \big( \Lambda (\sigma), |\Psi\> \big), 
\end{equation}
where $F(\Lambda (\sigma),|\Psi\>)$ is the fidelity between $\Lambda (\sigma)$ and $|\Psi\>$, and the maximization is taken over all local isometries. 

Robust self-testing is a key for constructing DI QSC protocols. 
Suppose a Bell inequality $\mathcal{B}$ is used to self-test the target state $|\Psi\>$. 
Then $\mathcal{B}$ can be converted into a nonlocal game \cite{Silman08}, in which Alice performs a random test 
from a set of tests based on measuring local observables. Robust self-testing relies on
certain relation between the extractability and the average probability $p$ of successfully wining  the game. 
To be concrete, if a state $\sigma$ has extractability $1-\epsilon$, then its success probability $p$ cannot be larger than $p_{\rm QM}-\tilde c\,\epsilon$ \cite{Aleks22}. Here 
$p_{\rm QM}$ is the maximal success probability of the nonlocal game, which can only be reached by the target state $|\Psi\>$ (up to local isometries), 	
and $\tilde c$ is a constant depending on the Bell inequality $\mathcal{B}$.

In the DI QSC approach \cite{Aleks22}, 
Alice first randomly chooses $N_1\approx \mu N$ copies from all $N$ produced states;   
then she tests each of the chosen states with the nonlocal game mentioned above, and counts the number of successful rounds $q_1$. 
The idea of DI QSC can be intuitively understood in the limit of large $N_1$. 
In this case, the frequency $q_1/N_1$ of successes is close to the real (average) success probability $p$.   
In addition, as discussed above, a certain value of $p$ can in turn 
guarantee a lower bound on the (average) extractability of the $N_1$ tested states. 
Since the $N_1$ tested states are chosen at random, this can further imply a lower bound on 
the average extractability of the remaining unmeasured states. 
To insure that the average extractability of the remaining states is larger than $1-\epsilon$ with significance level $\delta$,
the sample cost of DI QSC is given by \cite{Aleks22}
\begin{align}
	N = \begin{cases}
		O\! \left( \frac{\ln\delta^{-1}}{\mu(1-\mu) c\epsilon} \right)   &\text{when $\beta_{Q}=\beta_{A}$},  \\  
		\vspace{-1em} \\
		O\! \left( \frac{\ln\delta^{-1}}{\mu(1-\mu)^2 c^2\epsilon^2} \right)  &\text{when $\beta_{Q}\ne\beta_{A}$}, 
	\end{cases}
\end{align}
where $\beta_{Q}$ and $\beta_{A}$ are the quantum bound and algebraic bound of the underlying Bell inequality $\mathcal{B}$, respectively, and $c$ is a constant depending on $\mathcal{B}$.

When the quantum bound of the Bell inequality $\mathcal{B}$ coincides with the algebraic bound, the scaling behaviors 
of the sample complexity with respect to $\delta$ and $\epsilon$ for DI QSC are the same as the counterpart of our verification protocol. 
However, this coincidence is quite rare. 
In a generic case, the coincidence does not hold, and 
the scaling behavior of $N$ with respect to $\epsilon$ is $O(\epsilon^{-2})$, which is quadratically worse. 
For general graph states studied in this work, no Bell inequality with such a coincidence has been proposed yet,  so the efficiency of DI QSC  is suboptimal compared with our protocol.

Although Ref.~\cite{Aleks22} established  a general framework of DI QSC, it is still not easy to construct a robust and efficient DI QSC protocol for a specific quantum state. Notably, to construct a concrete protocol for a state $|\Psi\>$, 
it is crucial not only to construct  a robust	self-testing protocol for $|\Psi\>$, but also to determine the values of several key parameters such as $p_{\rm QM}$ and $c$, which are usually quite challenging. For graph states considered in this work, a robust self-testing protocol was developed recently \cite{Baccari20}. 
However, this protocol only applies to  qubit graph states; in addition, it is not easy to calculate or give a nontrivial bound for $c$, since the derivation relies on the proofs of complicated operator inequalities \cite{Baccari20}. 
As a consequence, no concrete DI QSC protocol for general qudit graph states are known so far, and the scaling behavior of the sample cost $N$ with respect to the qudit number is not clear.

There are several similarities between our verification approach in the adversarial scenario and the DI QSC approach. 
First, after receiving states produced by the untrusted source, 
both approaches test some systems and provide a certificate on the remaining system(s) for application. 
This capability is powerful and crucial especially when the state preparation is controlled by a potentially malicious adversary. 
Second, both approaches are able to deal with the case in which not all tests are successfully passed. As discussed in the main text, this ability is crucial to achieve  robustness against noise and is appealing to 
practical applications.

The DI QSC approach has different assumptions from our protocol working in the adversarial scenario. 
On the one hand, its assumption on the devices is weaker than ours.  
In our protocol, the source is untrusted and the measurement device is trusted; 
while in DI QSC, the source and the measurement device are both untrusted. 
On the other hand, DI QSC has a stronger assumption on the source than our protocol.  
It assumes that the states produced by the source in different runs are mutually independent,
while our protocol applies to a more general source which may produce arbitrary correlated or entangled states.\footnote{
	Reference~\cite{Aleks22} also proposed a DI QSV approach for handling the general source. However, this approach can only make a conclusion about the states consumed in tests, which cannot be reused in quantum information processing tasks like MBQC. Due to this reason, we do not discuss this approach in this section.}
By combing DI QSC with our protocol, it might be possible to construct a robust and efficient 
verification protocol that applies to a general  untrusted source and an untrusted measurement device. We leave this line of research to future work.

\section{Proofs of results in the main text}\label{app:Proofmain}
In this section, we prove several results presented in the main text, 
including Eqs.~\eqref{eq:PLMstrategy} and \eqref{eq:pkLimit}, Theorems~\ref{thm:Boundeps}, \ref{thm:UBtestsNumber}, \ref{thm:iidHighProb}, 
and Proposition~\ref{prop:kbounds}.

\subsection{Proof of the second equality of \eref{eq:PLMstrategy} in the main text}\label{app:ProofEq5}
When the local dimension $d$ is a prime, the sum of all $P_{\mathbf{k}}$ can be expressed as
\begin{align}\label{proveEq5}
	\sum_{\mathbf{k} \in \mathbb{Z}_{d}^{n}} P_{\mathbf{k}}
	&\stackrel{(a)}{=} \frac{1}{d} \sum_{\mathbf{k} \in \mathbb{Z}_{d}^{n}} \sum_{j\in \bbZ_d} g_{\mathbf{k}}^j
	=\frac{1}{d} \left( \sum_{\mathbf{k} \in \mathbb{Z}_{d}^{n}}  g_{\mathbf{k}}^0 + \sum_{j=1}^{d-1} \sum_{\mathbf{k} \in \mathbb{Z}_{d}^{n}}  g_{\mathbf{k}}^j \right) 
	\stackrel{(b)}{=} \frac{1}{d} \left(  d^n \openone + \sum_{j=1}^{d-1} d^n |G\>\<G| \right) 
	\nonumber\\
	&=d^n \left[ |G\>\< G|+\frac{1}{d}(\openone-|G\>\<G|)\right] , 
\end{align}
which implies \eref{eq:PLMstrategy} in the main text. Here $(a)$ follows from \eref{eq:DefPk} in the main text, and $(b)$ follows from the following relation:  
\begin{align}\label{eq:sumgk}
	\frac{1}{d^n} \sum_{\mathbf{k} \in \mathbb{Z}_{d}^{n}}  g_{\mathbf{k}}^j 
	\stackrel{(c)}{=} \frac{1}{d^n} \sum_{\mathbf{k} \in \mathbb{Z}_{d}^{n}}  g_{\mathbf{k}}
	\stackrel{(d)}{=} |G\>\<G| \qquad \forall j\in\{1,2,\dots,d-1\},    
\end{align}
where $(c)$ holds because $\big\{ g_{\mathbf{k}}^j \big\}_{\mathbf{k}}=\big\{ g_{\mathbf{k}} \big\}_{\mathbf{k}}$,
and $(d)$ holds because $\sum_{\mathbf{k}\in\mathbb{Z}_{d}^{n}}g_{\mathbf{k}}/d^n$ is the stabilizer projector associated with the stabilizer group $S$.

It worth pointing out that the equality $(d)$ of \eref{eq:sumgk} holds for any local dimension $d$, but 
the equality $(c)$ of \eref{eq:sumgk} and the second equality of \eref{eq:PLMstrategy} in the main text may fail if $d$ is not a prime.  
For example, for the single-vertex graph state $|G\>$ with $n=1$ and $d=4$, we have 
\begin{align}
	|G\>=\frac{1}{2} \left( |0\> +|1\>+|2\>+|3\>\right) , \qquad 
	\bigl\{ g_{\mathbf{k}}^j\bigr\}_{\mathbf{k}} = 
	\begin{cases}
		\left\{ \identity, X, X^2, X^3 \right\}   &j=1,  \\
		\left\{ \identity, X^2, \identity, X^2 \right\}   &j=2,  \\
		\left\{ \identity, X^3, X^2, X \right\}   &j=3. 
	\end{cases}
\end{align}
It turns out that 
\begin{align}
	\frac{1}{d^{n}}\sum_{\mathbf{k} \in \mathbb{Z}_{d}^{n}} P_{\mathbf{k}}
	&= \frac{1}{d^{n}\cdot d} \sum_{j\in \bbZ_d} \sum_{\mathbf{k} \in \mathbb{Z}_{d}^{n}}  g_{\mathbf{k}}^j
	= \frac{1}{8}(4\identity+X+2X^2+X^3)
	\ne |G\>\< G|+\frac{1}{4}(\openone-|G\>\<G|).   
\end{align}
So the second equality of \eref{eq:PLMstrategy} in the main text does not hold in this case.

Our verification protocol proposed in the main text works for graph states with prime local dimensions, which is in line with the convention of most
literature on graph states. Much less is known about graph states when the local dimension $d$ is not a prime. In particular, no homogeneous strategy for graph states is known so far, and many results in this work may not hold since they are based on homogeneous strategies. 
To remedy this problem, one may either try to construct homogeneous strategies for general graph states or try to extend our work to non-homogeneous cases. Both lines of researches deserve further exploration in the future.

\subsection{Proof of Theorem \ref{thm:Boundeps}}\label{app:ProofThm1}
First, \lref{lem:LLP} implies that 
\begin{align}
	\bar{\epsilon}_{\lambda}(\lfloor\nu s N\rfloor, N, \delta)
	> \frac{\lfloor\nu s N\rfloor}{\nu N}
	\geq s- \frac{1}{\nu N}, 
\end{align}
which confirms the lower bound in Theorem~\ref{thm:Boundeps}.	

In addition, we have
\begin{align}
	\bar{\epsilon}_{\lambda}(\lfloor\nu s N\rfloor,N,\delta)
	&\stackrel{(a)}{\leq}
	\frac{z^*(\lfloor\nu s N\rfloor,\delta,\lambda) -\lfloor\nu s N\rfloor+1+\sqrt{\lambda \lfloor\nu s N\rfloor}}{\lambda N } 
	\nonumber\\
	&\stackrel{(b)}{\leq}
	\frac{1}{\lambda N }\left[ \bigg(\frac{\lfloor\nu s N\rfloor}{\nu}+\frac{\sqrt{2\nu \lfloor\nu s N\rfloor\ln\delta^{-1}}}{2\nu^2}+\frac{\ln\delta^{-1}}{2\nu^2}+1\bigg) -\lfloor\nu s N\rfloor+1+\sqrt{\lambda \lfloor\nu s N\rfloor}\right]
	\nonumber	\\
	&\leq
	\frac{1}{\lambda N }\left[\lambda sN + \bigg(\frac{\sqrt{2 \ln\delta^{-1}}}{2\nu}+\sqrt{\lambda \nu }\bigg)\sqrt{s N}  +\frac{\ln\delta^{-1}}{2\nu^2}+2\right]
	\nonumber	\\
	&\stackrel{(c)}{\leq}
	s+ \frac{1}{\nu \lambda}
	\sqrt{\frac{s \ln\delta^{-1}}{N}} + \frac{\ln\delta^{-1}}{2\nu^2\lambda N}+\frac{2}{\lambda N}, 
\end{align}
which confirms the upper bound in Theorem~\ref{thm:Boundeps}. 
Here $(a)$ follows from \lref{lem:zeta2Bound2}, 
$(b)$ follows from \lref{lem:z*UB3}, and $(c)$ follows from \lref{lem:1.098}.

\subsection{Proof of Theorem \ref{thm:UBtestsNumber}}\label{app:ProofThm2}

By assumption we have $0<\lambda<1$, $0\leq s<\epsilon<1$,  $0<\delta\leq1/2$, and  $N\in \bbZ^{\geq1}$. Therefore,
\begin{align}
	&\bar{\epsilon}_{\lambda}(\lfloor\nu s N\rfloor,N,\delta) \leq \epsilon , 
	\nonumber\\   \stackrel{(a)}{\Leftarrow} \quad
	&\frac{z^*(\lfloor\nu s N\rfloor,\delta,\lambda)-\lfloor\nu s N\rfloor+1+\sqrt{\lambda\lfloor\nu s N\rfloor}}{\lambda N} \leq \epsilon ,
	\nonumber\\   \stackrel{(b)}{\Leftarrow} \quad
	&\left(\frac{\lfloor\nu s N\rfloor}{\nu}+\frac{\sqrt{2\nu \lfloor\nu s N\rfloor\ln\delta^{-1}}}{2\nu^2}+\frac{\ln\delta^{-1}}{2\nu^2}+1 \right) -\lfloor\nu s N\rfloor+1+\sqrt{\lambda\lfloor\nu s N\rfloor} \leq \epsilon\lambda N,
	\nonumber\\   \Leftarrow\quad
	&\left(\frac{\nu s N}{\nu}+\frac{\sqrt{2\nu (\nu s N) \ln\delta^{-1}}}{2\nu^2}+\frac{\ln\delta^{-1}}{2\nu^2}+1 \right)
	-\nu s N+1+\sqrt{\lambda\nu s N} \leq \epsilon\lambda N,
	\nonumber\\   \Leftarrow\quad
	&\lambda (\epsilon-s)N - \left(\frac{\sqrt{2 s \ln\delta^{-1}}}{2\nu}+\sqrt{\lambda \nu s} \right) \sqrt{N} - \left(\frac{\ln\delta^{-1}}{2\nu^2}+2\right)
	\geq 0, \label{eq:T2solve}
	\\
	\stackrel{(c)}{\Leftrightarrow} \quad
	&N\geq \left\lceil \frac{1}{\left[ 2\lambda (\epsilon-s) \right]^2}
	\left[\sqrt{\left(\frac{\sqrt{2 s \ln\delta^{-1}}}{2\nu}+\sqrt{\lambda \nu s} \right)^2
		+4\lambda (\epsilon-s) \left(\frac{\ln\delta^{-1}}{2\nu^2}+2\right)}
	+\left(\frac{\sqrt{2 s \ln\delta^{-1}}}{2\nu}+\sqrt{\lambda \nu s} \right)\right]^2 \right\rceil,
	\nonumber\\    \stackrel{(d)}{\Leftarrow} \quad
	&N\geq \left\lceil \frac{1}{\left[ \lambda (\epsilon-s) \right]^2}
	\left[ \left(\frac{\sqrt{2 s \ln\delta^{-1}}}{2\nu}+\sqrt{\lambda \nu s} \right)^2
	+2\lambda (\epsilon-s) \left(\frac{\ln\delta^{-1}}{2\nu^2}+2\right)\right] \right\rceil,
	\nonumber\\    \stackrel{(e)}{\Leftarrow} \quad
	&N\geq
	\left\lceil \frac{1}{\left[ \lambda (\epsilon-s) \right]^2}
	\left[ \frac{s \ln\delta^{-1}}{\nu^2}+2\lambda \nu s
	+\lambda (\epsilon-s) \left(\frac{\ln\delta^{-1}}{\nu^2}+4\right)\right] \right\rceil,
	\nonumber\\    \Leftrightarrow\quad
	&N\geq
	\left\lceil \frac{1}{\left[ \lambda (\epsilon-s) \right]^2}
	\left( [s+\lambda (\epsilon-s)] \frac{ \ln\delta^{-1}}{\nu^2}+ [2\lambda \nu s+4\lambda (\epsilon-s)]
	\right) \right\rceil,
	\nonumber\\    \stackrel{(f)}{\Leftarrow} \quad
	&N\geq
	\frac{1}{\left[ \lambda (\epsilon-s) \right]^2}
	\left( \frac{ \epsilon \ln\delta^{-1}}{\nu^2}+ 4\lambda\epsilon
	\right)
	\quad \Leftarrow \quad
	N\geq
	\frac{\epsilon}{\left[\lambda \nu (\epsilon-s) \right]^2}
	\left( \ln\delta^{-1}+ 4\lambda\nu^2
	\right) , \nonumber
\end{align}
which confirm Theorem~\ref{thm:UBtestsNumber}. 
Here the relation $A\Leftarrow B$ means  $B$ is a sufficient condition of $A$; 
$(a)$ follows from \lref{lem:zeta2Bound2}; 
$(b)$ follows from \lref{lem:z*UB3}; 
$(c)$ is obtained by directly solving \eref{eq:T2solve};
both $(d)$ and $(e)$  follow from the inequality $(x+y)^2\leq 2x^2+2y^2$; and 
$(f)$ follows from the assumption $s<\epsilon$.

\subsection{Proof of Theorem \ref{thm:iidHighProb}}\label{app:ProofThm3}

Let 
\begin{align}
	s= \frac{\lambda\sqrt{2\nu}+r}{\lambda\sqrt{2\nu}+1} \,\epsilon, 
	\qquad k= \left\lfloor \nu s N\right\rfloor,
	\qquad
	N\geq \frac{(\ln\delta^{-1}+  4\lambda\nu^2)\epsilon}{\left[\lambda \nu (\epsilon-s) \right]^2}
	= \bigg[ \frac{\lambda\sqrt{2\nu}+1}{\lambda \nu (1-r)} \bigg]^2 \;
	\frac{\ln\delta^{-1} +4\lambda\nu^2}{ \epsilon};
\end{align}
then $\bar{\epsilon}_{\lambda}(k, N, \delta)\leq\epsilon$ by Theorem~\ref{thm:UBtestsNumber},
which confirms the first inequality of \eref{eq:robustCondition} in the main text. 

To complete the proof, it remains to prove the second inequality of \eref{eq:robustCondition} in the main text, that is, $B_{N,k}(\nu r\epsilon)\geq1-\delta$.
When  $r=0$, this inequality is obvious. 
When $r>0$, we have
\begin{align}
	B_{N,k}(\nu r\epsilon)
	&=1-B_{N,N-k-1}(1-\nu r\epsilon)
	\stackrel{(a)}{\geq} 1-\exp\left[-N D\!\left( \frac{N-k-1}{N}\bigg\| 1-\nu r\epsilon \right) \right] 
	\nonumber\\
	&\stackrel{(b)}{\geq} 1-\rme^{-N D\left(1- \nu s \|1-\nu r\epsilon \right) }
	\stackrel{(c)}{=}    1-\rme^{-N D\left(\nu s \| \nu r\epsilon \right) },  \label{eq:44TH3}
\end{align}
where $(a)$ follows from the Chernoff bound \eqref{eq:ChernoffB} and the inequality $N\nu r\epsilon\leq k+1$;
$(b)$ follows from \lref{lem:DpqMonoton} and the inequality $(k+1)/N\geq \nu s$;
$(c)$ follows from the relation $D(x\|y)=D(1-x\|1-y)$.
In addition, we have
\begin{align}
	N D(\nu s \| \nu r\epsilon )
	&\geq \frac{\epsilon\ln\delta^{-1}}{\left[\lambda \nu (\epsilon-s) \right]^2} D(\nu s \| \nu r\epsilon )
	\geq \frac{\epsilon\ln\delta^{-1}}{\left[\lambda \nu (\epsilon-s) \right]^2} \frac{(\nu s- \nu r\epsilon)^2}{2\nu \epsilon}
	= \frac{(s-r\epsilon)^2}{2\nu \lambda^2 (\epsilon-s)^2} \ln\delta^{-1}
	\nonumber\\
	&= \frac{\Big(\frac{\lambda\sqrt{2\nu}+r}{\lambda\sqrt{2\nu}+1}\,\epsilon-r\epsilon\Big)^2}{2\nu \lambda^2 \Big(\epsilon-\frac{\lambda\sqrt{2\nu}+r}{\lambda\sqrt{2\nu}+1}\,\epsilon\Big)^2} \ln\delta^{-1}
	=\ln\delta^{-1},  \label{eq:45TH3}
\end{align}
where the second inequality follows from the assumption $s\leq \epsilon$ and the inequality $D(x\|y)\geq(x-y)^2/(2x)$ for $x\geq y$.
Equations~\eqref{eq:44TH3} and~\eqref{eq:45TH3} together confirm the inequality $B_{N,k}(\nu r\epsilon)\geq1-\delta$ and complete the proof of Theorem~\ref{thm:iidHighProb}.

\subsection{Proof of Proposition \ref{prop:kbounds}}\label{app:ProofProp2}
If $\nu \epsilon N\leq k\leq N-1$, then \lref{lem:LLP} implies that $\bar{\epsilon}_{\lambda}(k, N, \delta)>\epsilon$, 
which confirms the first statement in Proposition~\ref{prop:kbounds}.

For given $0<\lambda,\epsilon <1$, $0<\delta\leq1/4$, and $N,k\in\bbZ^{\geq 0}$, we have
\begin{align}
	\begin{split}
		&k \leq l(\lambda, N, \epsilon,\delta), 
		\\  \Leftrightarrow \quad
		&k
		\leq \bigg\lfloor\nu \epsilon N - \frac{\sqrt{\ln\delta^{-1}}}{\nu}\cdot \frac{\nu\sqrt{  N \epsilon} }{\lambda}
		-\frac{\ln\delta^{-1}}{2\lambda\nu}-\frac{2\nu}{\lambda} \bigg\rfloor,
		\\   \stackrel{(a)}{\Rightarrow} \quad
		&k
		\leq \nu \epsilon N - \bigg( \frac{\sqrt{2 \ln\delta^{-1}}}{2\nu}+\sqrt{\lambda \nu} \bigg) \frac{\nu\sqrt{  N \epsilon} }{\lambda}
		-\frac{\nu}{\lambda} \left(\frac{\ln\delta^{-1}}{2\nu^2}+2 \right),
		\\   \Rightarrow\quad
		&k\leq \nu \epsilon N,\ 
		\frac{\lambda}{\nu}k +\bigg( \frac{\sqrt{2\nu \ln\delta^{-1}}}{2\nu^2}+\sqrt{\lambda} \bigg) \sqrt{\nu \epsilon N} + \left(\frac{\ln\delta^{-1}}{2\nu^2}+2- \lambda \epsilon N \right)\leq 0,
		\\   \Rightarrow\quad
		&k\leq N-1,\ 
		\frac{\lambda}{\nu}k +\bigg( \frac{\sqrt{2\nu \ln\delta^{-1}}}{2\nu^2}+\sqrt{\lambda} \bigg) \sqrt{k} + \left(\frac{\ln\delta^{-1}}{2\nu^2}+2-  \lambda \epsilon N \right)\leq 0,
		\\   \Rightarrow\quad
		&k\leq N-1,\ 
		\bigg(\frac{k}{\nu}+\frac{\sqrt{2\nu k\ln\delta^{-1}}}{2\nu^2}+\frac{\ln\delta^{-1}}{2\nu^2}+1 \bigg) -k+1+\sqrt{\lambda k} 
		\leq \epsilon\lambda N,
		\\   \stackrel{(b)}{\Rightarrow}\quad
		&k\leq N-1,\ 
		\frac{z^*(k,\delta,\lambda)-k+1+\sqrt{\lambda k}}{\lambda N} \leq \epsilon 
		\quad  \stackrel{(c)}{\Rightarrow} \quad
		k\in\bbZ^{\geq 0},\ k\leq N-1,\ 
		\bar{\epsilon}_{\lambda}(k,N,\delta) \leq \epsilon , 
	\end{split}
\end{align}
which confirm the second statement in Proposition \ref{prop:kbounds}.
Here the notation $A\Rightarrow B$ means  $A$ is a sufficient condition of $B$; 
$(a)$ follows from \lref{lem:1.098}; 
$(b)$ follows from \lref{lem:z*UB3}; and
$(c)$ follows from \lref{lem:zeta2Bound2}.

\subsection{Proof of \eref{eq:pkLimit} in the main text}\label{app:ProofeqpkLimit}

In the limit $N\to+\infty$ we have
\begin{align}\label{eq:limkinftyNadv}
	l(\lambda, N, \epsilon,\delta) 
	= \nu \epsilon N + O( \sqrt{  N})
	> \nu \epsilon_{\tau} N,   
\end{align}
where the inequality follows from the assumption $0<\epsilon_\tau<\epsilon$.
Let $l$ be a shorthand for $l(\lambda, N, \epsilon,\delta)$. Then
\begin{align}
	p^{\iid}_{N,l}(\tau)
	&\stackrel{(a)}{=} B_{N,l}(\nu\epsilon_\tau)
	=1-B_{N,N-l-1}(1-\nu\epsilon_\tau)
	\stackrel{(b)}{\geq} 1-\exp\left[- D\!\left( \frac{N-l-1}{N}\bigg\| 1-\nu\epsilon_\tau \right) N \right]  \nonumber \\
	&\stackrel{(c)}{=}    1-\exp\left[- D\!\left( \frac{\nu \epsilon N + O( \sqrt{  N}\,)}{N}\bigg\| \nu\epsilon_\tau \right) N \right]
	=     1-\exp\left[- D\!\left( \nu\epsilon + O\Bigl(\frac{1}{\sqrt{N}}\Bigr) \bigg\| \nu\epsilon_\tau \right) N \right] \nonumber \\
	&  =     1-\exp\bigl[- D ( \nu\epsilon \| \nu\epsilon_\tau ) N +O(\sqrt{N}\,)\bigr],  \label{eq:limkinftyNadvLB}
\end{align}
where $(a)$ follows from \eref{eq:SuccessProbpk} in the main text;
$(b)$ follows from the Chernoff bound \eref{eq:ChernoffB} and \eref{eq:limkinftyNadv};
$(c)$ follows from \eref{eq:limkinftyNadv} and the relation $D(x\|y)=D(1-x\|1-y)$.

On the other hand,
\begin{align}
	p^{\iid}_{N,l}(\tau)
	&= B_{N,l}(\nu\epsilon_\tau)
	=1-B_{N,N-l-1}(1-\nu\epsilon_\tau)
	\stackrel{(a)}{\leq} 1-\frac{1}{\rme \sqrt{l}}\exp\left[- D\!\left( \frac{N-l-1}{N}\bigg\| 1-\nu\epsilon_\tau \right) N \right]  \nonumber \\
	&\stackrel{(b)}{=}   1-\exp\Bigl[-1-\frac{1}{2}\ln \big(\nu \epsilon N + O( \sqrt{  N}\,) \big) \Bigr] \exp\left[- D\!\left( \nu\epsilon + O\Bigl(\frac{1}{\sqrt{N}}\Bigr) \bigg\| \nu\epsilon_\tau \right) N \right] \nonumber \\
	&               =     1-\exp\bigl[- D ( \nu\epsilon \| \nu\epsilon_\tau ) N +O(\sqrt{N}\,)\bigr],  \label{eq:limkinftyNadvUB}
\end{align}
where $(a)$ follows from the reverse Chernoff bound in \lref{lem:ChernoffRev}, 
and $(b)$ follows from \eref{eq:limkinftyNadv}. 

Equations \eqref{eq:limkinftyNadvLB} and \eqref{eq:limkinftyNadvUB} together confirm \eref{eq:pkLimit} in the main text.

\section{Robust and efficient verification of quantum states in the adversarial scenario}\label{app:relationQSV}

In the main text, we proposed a robust and efficient protocol for verifying the resource graph state in blind MBQC. To achieve this goal we
proposed a robust and efficient protocol for verifying qudit graph states with a prime local dimension in the adversarial scenario. 
As mentioned in the Discussion section, our verification protocol can also be used to verify many 
other pure quantum states in the adversarial scenario. 
Suppose a homogeneous strategy as given in \eref{eq:homo} in the main text can be constructed for a given state  $|\Psi\>\in\caH$;
then our protocol and most related results are still applicable if the target graph state $|G\>$ is replaced by $|\Psi\>$.

\subsection{Better results when the spectral gap $\nu$ is small}
As shown in the main text, for any qudit graph state with a prime local dimension $d$
we can construct a homogeneous verification strategy with $0<\lambda\leq1/2$, that is, $1/2\leq \nu<1$.  
However, for a general multipartite state $|\Psi\>$, it is not always easy to construct a homogeneous strategy with a large spectral gap. 
In this case, the following propositions proved in \ref{sec:proofsmallnu}
can provide stronger performance guarantee than the results presented in the main text. 
To be specific, when $1/2\leq\lambda<1$, that is, $0<\nu\leq 1/2$, 
Proposition \ref{prop:BoundepsSmallnu} is stronger than Theorem \ref{thm:Boundeps};  
Proposition \ref{prop:UBtestsNumberSmallnu} is stronger than Theorem \ref{thm:UBtestsNumber}; 
Proposition \ref{prep:iidHighProbSmallnu} is stronger than Theorem \ref{thm:iidHighProb}; and 
Proposition~\ref{prep:kboundsSmallnu} is stronger than the second statement in Proposition \ref{prop:kbounds}.

\begin{proposition}\label{prop:BoundepsSmallnu}
	Suppose $1/2\leq\lambda<1$, $0<s<1$, $0<\delta\leq1/3$, and $N\in\bbZ^{\geq 1}$; then
	\begin{align}
		\bar{\epsilon}_{\lambda}(\lfloor\nu s N\rfloor, N, \delta)
		\leq s+ 2\sqrt{\frac{s \ln\delta^{-1}}{\nu\lambda N}} + \frac{2\ln\delta^{-1}}{\nu N}+\frac{2}{\lambda N}.
	\end{align}
\end{proposition}

\begin{proposition}\label{prop:UBtestsNumberSmallnu}
	Suppose  $1/2\leq\lambda<1$, $0\leq s<\epsilon<1$, and $0<\delta\leq1/2$.
	If the number of tests $N$ satisfies
	\begin{align}
		N\geq \frac{4\, \epsilon }{\lambda\nu (\epsilon-s)^2}
		\left(  \ln\delta^{-1} + \nu \right) ,
	\end{align}
	then we have $\bar{\epsilon}_{\lambda}(\lfloor\nu s N\rfloor,N, \delta)\leq\epsilon$.
\end{proposition}

\begin{proposition}\label{prep:iidHighProbSmallnu}
	Suppose $1/2\leq\lambda< 1$, $0<\delta\leq1/2$, $0<\epsilon<1$, and $0\leq r<1$. Then the conditions of soundness
	and robustness in \eref{eq:robustCondition} in the main text hold as long as
	\begin{align}
		k
		= \bigg\lfloor  \bigg(\frac{\sqrt{\lambda}+\sqrt{2}\,r}{\sqrt{\lambda}+\sqrt{2}}\bigg)\nu \epsilon N\bigg\rfloor, 
		\qquad
		N
		\geq
		\left\lceil 2\bigg[ \frac{\sqrt{2}+\sqrt{\lambda}}{\sqrt{\lambda}\,(1-r)} \bigg]^2\frac{ \ln\delta^{-1} + \nu}{ \nu \epsilon} \right\rceil.
	\end{align}
\end{proposition}

According to Proposition \ref{prep:iidHighProbSmallnu}, 
when the robustness $r$ is a constant, the number of tests 
required  to verify the target state $|\Psi\>$ in the adversarial scenario 
is only $O\big( \frac{\ln \delta^{-1}}{\nu\epsilon}\big)$.  
This scaling behavior is the same as the counterpart in the i.i.d. scenario as clarified in the Methods section.

\begin{proposition}\label{prep:kboundsSmallnu}
	Suppose $1/2\leq\lambda< 1$, $0<\epsilon <1$, $0<\delta\leq1/4$, and $N,k\in\bbZ^{\geq 0}$. If the number $k$ of allowed failures satisfies 
	\begin{align}
		0\leq k\leq 
		\bigg\lfloor\nu \epsilon N - 2\sqrt{\frac{\nu\epsilon N\ln\delta^{-1}}{\lambda}}
		-2\ln\delta^{-1}-\frac{2\nu}{\lambda} \bigg\rfloor,
	\end{align}
	then we have $\bar{\epsilon}_{\lambda}(k, N, \delta)\leq\epsilon$. 
\end{proposition}

\subsection{Verification with a fixed number of allowed failures}
In the main text, we have considered the verification of quantum states 
with a fixed error rate, in which  the number of allowed failures $k$ is proportional to the number of tests $N$. 
Here, we shall consider the case in which $k$ is a fixed integer. Then $\bar{\epsilon}_{\lambda}(k,N,\delta)=O(N^{-1})$ according to \lref{lem:zeta2Bound2}. As the number of tests $N$ increases, the guaranteed infidelity approaches $0$. 
To  study the minimum number of tests required to reach a given precision,  define 
\begin{align}\label{eq:DefNTConst}
	N_{k}(\epsilon,\delta,\lambda) := \min\{ N\geq k+1\,|\, \bar{\epsilon}_{\lambda}(k,N,\delta)\leq \epsilon \}, \quad k\in\bbZ^{\geq 0},\quad  0<\lambda,\delta,\epsilon<1.
\end{align}
If the number of tests $N$ satisfies $N\geq N_{k}(\epsilon,\delta,\lambda)$, then   $\bar{\epsilon}_{\lambda}(k,N,\delta)\leq \epsilon$ given that $\bar{\epsilon}_{\lambda}(k,N,\delta)$ is nonincreasing in $N$.

Next, we  provide several informative bounds for  $N_{k}(\epsilon,\delta,\lambda)$; 
see \ref{sec:prooffixkadv} for proofs. Recall that $z^*(k,\delta,\lambda)$ is the smallest integer $z$ that satisfies $B_{z,k}(\nu)\leq \delta$, as defined in \eref{eq:definez}, and $z_*(k,\delta,\lambda)=z^*(k,\delta,\lambda)-1$.

\begin{lemma}\label{lem:zstarBoundsN}
	Suppose $0<\lambda,\epsilon< 1$, $0<\delta\leq1/2$, and $k\in\bbZ^{\geq 0}$. Then 
	\begin{align}
		\frac{(1-\nu\epsilon)z_*-(1-\epsilon)k}{\lambda \epsilon} - 1 
		\leq N_{k}(\epsilon,\delta,\lambda)
		< \frac{z^*-k+1+\sqrt{\lambda k}}{\lambda \epsilon}, 
	\end{align}
	where $z^*$ and $z_*$ are shorthands for $z^*(k,\delta,\lambda)$ and $z_*(k,\delta,\lambda)$, respectively.
\end{lemma}

\begin{proposition}\label{prop:Nk}
	Suppose $0<\lambda,\epsilon< 1$, $0<\delta\leq1/2$, and $k\in\bbZ^{\geq 0}$. Then we have 
	\begin{align}
		\frac{k}{\nu\epsilon}	
		\leq 
		N_{k}(\epsilon,\delta,\lambda)
		&< \frac{k}{\nu\epsilon} + \frac{\ln\delta^{-1}}{2\nu^2\lambda \epsilon} 
		+ \frac{\sqrt{2\nu k\ln\delta^{-1}}}{2\nu^2\lambda \epsilon} + \frac{\sqrt{\lambda k}}{\lambda \epsilon} + \frac{2}{\lambda \epsilon}, 
		\label{eq:NuBinformCR} 
		\\
		N_{k}(\epsilon,\delta,\lambda)
		&\geq \frac{(1-\nu\epsilon)\ln\delta^{-1}}{\lambda\epsilon\ln\lambda^{-1}} - \frac{(1-\epsilon)(k+1)}{\lambda \epsilon} -2.
		\label{eq:NlBinformCR}
	\end{align}
	If in addition $1/2\leq\lambda<1$, then we have 
	\begin{align}\label{eq:NuBinformCRsmallnu} 
		N_{k}(\epsilon,\delta,\lambda)
		< \frac{k}{\nu\epsilon} + \frac{2\ln\delta^{-1}}{\nu \epsilon} 
		+ \frac{\sqrt{2\lambda k\ln\delta^{-1}}}{\lambda\nu \epsilon} + \frac{\sqrt{\lambda k}}{\lambda \epsilon} + \frac{2}{\lambda \epsilon}. 
	\end{align}
\end{proposition}

According to Proposition~\ref{prop:Nk}, for a fixed number of allowed failures, 
the sample cost of our protocol is only $O\big( \frac{\ln \delta^{-1}}{\nu\epsilon}\big)$, which has 
the optimal scaling behaviors with respect to $\nu$, $\epsilon$, and $\delta$  as one can expect. 
When $\lambda=1/2$ for example,  Proposition~\ref{prop:Nk} implies that 
\begin{align}
	N_{k}(\epsilon,\delta,\lambda=1/2)
	< \frac{2k+4\ln\delta^{-1}+4\sqrt{k\ln\delta^{-1}}+\sqrt{2k}+4}{\epsilon} .
\end{align}

Next, we consider the high-precision limit $\epsilon,\delta \to 0$. 
\begin{proposition}[Equation (89) and Lemma 5.4 in \rcite{Classical22}]\label{prop:MLT}
	Suppose $0<\lambda,\epsilon,\delta< 1$ and $k\in\bbZ^{\geq 0}$. Then
	\begin{align}
		\lim_{\delta\to 0} \frac{ N_{k}(\epsilon,\delta,\lambda)}{\ln \delta^{-1}}
		&=\frac{1-\nu\epsilon}{\lambda \epsilon \ln \lambda^{-1}},
		\\
		\lim_{\delta,\epsilon\to 0} \frac{\epsilon N_{k}(\epsilon,\delta,\lambda)}{\ln \delta^{-1}}
		&=\frac{1}{\lambda \ln \lambda^{-1}},\label{eq:XMP98}
	\end{align}
	where the order of the two limits $\epsilon\to 0$ and $\delta\to 0$ does not matter in \eref{eq:XMP98}.
\end{proposition}
This proposition follows from Lemma 5.4 in our companion paper \cite{Classical22} according to the discussions in the Methods section. 
By this proposition, in the 
high-precision limit $\epsilon,\delta \to 0$, the  efficiency of our protocol is determined by the factor $(\lambda\ln\lambda^{-1})^{-1}$, which attains its minimum $\rme$ when $\lambda=1/\rme$. 
Hence, the strategy with $\lambda=1/\rme$ can achieve the smallest sample cost for high-precision verification.

\subsection{Minimum number of tests for robust verification}\label{app:OptNadv}
For $0<\lambda,\delta,\epsilon<1$ and $0\leq r<1$, the  minimum number of tests required for  robust verification  in the adversarial scenario is denoted by  $N_{\rm min}(\epsilon,\delta,\lambda,r)$ as defined in \eref{eq:OptNadvDef} in the main text.  
Algorithm~\ref{alg:NoptAdv} proposed in the main text can be used to calculate $N_{\rm min}(\epsilon,\delta,\lambda,r)$ and 
the corresponding number of allowed failures.

To understand why Algorithm~\ref{alg:NoptAdv} works,
denote by $k_{\rm min}(\epsilon,\delta,\lambda,r)$ the minimum nonnegative integer $k$ such that \eref{eq:robustCondition} in the main text holds 
for some $N\geq k+1$.  
When $r=0$ for example, we have $k_{\rm min}(\epsilon,\delta,\lambda,r)=0$, given that $B_{N,0}(0)=1$ and $\lim_{N\to \infty }\bar{\epsilon}_{\lambda}(k=0, N, \delta)=0$ \cite{ZhuEVQPSshort19,ZhuEVQPSlong19}. To simplify the notation, we shall use $k_{\rm min}$ as  a shorthand for $k_{\rm min}(\epsilon,\delta,\lambda,r)$. 
Then $N_{\rm min}(\epsilon,\delta,\lambda,r)$ can be expressed as 
\begin{align}
	N_{\rm min}(\epsilon,\delta,\lambda,r)&\;=
	\min_{k\geq k_{\rm min}} N_{k}(\epsilon,\delta,\lambda,r), \label{eq:DefN_kB3}
	\\
	N_{k}(\epsilon,\delta,\lambda,r)
	&:= 
	\min \left\{N\in\bbZ^{k+1} \big|\, \bar{\epsilon}_{\lambda}(k, N, \delta)\leq\epsilon, B_{N,k}(\nu r\epsilon)\geq1-\delta \right\} 
	\nonumber\\
	&\;=
	\min \left\{N\in\bbZ^{k+1} \big|\, \bar{\epsilon}_{\lambda}(k, N, \delta)\leq\epsilon \right\}	
	\qquad \forall k\in\bbZ^{k_{\rm min}}. \label{eq:DefN_kF19}
\end{align}
The second equality in \eref{eq:DefN_kF19}  follows from the definition of $k_{\rm min}$ and the fact that $B_{N,k}(\nu r\epsilon)$ is nonincreasing in $N$ for $N\geq k$ by \lref{lem:Bzkmono} in \ref{app:UsefulLemma}. 
Note that $N_{k}(\epsilon,\delta,\lambda,r)$ reduces to $N_{k}(\epsilon,\delta,\lambda)$ defined in \eref{eq:DefNTConst} when $k\geq k_{\rm min}$. 
Recall that $\overline{\epsilon}_\lambda(k,N,\delta)$ is nondecreasing in $k$ for $0\leq k\leq N-1$ by Proposition~\ref{prop:epsMonoton}, so  $N_{k}(\epsilon,\delta,\lambda,r)$
is nondecreasing in $k$ for $k\geq k_{\rm min}$, which  implies that
\begin{align}\label{eq:optNRewriteB5}
	N_{\rm min}(\epsilon,\delta,\lambda,r)=N_{k_{\rm min}}(\epsilon,\delta,\lambda,r).
\end{align} 
By  definition in \eref{eq:DefN_kF19}, \eref{eq:robustCondition} in the main text holds when 
$N=\!N_{\rm min}(\epsilon,\delta,\lambda,r)$ and $k=k_{\rm min}(\epsilon,\delta,\lambda,r)$.  

In Algorithm~\ref{alg:NoptAdv}, the aim of steps 1--11 is to find $k_{\rm min}$. 
Notably, steps 1--2 aim to find $k_{\rm min}$ in the case $r=0$; 
steps 3--10 aim to find $k_{\rm min}$ in the case $r>0$ by virtue of  the following properties:   
(1) $B_{M,k}(\nu r\epsilon)$ is strictly decreasing in $M$ for $M\geq k$ by \lref{lem:Bzkmono}; and
(2) $\overline{\epsilon}_\lambda(k,M,\delta)$ is nonincreasing in $M$ for $M\geq k+1$ by Proposition \ref{prop:epsMonoton}.  
Steps 12--13 aim to find $N_{\rm min}$ by virtue of  Eqs.~\eqref{eq:DefN_kF19} and \eqref{eq:optNRewriteB5}. 

\begin{figure}[b]
	\includegraphics[width=7.4cm]{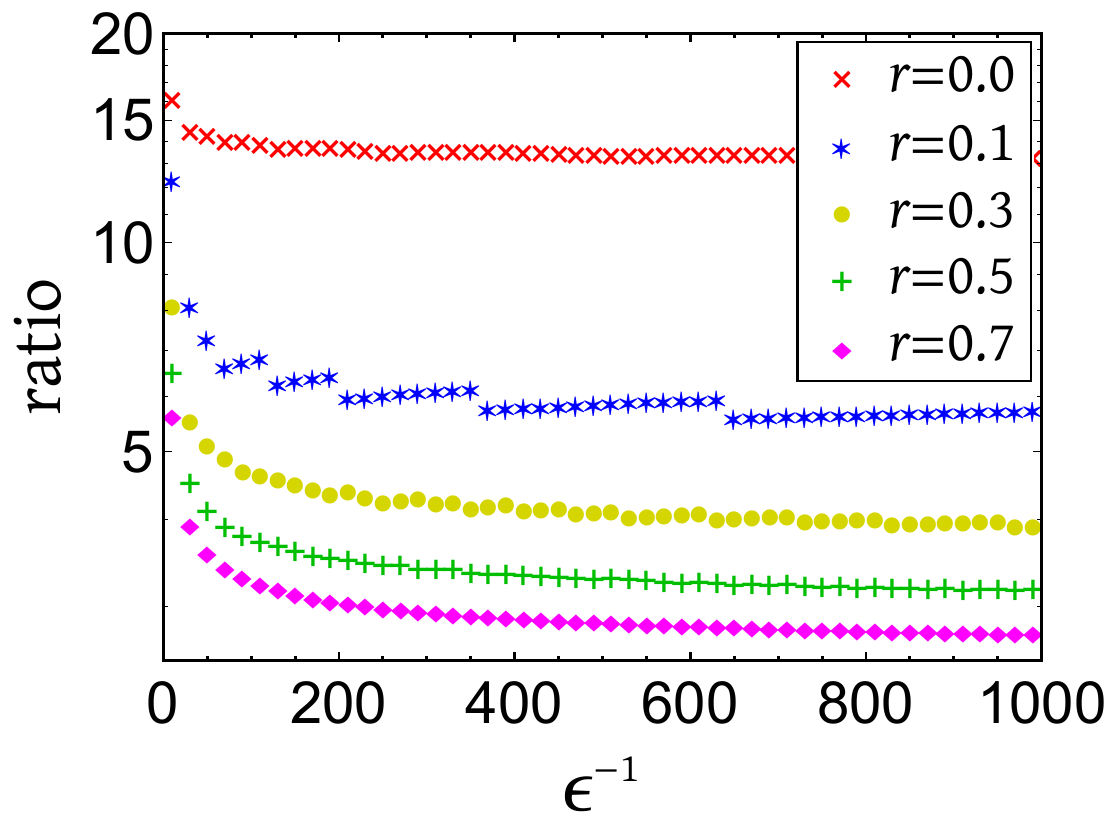}
	\caption{\label{fig:RatioNubNadv}
		The ratio of the upper bound in \eref{eq:UB1N_opt} over $N_{\rm min}(\epsilon,\delta,\lambda,r)$
		with $\lambda=1/2$ and $\delta=\epsilon$. 
	}
\end{figure}

Algorithm~\ref{alg:NoptAdv} is also complementary to 
Theorem \ref{thm:iidHighProb} and Proposition \ref{prep:iidHighProbSmallnu}, which
imply the following upper bounds for $N_{\rm min}(\epsilon,\delta,\lambda,r)$, 
\begin{align}
	N_{\rm min}(\epsilon,\delta,\lambda,r)
	&\leq
	\left\lceil \bigg[ \frac{\lambda\sqrt{2\nu}+1}{\lambda \nu (1-r)} \bigg]^2 \;
	\frac{\ln\delta^{-1} +4\lambda\nu^2}{ \epsilon} \right\rceil 
	\qquad \forall\, 0<\lambda<1,0<\delta\leq\frac{1}{2}, \label{eq:UB1N_opt}
	\\
	N_{\rm min}(\epsilon,\delta,\lambda,r)
	&\leq
	\left\lceil 2\bigg[ \frac{\sqrt{2}+\sqrt{\lambda}}{\sqrt{\lambda}\,(1-r)} \bigg]^2\frac{ \ln\delta^{-1} + \nu}{ \nu \epsilon} \right\rceil
	\qquad \forall\, \frac{1}{2}\leq\lambda<1, 0<\delta\leq\frac{1}{2}.  
\end{align}
These bounds prove that the sample cost of our protocol is only $O(\epsilon^{-1}\ln\delta^{-1})$, which achieves the optimal scaling behaviors with respect to the infidelity $\epsilon$ and significance level $\delta$. 
Numerical calculation based on Algorithm~\ref{alg:NoptAdv} shows that these bounds are still a bit conservative, especially when $r$ is small. 
Nevertheless, these bounds are quite informative about the general trends. 
In the case $\lambda=1/2$, the ratio of the bound in \eref{eq:UB1N_opt} over 
$N_{\rm min}(\epsilon,\delta,\lambda,r)$ is illustrated in Supplementary Figure~\ref{fig:RatioNubNadv}, which shows that this ratio is decreasing in $r$ and is not larger than 14 in the high-precision limit $\epsilon,\delta\to0$.

\subsection{Proofs of Propositions \ref{prop:BoundepsSmallnu}, \ref{prop:UBtestsNumberSmallnu}, \ref{prep:iidHighProbSmallnu}, and \ref{prep:kboundsSmallnu}}\label{sec:proofsmallnu}

\begin{proof}[Proof of Proposition \ref{prop:BoundepsSmallnu}]
	We have
	\begin{align}
		\bar{\epsilon}_{\lambda}(\lfloor\nu s N\rfloor,N,\delta)
		&\stackrel{(a)}{\leq}
		\frac{z^*(\lfloor\nu s N\rfloor,\delta,\lambda) -\lfloor\nu s N\rfloor+1+\sqrt{\lambda \lfloor\nu s N\rfloor}}{\lambda N } 
		\nonumber \\
		&\stackrel{(b)}{\leq}
		\frac{1}{\lambda N }\left[ \bigg(\frac{\lfloor\nu s N\rfloor+2\lambda\ln\delta^{-1}+\sqrt{2\lambda \lfloor\nu s N\rfloor\ln\delta^{-1}}}{\nu}+1\bigg) 
		-\lfloor\nu s N\rfloor+1+\sqrt{\lambda \lfloor\nu s N\rfloor}\right]
		\nonumber \\
		&\leq
		\frac{1}{\lambda N }\left[\lambda sN + \bigg(\frac{\sqrt{2 \ln\delta^{-1}}+\nu}{\nu}\bigg)\sqrt{\lambda \nu s N}  +\frac{2\lambda\ln\delta^{-1}}{\nu}+2\right]
		\nonumber \\
		&\stackrel{(c)}{\leq}
		\frac{1}{\lambda N }\left[\lambda sN + \frac{2\sqrt{\ln\delta^{-1}}\sqrt{\lambda \nu s N} }{\nu}  +\frac{2\lambda\ln\delta^{-1}}{\nu}+2\right]
		\nonumber \\&=
		s+ 2\sqrt{\frac{s \ln\delta^{-1}}{\nu\lambda N}} + \frac{2\ln\delta^{-1}}{\nu N}+\frac{2}{\lambda N}, 
	\end{align}
	which confirms Proposition~\ref{prop:BoundepsSmallnu}. 
	Here $(a)$ follows from \lref{lem:zeta2Bound2}; 
	$(b)$ follows from \lref{lem:z*UB3smallnu}; and $(c)$ follows from 
	the inequality $\sqrt{2 \ln\delta^{-1}}+\nu\leq 2\sqrt{\ln\delta^{-1}}$ for $0<\delta\leq1/3$ and $0<\nu\leq1/2$. 
\end{proof}

\begin{proof}[Proof of Proposition \ref{prop:UBtestsNumberSmallnu}]
	By assumption we have $1/2\leq \lambda<1$, $0\leq s<\epsilon<1$, $0<\delta\leq1/2$, and $N\in \bbZ^{\geq1}$. Therefore,
	\begin{align}
		&\bar{\epsilon}_{\lambda}(\lfloor\nu s N\rfloor,N,\delta) \leq \epsilon , 
		\nonumber\\   \stackrel{(a)}{\Leftarrow} \quad
		&\frac{z^*(\lfloor\nu s N\rfloor,\delta,\lambda)-\lfloor\nu s N\rfloor+1+\sqrt{\lambda\lfloor\nu s N\rfloor}}{\lambda N} \leq \epsilon ,
		\nonumber\\   \stackrel{(b)}{\Leftarrow} \quad
		&\left(\frac{\lfloor\nu s N\rfloor+2\lambda\ln\delta^{-1}+\sqrt{2\lambda \lfloor\nu s N\rfloor\ln\delta^{-1}}}{\nu}+1 \right) -\lfloor\nu s N\rfloor+1+\sqrt{\lambda\lfloor\nu s N\rfloor} \leq \epsilon\lambda N,
		\nonumber\\   \Leftarrow\quad
		&\left(\frac{\nu s N+2\lambda\ln\delta^{-1}+\sqrt{2\lambda \nu s N\ln\delta^{-1}}}{\nu}+1  \right)
		-\nu s N+1+\sqrt{\lambda\nu s N} \leq \epsilon\lambda N,
		\nonumber\\   \Leftarrow\quad
		&\lambda\nu (\epsilon-s)N - \left(\sqrt{2 \ln\delta^{-1}}+\nu\right) \sqrt{\lambda\nu s N} - \left(2\lambda\ln\delta^{-1}+2\nu\right)
		\geq 0, \label{eq:T2solveSmallnu}
		\\
		\stackrel{(c)}{\Leftrightarrow} \quad
		&N\geq \left\lceil \frac{1}{\left[ 2\lambda\nu (\epsilon-s) \right]^2}
		\left[\sqrt{\lambda\nu s \left(\sqrt{2 \ln\delta^{-1}}+\nu\right)^2
			+4\lambda\nu (\epsilon-s) \left(2\lambda\ln\delta^{-1}+2\nu\right)}
		+\left(\sqrt{2 \ln\delta^{-1}}+\nu\right)\sqrt{\lambda\nu s} \right]^2 \right\rceil,
		\nonumber\\    \stackrel{(d)}{\Leftarrow} \quad
		&N\geq \left\lceil \frac{1}{\left[ \lambda\nu (\epsilon-s) \right]^2}
		\left[ \lambda\nu s \left(\sqrt{2 \ln\delta^{-1}}+\nu\right)^2
		+2\lambda\nu (\epsilon-s) \left(2\lambda\ln\delta^{-1}+2\nu\right) \right] \right\rceil,
		\nonumber\\    \stackrel{(e)}{\Leftarrow} \quad
		&N\geq \left\lceil \frac{1}{\lambda\nu (\epsilon-s)^2}
		\left[ 4 s \ln\delta^{-1} + 2\nu^2 s +2 (\epsilon-s) \left(2\lambda\ln\delta^{-1}+2\nu\right) \right] \right\rceil,
		\nonumber\\    \Leftrightarrow \quad
		&N\geq \left\lceil \frac{1}{\lambda\nu (\epsilon-s)^2}
		\left[ (4 \nu s +4\lambda \epsilon) \ln\delta^{-1} + (4\nu\epsilon +2\nu^2 s-4\nu s) \right] \right\rceil,
		\nonumber\\    \stackrel{(f)}{\Leftarrow} \quad
		&N\geq \frac{1}{\lambda\nu (\epsilon-s)^2}
		\left( 4 \epsilon \ln\delta^{-1} + 4\nu \epsilon \right) 
		\quad   \Leftrightarrow \quad
		N\geq \frac{4\, \epsilon }{\lambda\nu (\epsilon-s)^2}
		\left(  \ln\delta^{-1} + \nu \right) , 
		\nonumber 
	\end{align}
	which confirms Proposition \ref{prop:UBtestsNumberSmallnu}. 
	Here the relation $A\Leftarrow B$ means  $B$ is a sufficient condition of $A$; 
	$(a)$ follows from \lref{lem:zeta2Bound2}; 
	$(b)$ follows from \lref{lem:z*UB3smallnu}; 
	$(c)$ is obtained by directly solving \eref{eq:T2solveSmallnu};
	both $(d)$ and $(e)$  follow from the inequality $(x+y)^2\leq 2x^2+2y^2$; 
	and $(f)$ follows from the assumption $s<\epsilon$. 
\end{proof}

\begin{proof}[Proof of Proposition \ref{prep:iidHighProbSmallnu}]
	Let
	\begin{align}
		s= \frac{\sqrt{\lambda}+\sqrt{2}\, r}{\sqrt{2}+\sqrt{\lambda}}\,\epsilon,  \qquad	k= \left\lfloor \nu s N\right\rfloor,
		\qquad
		N\geq \frac{4\, \epsilon\, (\ln\delta^{-1} + \nu )}{\lambda\nu (\epsilon-s)^2} 
		=2\bigg[ \frac{\sqrt{2}+\sqrt{\lambda}}{\sqrt{\lambda}\,(1-r)} \bigg]^2\frac{ \ln\delta^{-1} + \nu}{ \nu \epsilon}.
	\end{align}
	Then  $\bar{\epsilon}_{\lambda}(k, N, \delta)\leq\epsilon$ by Proposition \ref{prop:UBtestsNumberSmallnu}, which confirms the first inequality of \eref{eq:robustCondition} in the main text. 
	
	To complete the proof, it remains to prove the second inequality of \eref{eq:robustCondition} in the main text, that is, $B_{N,k}(\nu r\epsilon)\geq1-\delta$. 
	If $r=0$, then this inequality is obvious. 
	If $r>0$, then 
	\begin{align}
		B_{N,k}(\nu r\epsilon)
		&=1-B_{N,N-k-1}(1-\nu r\epsilon)
		\stackrel{(a)}{\geq} 1-\exp\left[-N D\!\left( \frac{N-k-1}{N}\bigg\| 1-\nu r\epsilon \right) \right] 
		\nonumber \\
		&\stackrel{(b)}{\geq} 1-\rme^{-N D\left(1- \nu s \|1-\nu r\epsilon \right) }
		\stackrel{(c)}{=}     1-\rme^{-N D\left(\nu s \| \nu r\epsilon \right) }, \label{eq:44TH3smallnu}
	\end{align}
	where $(a)$ follows from the Chernoff bound \eqref{eq:ChernoffB} and the inequality $N\nu r\epsilon\leq k+1$;
	$(b)$ follows from \lref{lem:DpqMonoton} and the inequality $(k+1)/N\geq \nu s$;
	$(c)$ follows from the relation $D(x\|y)=D(1-x\|1-y)$.
	In addition, we have
	\begin{align}
		N D(\nu s \| \nu r\epsilon )
		&\geq \frac{4\epsilon\ln\delta^{-1}}{\lambda \nu (\epsilon-s)^2} D(\nu s \| \nu r\epsilon )
		\geq \frac{4\epsilon\ln\delta^{-1}}{\lambda \nu (\epsilon-s)^2} \frac{(\nu s- \nu r\epsilon)^2}{2\nu \epsilon}
		= \frac{2(s- r\epsilon)^2}{\lambda (\epsilon-s)^2} \ln\delta^{-1}
		\nonumber \\
		&= \frac{2\Bigl(\frac{\sqrt{\lambda}+\sqrt{2}\, r}{\sqrt{2}+\sqrt{\lambda}}\epsilon-r\epsilon\Bigr)^2}{\lambda \Bigl(\epsilon-\frac{\sqrt{\lambda}+\sqrt{2}\, r}{\sqrt{2}+\sqrt{\lambda}}\epsilon\Bigr)^2} \ln\delta^{-1}
		=\ln\delta^{-1},  \label{eq:45TH3smallnu}
	\end{align}
	where the second inequality follows from the inequality $D(x\|y)\geq(x-y)^2/(2x)$ for $x\geq y$ and the fact that $s\leq \epsilon$.
	Equations~\eqref{eq:44TH3smallnu} and~\eqref{eq:45TH3smallnu} together confirm the inequality $B_{N,k}(\nu r\epsilon)\geq1-\delta$ and complete the proof of Proposition~\ref{prep:iidHighProbSmallnu}.
\end{proof}

\begin{proof}[Proof of Proposition~\ref{prep:kboundsSmallnu}]
	For given $1/2\leq\lambda<1$, $0<\epsilon <1$, $0<\delta\leq1/4$, and $N,k\in\bbZ^{\geq 0}$, we have
	\begin{align}
		\begin{split}
			&k\leq 
			\bigg\lfloor\nu \epsilon N - 2\sqrt{\frac{\nu\epsilon N\ln\delta^{-1}}{\lambda}}
			-2\ln\delta^{-1}-\frac{2\nu}{\lambda} \bigg\rfloor,
			\\   \stackrel{(a)}{\Rightarrow} \quad 
			&k
			\leq \nu \epsilon N - \Big( \sqrt{2 \ln\delta^{-1}}+\nu \Big) \sqrt{\frac{\nu\epsilon N}{\lambda}}
			- \left(2\ln\delta^{-1}+\frac{2\nu}{\lambda}\right),
			\\   \Rightarrow\quad 
			&k\leq \nu \epsilon N,\quad  
			\frac{\lambda}{\nu}k +\bigg( \frac{\sqrt{2\lambda \ln\delta^{-1}}}{\nu}+\sqrt{\lambda} \bigg) \sqrt{\nu \epsilon N} 
			+ \left(\frac{2\lambda\ln\delta^{-1}}{\nu}+2-  \lambda \epsilon N \right)\leq 0,
			\\   \Rightarrow\quad 
			&k\leq N-1,\quad  
			\frac{\lambda}{\nu}k +\bigg( \frac{\sqrt{2\lambda \ln\delta^{-1}}}{\nu}+\sqrt{\lambda} \bigg) \sqrt{k} 
			+ \left(\frac{2\lambda\ln\delta^{-1}}{\nu}+2-  \lambda \epsilon N \right)\leq 0,
			\\   \Rightarrow\quad 
			&k\leq N-1,\quad  
			\bigg(\frac{k+2\lambda\ln\delta^{-1}+\sqrt{2\lambda k\ln\delta^{-1}}}{\nu}+1\bigg) -k+1+\sqrt{\lambda k} 
			\leq \epsilon\lambda N,
			\\   \stackrel{(b)}{\Rightarrow}\quad 
			&k\leq N-1,\quad  
			\frac{z^*(k,\delta,\lambda)-k+1+\sqrt{\lambda k}}{\lambda N} \leq \epsilon 
			\\
			\stackrel{(c)}{\Rightarrow} \quad  
			&k\leq N-1,\quad  
			\bar{\epsilon}_{\lambda}(k,N,\delta) \leq \epsilon , 
		\end{split}
	\end{align}
	which confirm Proposition~\ref{prep:kboundsSmallnu}.
	Here 
	$(a)$ follows from the inequality $\sqrt{2 \ln\delta^{-1}}+\nu\leq 2\sqrt{\ln\delta^{-1}}$ for $0<\delta\leq1/4$ and $0<\nu\leq1/2$; 
	$(b)$ follows from \lref{lem:z*UB3smallnu}; and
	$(c)$ follows from \lref{lem:zeta2Bound2}.
\end{proof}

\subsection{Proofs of \lref{lem:zstarBoundsN} and Proposition \ref{prop:Nk}}\label{sec:prooffixkadv}

\begin{proof}[Proof of \lref{lem:zstarBoundsN}]
	First, we have
	\begin{align}
		N_{k}(\epsilon,\delta,\lambda)
		&= \min\{ N\geq k+1\,|\, \bar{\epsilon}_{\lambda}(k,N,\delta)\leq \epsilon \} 
		\leq \min\left\{ N\geq k+1\,\Bigg|\,\frac{z^*-k+1+\sqrt{\lambda k}}{\lambda(N-z^*)+z^*-k+1+\sqrt{\lambda k} } \leq \epsilon \right\} \nonumber \\
		&= \biggl\lceil\frac{z^*-k+1+\sqrt{\lambda k}}{\lambda}\Bigl(\frac{1}{\epsilon}-1\Bigr)+z^*\biggr\rceil 
		< \frac{z^*-k+1+\sqrt{\lambda k}}{\lambda \epsilon},  
	\end{align}
	which confirms the upper bound in \lref{lem:zstarBoundsN}. Here the first inequality follows from \lref{lem:zeta2Bound2}. 
	
	In addition, we have
	\begin{align}
		N_{k}(\epsilon,\delta,\lambda)
		&= \min\{ N\geq k+1\,|\, \bar{\epsilon}_{\lambda}(k,N,\delta)\leq \epsilon \} \geq \min\left\{ N\geq k+1\,\Bigg|\,\frac{z_*-k}{\lambda(N+1)+\nu z_*-k } \leq \epsilon \right\} 
		\nonumber \\
		&= \biggl\lceil \frac{(1-\nu\epsilon)z_*-(1-\epsilon)k}{\lambda \epsilon} - 1 \biggr\rceil , 
	\end{align}
	which confirms the lower bound in \lref{lem:zstarBoundsN}. Here the inequality follows from \lref{lem:zeta2Bound2}. 
\end{proof}

\begin{proof}[Proof of Proposition~\ref{prop:Nk}]
	If integer $N$ satisfies $k+1 \leq N<k/(\nu\epsilon)$, then  \lref{lem:LLP} implies that 
	$\bar{\epsilon}_{\lambda}(k,N,\delta)> \epsilon$. 
	This fact and the definition in \eref{eq:DefNTConst} together confirm the lower bound in \eref{eq:NuBinformCR}. 
	
	The upper bound in \eref{eq:NuBinformCR} can be derived as follows, 
	\begin{align} 
		N_{k}(\epsilon,\delta,\lambda)
		&\stackrel{(a)}{<}
		\frac{z^*-k+1+\sqrt{\lambda k}}{\lambda \epsilon}
		\stackrel{(b)}{\leq} 
		\frac{\big( \frac{k}{\nu}+\frac{\sqrt{2\nu k\ln\delta^{-1}}}{2\nu^2}+\frac{\ln\delta^{-1}}{2\nu^2}+1 \big)  -k+1+\sqrt{\lambda k}}{\lambda \epsilon}
		\nonumber \\
		&=\frac{k}{\nu\epsilon} + \frac{\ln\delta^{-1}}{2\nu^2\lambda \epsilon} 
		+ \frac{\sqrt{2\nu k\ln\delta^{-1}}}{2\nu^2\lambda \epsilon} + \frac{\sqrt{\lambda k}}{\lambda \epsilon} + \frac{2}{\lambda \epsilon},   
	\end{align}
	where $(a)$ follows from the upper bound in \lref{lem:zstarBoundsN}, and $(b)$ follows from \lref{lem:z*UB3}.

	Equation \eqref{eq:NlBinformCR} can be derived as follows,  
	\begin{align} 
		N_{k}(\epsilon,\delta,\lambda)
		&\stackrel{(a)}{\geq}
		\frac{(1-\nu\epsilon)(z^*-1)-(1-\epsilon)k}{\lambda \epsilon} - 1 
		\stackrel{(b)}{\geq} 
		\frac{(1-\nu\epsilon)\left(\frac{\ln \delta}{\ln \lambda}-1 \right) -(1-\epsilon)k}{\lambda \epsilon} - 1 
		\nonumber \\
		&= \frac{(1-\nu\epsilon)\ln\delta^{-1}}{\lambda\epsilon\ln\lambda^{-1}} - \frac{(1-\epsilon)(k+1)}{\lambda \epsilon} -2,  
	\end{align}
	where $(a)$ follows from the lower bound in \lref{lem:zstarBoundsN}, and $(b)$ holds because
	\begin{align}
		z^*&=\min\{z \in\bbZ^{\geq k} \, |\, B_{z,k}(\nu)\leq \delta \}
		\geq \min\bigl\{z\in\bbZ^{\geq k} \, |\,\lambda^z\leq \delta \bigr\}
		=\biggl\lceil\frac{\ln \delta}{\ln \lambda}\biggr\rceil .
	\end{align}

	Equation \eqref{eq:NuBinformCRsmallnu}  can be derived as follows,  
	\begin{align} 
		N_{k}(\epsilon,\delta,\lambda)
		&\stackrel{(a)}{<}
		\frac{z^*-k+1+\sqrt{\lambda k}}{\lambda \epsilon}
		\stackrel{(b)}{\leq} 
		\frac{\frac{1}{\nu} \big( k+2\lambda\ln\delta^{-1}+\sqrt{2\lambda k\ln\delta^{-1}}+\nu \big)  -k+1+\sqrt{\lambda k}}{\lambda \epsilon}
		\nonumber \\
		&=\frac{k}{\nu\epsilon} + \frac{2\ln\delta^{-1}}{\nu \epsilon} 
		+ \frac{\sqrt{2\lambda k\ln\delta^{-1}}}{\lambda\nu \epsilon} + \frac{\sqrt{\lambda k}}{\lambda \epsilon} + \frac{2}{\lambda \epsilon},  
	\end{align}
	where $(a)$ follows from the upper bound in \lref{lem:zstarBoundsN}, and $(b)$ follows from \lref{lem:z*UB3smallnu}. 
\end{proof}

\section{Robust and efficient verification of quantum states in the i.i.d.\! scenario}\label{app:IIDsetting}

In the Methods section, we proposed a robust and efficient protocol for verifying pure states in the i.i.d.\! scenario. 
In this section, we provide more results about QSV in the i.i.d.\! scenario and more details on the performance of our protocol.  
In \ref{app:IIDsetting6A}, we review the basic framework of QSV in the i.i.d.\! scenario adopted by most previous works, 
and explain its limitation with respect to robustness. 
In \ref{app:iidrelation}, we clarify the relation between the guaranteed infidelity $\bar{\epsilon}^{\,\iid}_{\lambda}(k,N,\delta)$ of our protocol and the companion paper \cite{Classical22}. 
In \ref{app:IIDsetting6C}, we prove Propositions~\ref{prop:epsiidMonoton}--\ref{prop:iidHighProbiid} presented in the Methods section. 
In \ref{app:IIDsetting6D}, we clarify the performance of our protocol
for a fixed number of allowed failures. 
In \ref{app:IIDsetting6E}, we clarify the number of allowed failures for a fixed number of tests. 
In \ref{app:IIDsetting6F}, we explain Algorithm~\ref{alg:iidNoptAdv} proposed in the Methods section.
In \ref{app:IIDsetting6G}, we derive the bounds  for the minimum number of 
tests $N_{\rm min}^{\rm iid}(\epsilon,\delta,\lambda,r)$ presented in \eref{eq:iidNoptUB} in the main text and 
illustrate the tightness of these bounds.

\subsection{QSV in the i.i.d.\! scenario: basic framework and limitations}\label{app:IIDsetting6A}

Suppose a quantum device is expected to produce the pure state $|\Psi\>\in\caH$, but actually produces the  states $\sigma_1,\sigma_2,\dots,\sigma_N$  in $N$ runs. In the i.i.d.\! scenario, all these states are identical to $\sigma$. The goal of Alice is to verify whether $\sigma$ is sufficiently close to the target state $|\Psi\>$. To this end, in each run Alice can perform a random test from  a set of accessible tests, where each test corresponds to a two-outcome measurement \cite{PLM18,ZhuEVQPSlong19}. The overall effect of these tests can also be described by a two-outcome measurement $\{\Omega, \openone-\Omega\}$, where $\Omega$ satisfies $0\leq \Omega\leq \openone$ and  $\Omega|\Psi\>=|\Psi\>$, so that the target state can always pass each test. 
If all $N$ tests are passed,  then Alice accepts the states prepared; otherwise, she rejects. 

When the infidelity between $\sigma$ and the target state $|\Psi\>$ is $\epsilon_\sigma$, 
the maximum probability that $\sigma$ can pass each test on average is given by \cite{PLM18,ZhuEVQPSlong19} 
\begin{equation}\label{eq:NAPrPass1test}
	\max_{\sigma\in\mathcal{D}(\caH)} \; \tr(\Omega \sigma) 
	= 1- [1-\lambda(\Omega)]\epsilon_\sigma
	= 1- \nu(\Omega)\epsilon_\sigma,
\end{equation}
where 
$\lambda(\Omega):=\|\Omega-|\Psi\>\<\Psi|\|$ is the second largest eigenvalue of $\Omega$, 
and $\nu(\Omega):=1-\lambda(\Omega)$ is the spectral gap from the largest eigenvalue. 
The probability of passing all $N$ tests is at most $[1-\nu(\Omega)\epsilon_\sigma]^N$.
To verify the target state $|\Psi\>$ within infidelity $\epsilon$ and significance level $\delta$, which means $[1-\nu(\Omega)\epsilon]^N\leq \delta$,
it suffices to choose the following number of tests \cite{PLM18,ZhuEVQPSlong19}, 
\begin{equation}\label{eq:NumberTestNon}
	N_{\rm NA}(\epsilon, \delta, \lambda)
	= \biggl\lceil\frac{ \ln \delta}{\ln[1-\nu(\Omega)\epsilon]}\biggr\rceil. 
\end{equation}
For many quantum states of interest, verification strategies with large spectral gaps  can be constructed using LOCC \cite{PLM18,ZH3,ZH4,LHZ19,Wang19,Yu19,LiGHZ19,Liu19,Li21,HayaTake19,Dangn20,LiuYC21,ZLC22,CLZ22}. So these states can be verified efficiently in the sense that deviation from the target state can be detected efficiently.

In the basic framework outlined above,  Alice can draw a meaningful conclusion about the state prepared  if  all $N$ tests are passed, but little information can be extracted  if one or more failures are observed. If the state prepared is not perfect, then it may be rejected with a high probability even if it has a high fidelity. 
For example, suppose $0\leq\lambda<1$, $0<\delta<1$, and  $\epsilon_\sigma=\epsilon/2\leq1/4$, where $\epsilon$ and $\delta$ are the target infidelity and significance level; 
then the probability that Alice accepts is upper bounded by 
\begin{align}\label{eq:NONpassPrIIDUB}
	p_{\rm acc}
	\leq
	\left( 1-\frac{\nu\epsilon}{2}\right) ^{N_{\rm NA}(\epsilon, \delta, \lambda)}
	\leq 
	\delta^{\, 0.415},  
\end{align}
where the second inequality is proved below.  
This probability decreases rapidly as $\delta$ decreases, which means  previous verification protocols
are not robust to noise in state preparation. 
As a consequence, many repetitions are necessary to guarantee that Alice accepts the state $\sigma$ at least once. 
To achieve confidence level $1-\delta$ for example, the number of repetitions required is given by [cf. \eref{eq:repetitionNum}]
\begin{align}
	M_{\rm NA}
	= \left\lceil \frac{\ln\delta}{\ln\left( 1-p_{\text{acc}}\right) } \right\rceil
	\geq  \frac{\ln\delta^{-1}}{p_{\text{acc}}}-\ln\delta^{-1}
	\geq  \frac{\ln\delta^{-1}}{\delta^{\, 0.415}}-\ln\delta^{-1}
	\approx \frac{\ln\delta^{-1}}{\delta^{\, 0.415}}. 
\end{align}
Accordingly, the total sample cost rises up to 
\begin{align}
	N_{\rm NA} M_{\rm NA} 
	= \biggl\lceil\frac{ \ln \delta}{\ln(1-\nu\epsilon)}\biggr\rceil \left\lceil \frac{\ln\delta}{\ln\left( 1-p_{\text{acc}}\right) } \right\rceil
	\geq \Theta\left( \frac{(\ln\delta)^2}{ \delta^{\, 0.415} \epsilon}  \right) , 
\end{align}
which is substantially larger than the sample cost in \eref{eq:NumberTestNon}, which does not take robustness into account.

\begin{proof}[Proof of \eref{eq:NONpassPrIIDUB}]
	For $0<c<1$ and $0<x<1$ define the function
	\begin{align}\label{eq:etax}
		\eta(c,x) := \frac{\ln(1-cx)}{\ln(1-x)}. 
	\end{align}	
	By assumption we have $0\leq\lambda<1$, $0<\delta<1$, and $0<\epsilon\leq1/2$; therefore,
	\begin{align}
		\left( 1-\frac{\nu\epsilon}{2}\right) ^{N_{\rm NA}(\epsilon, \delta, \lambda)}
		&\stackrel{(a)}{\leq}  \exp\left[\frac{ \ln \delta}{\ln(1-\nu\epsilon)} \ln \left( 1-\frac{\nu\epsilon}{2}\right) \right]
		=  \exp\left[  (\ln \delta)  \, \eta\left(  \frac{1}{2}, \nu\epsilon\right)  \right]
		\nonumber \\
		&\stackrel{(b)}{\leq} \exp\left[  (\ln \delta)  \, \eta\left(  \frac{1}{2}, \frac{1}{2} \right) \right]
		=  \exp\left[  (\ln \delta)  \frac{\ln(3/4)}{\ln(1/2)}  \right] 
		\stackrel{(c)}{\leq} \delta^{\,0.415}, \label{eq:NONpdelexp1}
	\end{align}
	where $(a)$ follows from \eref{eq:NumberTestNon}; $(b)$ follows from \lref{lemma:Monotetax} below and the assumption $\nu\epsilon\leq1/2$; 
	and $(c)$ follows from the inequality $\ln(3/4)/\ln(1/2)>0.415$.
\end{proof}

\begin{lemma}\label{lemma:Monotetax}
	Suppose $0<c<1$. Then $\eta(c,x)$  defined in \eref{eq:etax}
	is strictly decreasing in $x$ for $0<x<1$.  
\end{lemma}

\begin{proof}[Proof of \lref{lemma:Monotetax}]
	For $0<c<1$ and $0\leq x<1$ define the function
	\begin{align}
		\hat{\eta}_c(x) := c(x-1)\ln(1-x) + (1-cx)\ln(1-cx). 
	\end{align}
	Then we have 
	\begin{align}
		\frac{{\rm d }\hat{\eta}_c(x) }{{\rm d} x} 
		&=  c\ln(1-x) - c\ln(1-cx), \quad \frac{{\rm d }\hat{\eta}_c(x) }{{\rm d} x}\bigg|_{x=0}=0,\\
		\frac{{\rm d ^2}\hat{\eta}_c(x) }{{\rm d} x^2} 
		&=  \frac{-c}{1-x} + \frac{c^2}{1-cx}
		=  \frac{c(c-1)}{(1-x)(1-cx)}
		< 0  
		\qquad \forall\, 0<x<1, 
	\end{align}
	which imply that $\frac{{\rm d }\hat{\eta}_c(x) }{{\rm d} x}<0$ for $0<x<1$. 
	So $\hat{\eta}_c(x)<0$ for $0<x<1$ given that
	$\hat{\eta}_c(0)=0$. As a consequence,
	\begin{align}
		\frac{\partial \eta(c,x) }{\partial x} 
		= \frac{1}{[\ln(1-x)]^2} \left[ -\frac{c \ln(1-x)}{1-cx} + \frac{\ln(1-cx)}{1-x}\right]  
		= \frac{\hat{\eta}_c(x)}{[\ln(1-x)]^2 (1-cx)(1-x)} 
		< 0 
		\qquad \forall\, 0<x<1,  
	\end{align}
	which confirms \lref{lemma:Monotetax}. 
\end{proof}

\subsection{Relation between $\bar{\epsilon}^{\,\iid}_{\lambda}(k,N,\delta)$ and the companion paper \cite{Classical22}}\label{app:iidrelation}
Recall that in \eref{eq:epslamDef} of the Methods section, the two optimizations are taken over all permutation-invariant 
joint distributions on the $N+1$ variables $Y_1, Y_2,\ldots, Y_{N+1}$.  
Here we turn to an alternative setting in which  the variables $Y_1,Y_2, \ldots, Y_{N+1}$ are independent and identically distributed (i.i.d.).
Consider 
\begin{align}\label{eq:epsiidDef}
	\max_{\text{i.i.d.}} \, \{\Pr(Y_{N+1}=1|K \le k)\,|\Pr(K \le k) \geq \delta \} ,
\end{align}
where the maximization is over all independent and identical distributions on $Y_1, \ldots, Y_{N+1}$. 
This quantity is called the upper confidence limit for the i.i.d.\! setting in our companion paper \cite{Classical22},  which
studied the asymptotic behaviors of \eref{eq:epsiidDef} and related quantities.

The quantity in \eref{eq:epsiidDef} is actually equal to $\bar{\epsilon}^{\,\iid}_{\lambda}(k,N,\delta)$ defined in \eref{eq:hatzetaiid} of the Methods section.
According to \eref{eq:hatzetaiid} in the Methods section we have 
\begin{align}\label{eq:epsIIDlam0lamRel}
	\bar{\epsilon}^{\,\iid}_{\lambda}(k,N,\delta)=\frac{\bar{\epsilon}^{\,\iid}_{\lambda=0}(k,N,\delta)}{\nu}
	\qquad
	\forall\, \bar{\epsilon}^{\,\iid}_{\lambda=0}(k,N,\delta)\leq \nu.   
\end{align}
Based on this relation, 
many properties of $\bar{\epsilon}^{\,\iid}_{\lambda=0}(k,N,\delta)$ derived in Ref.~\cite{Classical22}
also apply to the guaranteed infidelity $\bar{\epsilon}^{\,\iid}_{\lambda}(k,N,\delta)$ in this section after proper modifications. 
Notably, Eqs.~\eqref{eq:iiddel0} and \eqref{eq:iideps0} below in this paper are corollaries of Eqs.~(90) and (87) in Ref.~\cite{Classical22}, respectively.
In this paper, from a more practical perspective than Ref.~\cite{Classical22}, we  derive several finite bounds for 
$\overline{\epsilon}^{\,\iid}_\lambda(k,N,\delta)$ and related quantities.  
The main goal of this section is to provide a robust and efficient protocol for verifying quantum states
in the i.i.d.\! scenario.

\subsection{Proofs of Propositions \ref{prop:epsiidMonoton}, \ref{prop:Boundepsiid}, \ref{prop:iidUBtestsNumber}, and \ref{prop:iidHighProbiid}}\label{app:IIDsetting6C}

\begin{proof}[Proof of Proposition \ref{prop:epsiidMonoton}]
	Suppose $0<\delta_1<\delta_2\leq1$, $\epsilon_1=\bar{\epsilon}^{\,\iid}_{\lambda}(k,N,\delta_1)$, and $\epsilon_2=\bar{\epsilon}^{\,\iid}_{\lambda}(k,N,\delta_2)$. Then $B_{N,k}(\nu\epsilon_1) = \delta_1$ and $B_{N,k}(\nu\epsilon_2) = \delta_2$ according to Lemma~\ref{lem:epsiidsolve} below, which implies that  $\epsilon_1>\epsilon_2$, that is, $\bar{\epsilon}^{\,\iid}_{\lambda}(k,N,\delta_1)>\bar{\epsilon}^{\,\iid}_{\lambda}(k,N,\delta_2)$, given that 
	$B_{N,k}(\nu\epsilon)$ is strictly decreasing in $\epsilon$ by Lemma~\ref{lem:Bzkmono}. 
	
	The monotonicities with $N$ and $k$ follow from a similar reasoning.
\end{proof}

\begin{lemma}\label{lem:epsiidsolve}
	Suppose $0\leq\lambda< 1$, $0<\delta\leq1$, $k\in \bbZ^{\geq 0}$, and $N\in \bbZ^{\geq k+1}$. Then $\bar{\epsilon}^{\,\iid}_{\lambda}(k,N,\delta)$ 
	is the unique solution of $\epsilon$ to the equation
	\begin{align}\label{eq:epsiidsolve}
		B_{N,k}(\nu\epsilon) = \delta, \qquad 0\leq \epsilon\leq1.
	\end{align}
\end{lemma}

\lref{lem:epsiidsolve} follows from \eref{eq:hatzetaiid} in the Methods section, \lref{lem:Bzkmono}, as well as the facts that $B_{N,k}(0)=1\geq\delta>0=B_{N,k}(1)$ and that 
$B_{N,k}(\nu\epsilon)$ is continuous in $\epsilon$ for $0\leq \epsilon\leq1$.

\begin{proof}[Proof of Proposition \ref{prop:Boundepsiid}]
	We have  
	\begin{align}
		\bar{\epsilon}^{\,\iid}_{\lambda}(\lfloor\nu s N\rfloor,N,\delta) 
		&\stackrel{(a)}{=}
		\max_{\epsilon} \left\{0\leq \epsilon\leq1 \,|\, B_{N,\lfloor\nu s N\rfloor}(\nu\epsilon) \geq \delta  \right\}
		\nonumber \\&\stackrel{(b)}{\leq }
		\max_{\epsilon} \left\{\frac{\lfloor\nu s N\rfloor}{\nu N} \leq \epsilon\leq \frac{1}{\nu} \,\bigg| \exp\left[- D\!\left( \frac{\lfloor\nu s N\rfloor}{N}\bigg\| \nu\epsilon\right) N \right] \geq \delta  \right\}
		\nonumber \\&\stackrel{(c)}{\leq }
		\max_{\epsilon} \left\{ s \leq \epsilon\leq \frac{1}{\nu} \,\bigg| \exp\left[- D( \nu s \| \nu\epsilon) N \right] \geq \delta  \right\}
		\nonumber \\&\stackrel{(d)}{\leq }
		\max_{\epsilon} \left\{ s \leq \epsilon\leq \frac{1}{\nu} \,\bigg| \exp\left[- \frac{(\nu s-\nu\epsilon)^2}{2\nu\epsilon} N \right] \geq \delta  \right\}	
		\nonumber \\&= 
		\max_{\epsilon} \left\{\epsilon\geq s \,\bigg|\, \epsilon^2 - 2\left(s + \frac{\ln\delta^{-1}}{\nu N} \right)\epsilon + s^2 \leq 0   \right\}	
		\nonumber \\&= 
		s + \frac{\ln\delta^{-1}}{\nu N} + \sqrt{\left( \frac{\ln\delta^{-1}}{\nu N} \right)^2 + \frac{2s \ln\delta^{-1}}{\nu N}}
		\stackrel{(e)}{\leq }
		s + \sqrt{\frac{2s \ln\delta^{-1}}{\nu N}} + \frac{2\ln\delta^{-1}}{\nu N}, 
	\end{align}
	which confirms the upper bound in \eref{eq:iidBDeps} of the Methods section. Here 
	$(a)$ follows from \eref{eq:hatzetaiid} in the Methods section; 
	$(b)$ follows from the Chernoff bound \eqref{eq:ChernoffB}; 
	$(c)$ follows because $D( \nu s \| \nu\epsilon)\leq D( \lfloor\nu s N\rfloor/N \| \nu\epsilon)$ for $s\leq \epsilon$ 
	(according to \lref{lem:DpqMonoton}); 
	$(d)$ follows from the inequality $D(x\|y)\geq(x-y)^2/(2y)$ for $x\leq y$;
	and $(e)$ follows from the inequality $\sqrt{x+y}\leq\sqrt{x}+\sqrt{y}$ for $x,y\geq0$.

	In addition, we have
	\begin{align}
		\bar{\epsilon}^{\,\iid}_{\lambda}(\lfloor\nu s N\rfloor,N,\delta) 
		&\stackrel{(a)}{=}    \max\{0\leq\epsilon\leq1 \,| B_{N,\lfloor\nu s N\rfloor}(\nu\epsilon)\geq\delta\} \nonumber \\
		&\stackrel{(b)}{\geq} \max\{0\leq\epsilon\leq1 \,| B_{N,\lfloor\nu s N\rfloor}(\nu\epsilon)\geq1/2\}
		\stackrel{(c)}{>}    \frac{\lfloor\nu s N\rfloor}{\nu N}
		\geq s- \frac{1}{\nu N}, 
		\label{eq:barepsilonLB}
	\end{align}
	which confirms the lower bound in \eref{eq:iidBDeps} of the Methods section.
	Here $(a)$ follows from \eref{eq:hatzetaiid} in the Methods section, and $(b)$ follows from the assumption $\delta\leq 1/2$.
	To prove $(c)$,
	we shall consider two cases depending on the value of $\lfloor\nu s N\rfloor$.
	\begin{enumerate}
		\item[1.] $\lfloor\nu s N\rfloor=0$.
		In this case we have
		\begin{align}
			\max\{0\leq\epsilon\leq1 \,| B_{N,\lfloor\nu s N\rfloor}(\nu\epsilon)\geq1/2\}
			=\max\{0\leq\epsilon\leq1 \,|\, (1-\nu\epsilon)^N\geq1/2\}
			>0 ,
		\end{align}
		which confirms $(c)$ in \eref{eq:barepsilonLB}.
		
		\item[2.] $\lfloor\nu s N\rfloor\geq1$. 
		In this case we have $B_{N,\lfloor\nu s N\rfloor}\!\left(\lfloor\nu s N\rfloor/N\right)>1/2$
		by \lref{lem:Bnk1-k/n>1/2} below.
		Since $B_{N,\lfloor\nu s N\rfloor}(\nu\epsilon)$ is continuous in $\epsilon$,
		there exists a number $\epsilon> \lfloor\nu s N\rfloor/(\nu N)$ such that $B_{N,\lfloor\nu s N\rfloor}(\nu\epsilon) >1/2$.
		This fact confirms $(c)$ in \eref{eq:barepsilonLB} again.
	\end{enumerate}
	In conclusion, the inequality $(c)$ in \eref{eq:barepsilonLB} holds, which completes the proof of Proposition \ref{prop:Boundepsiid}.
\end{proof}

\begin{lemma}[Corollary B.3, \cite{Classical22}]\label{lem:Bnk1-k/n>1/2}
	Suppose $k\in\bbZ^{\geq1}$ and $z\in\bbZ^{\geq k}$; then $B_{z,k}(k/z)>1/2$.
\end{lemma}

\begin{proof}[Proof of Proposition \ref{prop:iidUBtestsNumber}]
	By assumption we have $0\leq \lambda<1$, $0\leq s<\epsilon<1$, and $0<\delta<1$. Therefore, 
	\begin{align}
		\begin{split} 
			\bar{\epsilon}^{\,\iid}_{\lambda}(\lfloor\nu s N\rfloor,N,\delta) \leq \epsilon  
			&\quad \stackrel{(a)}{\Leftarrow} \quad
			B_{N,\lfloor\nu s N\rfloor}(\nu\epsilon) \leq \delta
			\quad \stackrel{(b)}{\Leftarrow} \quad
			\exp\left[- D\!\left( \frac{\lfloor\nu s N\rfloor}{N}\bigg\| \nu\epsilon\right) N \right] \leq \delta
			\\
			&\quad \stackrel{(c)}{\Leftarrow} \quad
			\exp\left[- D(\nu s \| \nu\epsilon) N \right] \leq \delta
			\quad \Leftarrow \quad 
			N\geq\frac{\ln\delta^{-1}}{D(\nu s\|\nu\epsilon)},
		\end{split} 
	\end{align}
	which confirm Proposition \ref{prop:iidUBtestsNumber}. 
	Here the relation $A\Leftarrow B$ means $B$ is a sufficient condition of $A$; 
	$(a)$ follows from \lref{lem:epsiidANDBnk} below; 
	$(b)$ follows from the Chernoff bound \eqref{eq:ChernoffB} and the inequality  $\lfloor\nu s N\rfloor/N\leq\nu \epsilon$; 
	$(c)$ follows from \lref{lem:DpqMonoton} and the inequality  $\lfloor\nu s N\rfloor/N\leq\nu s$.  
\end{proof}

The following lemma implies that the two conditions $\bar{\epsilon}^{\,\iid}_{\lambda}(k,N,\delta)\leq \epsilon$ 
and $B_{N,k}(\nu\epsilon) \leq \delta$ are equivalent. 

\begin{lemma}\label{lem:epsiidANDBnk}
	Suppose $0\leq\lambda< 1$, $0<\delta\leq1$, $0\leq \epsilon\leq1$, $k\in \bbZ^{\geq 0}$, and $N\in \bbZ^{\geq k+1}$. 
	If $B_{N,k}(\nu\epsilon) \leq \delta$, then $\bar{\epsilon}^{\,\iid}_{\lambda}(k,N,\delta)\leq \epsilon$. 
	If $B_{N,k}(\nu\epsilon) > \delta$,    then $\bar{\epsilon}^{\,\iid}_{\lambda}(k,N,\delta)> \epsilon$. 
\end{lemma}

\begin{proof}[Proof of \lref{lem:epsiidANDBnk}]
	According to \lref{lem:epsiidsolve} we have $\epsilon= \bar{\epsilon}^{\,\iid}_{\lambda}\left( k,N,B_{N,k}(\nu\epsilon)\right)$. 
	In addition, $\overline{\epsilon}^{\,\iid}_\lambda(k,N,\delta)$
	is strictly decreasing in $\delta$ for $0<\delta\leq1$ according to Proposition \ref{prop:epsiidMonoton}. 
	The two facts together confirm \lref{lem:epsiidANDBnk}. 
\end{proof}	

\begin{proof}[Proof of Proposition \ref{prop:iidHighProbiid}]
	Let
	\begin{align}
		s\in (\nu\epsilon,\epsilon), 
		\qquad
		k = \left\lfloor  \nu s N \right\rfloor, 
		\qquad
		N \geq
		\frac{ \ln\delta^{-1}}{ \min\{D(\nu s\|\nu r\epsilon), D(\nu s\|\nu\epsilon)\} }.  
	\end{align}
	Then we have $\bar{\epsilon}^{\,\iid}_{\lambda}(k, N, \delta)\leq\epsilon$ by Proposition \ref{prop:iidUBtestsNumber}. 
	It follows that $B_{N,k}(\nu\epsilon)\leq\delta$ according to \lref{lem:epsiidANDBnk}, which confirms the first inequality of \eref{eq:robustConditionIID} in the Methods section.
	
	To complete the proof, it remains to prove the second inequality of \eref{eq:robustConditionIID} in the Methods section, 
	that is, $B_{N,k}(\nu r\epsilon)\geq1-\delta$.
	Note that
	\begin{align}
		B_{N,k}(\nu r\epsilon)
		&=1-B_{N,N-k-1}(1-\nu r\epsilon)
		\stackrel{(a)}{\geq} 1-\exp\left[-D\!\left( \frac{N-k-1}{N}\bigg\| 1-\nu r\epsilon \right) N\right] 
		\nonumber \\
		&\stackrel{(b)}{\geq} 1-\rme^{-D\left(1- \nu s \|1-\nu r\epsilon \right)N }
		\stackrel{(c)}{=}     1-\rme^{-D\left(\nu s \| \nu r\epsilon \right) N}, \label{eq:78iidTH3}
	\end{align}
	where $(a)$ follows from the Chernoff bound \eqref{eq:ChernoffB} and the inequality $N\nu r\epsilon\leq k+1$;
	$(b)$ follows from \lref{lem:DpqMonoton} and the inequality $(k+1)/N\geq \nu s$;
	$(c)$ follows from the relation $D(x\|y)=D(1-x\|1-y)$.
	In addition, we have
	\begin{align}\label{eq:79iidTH3}
		\begin{split}
			D(\nu s \| \nu r\epsilon )N
			\geq  \frac{ D(\nu s \| \nu r\epsilon ) \ln\delta^{-1}}{ \min\{D(\nu s\|\nu r\epsilon), D(\nu s\|\nu\epsilon)\} }
			\geq \ln\delta^{-1}. 
		\end{split}
	\end{align}
	Equations \eqref{eq:78iidTH3} and \eqref{eq:79iidTH3} together confirm the inequality  
	$B_{N,k}(\nu r\epsilon)\geq1-\delta$ and complete the proof.
\end{proof}

\subsection{Verification with a fixed number of allowed failures}\label{app:IIDsetting6D}
In the Methods section, we have considered the verification of quantum states with a fixed error rate, in which the number
of allowed failures $k$ is proportional to the number of tests $N$.
In this subsection we consider the case in which the number
of allowed failures $k$ is a fixed integer. 
Similar to $N_{k}(\epsilon,\delta,\lambda)$ for the adversarial scenario, for $k\in \bbZ^{\geq 0}$, $0<\epsilon,\delta<1$, and $0\leq\lambda<1$, we define 
\begin{align}\label{eq:DefNTConstiid}
	N_{k}^{\iid}(\epsilon,\delta,\lambda)
	:= \min\{ N\geq k+1\,|\, \bar{\epsilon}^{\,\iid}_{\lambda}(k,N,\delta)\leq \epsilon \}
	= \min\{ N\geq k+1\,| B_{N,k}(\nu\epsilon) \leq \delta \} , 
\end{align}
where the second equality follows from \lref{lem:epsiidANDBnk}. 
According to the monotonicity of $\bar{\epsilon}^{\,\iid}_{\lambda}(k,N,\delta)$ in Proposition \ref{prop:epsiidMonoton}, 
if the number of tests $N$ satisfies $N\geq N_{k}^{\iid}(\epsilon,\delta,\lambda)$, then 
$\bar{\epsilon}^{\,\iid}_{\lambda}(k,N,\delta)\leq \epsilon$.

The following proposition provides informative upper and lower bounds for $N_{k}^{\iid}(\epsilon,\delta,\lambda)$. 
It is the counterpart of Proposition~\ref{prop:Nk} in \ref{app:relationQSV}.

\begin{proposition}\label{prop:BoundNumIIDinfor}
	Suppose $0<\epsilon,\delta<1$, $0\leq\lambda<1$, and $k\in\bbZ^{\geq 0}$. Then we have
	\begin{align}\label{eq:UBiidNfixk}
		\frac{\ln\delta}{\ln (1-\nu\epsilon) }
		\leq N_{k}^{\iid}(\epsilon,\delta,\lambda)
		\leq \Biggl\lceil \frac{k+ \ln\delta^{-1}+\sqrt{(\ln\delta^{-1})^2+2k\ln\delta^{-1}}}{\nu\epsilon} \Biggr\rceil
		\leq \Biggl\lceil \frac{k+2\ln\delta^{-1}+\sqrt{2k\ln\delta^{-1}}}{\nu\epsilon} \Biggr\rceil.
	\end{align}
	If in addition $\delta\leq 1/2$, then
	\begin{align}\label{eq:LBiidNfixk}
		N_{k}^{\iid}(\epsilon,\delta,\lambda) \geq  \frac{k}{\nu\epsilon} .
	\end{align}
\end{proposition}

Next, we consider the high-precision limit $\epsilon,\delta \to 0$. 
In this case, as clarified in the following proposition, the resource cost of our protocol is 
inversely proportional to the spectral gap $\nu$. 
This result is the counterpart of Proposition~\ref{prop:MLT} in \ref{app:relationQSV}. 

\begin{proposition}\label{prop:iidMLT}
	Suppose $0<\epsilon,\delta< 1$, $0\leq \lambda<1$, and $k\in\bbZ^{\geq 0}$. Then
	\begin{align}
		\lim_{\delta\to 0} \frac{N_{k}^{\iid}(\epsilon,\delta,\lambda)}{\ln \delta^{-1}}
		&=\frac{-1}{\ln (1-\nu\epsilon)}, \label{eq:iiddel0}
		\\
		\lim_{\delta,\epsilon\to 0} \frac{\epsilon N_{k}^{\iid}(\epsilon,\delta,\lambda)}{\ln \delta^{-1}}
		&=\frac{1}{\nu}, \label{eq:iidepsdel0}
	\end{align}
	where the order of the two limits $\epsilon\to 0$ and $\delta\to 0$ does not matter in \eref{eq:iidepsdel0}.
\end{proposition}

\begin{proof}[Proof of Proposition \ref{prop:BoundNumIIDinfor}]
	The upper bounds in \eref{eq:UBiidNfixk} can be derived as follows, 
	\begin{align}
		N_{k}^{\iid}(\epsilon,\delta,\lambda)
		&= \min\{ N\geq k+1\,| B_{N,k}(\nu\epsilon) \leq \delta \}
		\stackrel{(a)}{\leq} \min\left\{ N\geq \frac{k}{\nu\epsilon}\,\bigg| \rme^{- D\left(k/N \| \nu\epsilon\right) N}  \leq \delta \right\}
		\nonumber \\
		&= \min\left\{ N\geq \frac{k}{\nu\epsilon}\,\bigg| D(k/N \| \nu\epsilon) N \geq  \ln\delta^{-1} \right\}
		\stackrel{(b)}{\leq} \min\left\{ N\geq \frac{k}{\nu\epsilon}\,\bigg| \frac{(k/N -\nu\epsilon)^2}{2\nu\epsilon} N  \geq \ln\delta^{-1} \right\}
		\nonumber \\
		&= \min\left\{ N\geq \frac{k}{\nu\epsilon}\,\bigg|\, (\nu\epsilon)^2 N^2-2\nu\epsilon(k+\ln\delta^{-1})N+k^2\geq0 \right\}
		\nonumber \\
		&= \Biggl\lceil \frac{k+ \ln\delta^{-1}+\sqrt{(\ln\delta^{-1})^2+2k\ln\delta^{-1}}}{\nu\epsilon} \Biggr\rceil
		\stackrel{(c)}{\leq} \Biggl\lceil \frac{k+2\ln\delta^{-1}+\sqrt{2k\ln\delta^{-1}}}{\nu\epsilon} \Biggr\rceil,  
	\end{align}
	where $(a)$ follows from the Chernoff bound \eqref{eq:ChernoffB}; 
	$(b)$ follows from the inequality $D(x\|y)\geq(x-y)^2/(2y)$ for $x\leq y$;
	and $(c)$ follows from the inequality $\sqrt{x+y}\leq\sqrt{x}+\sqrt{y}$ for $x,y\geq0$.
	
	The lower bound in \eref{eq:UBiidNfixk} can be derived as follows,  
	\begin{align}
		\begin{split}
			N_{k}^{\iid}(\epsilon,\delta,\lambda)
			= \min\{ N\geq k+1\,| B_{N,k}(\nu\epsilon) \leq \delta \}
			\geq \min\{ N\geq k+1\,|\, (1-\nu\epsilon)^N \leq \delta \}
			=\biggl\lceil\frac{\ln \delta}{\ln (1-\nu\epsilon)}\biggr\rceil, 
		\end{split}
	\end{align}		
	where the inequality holds because $B_{N,k}(\nu\epsilon)\geq (1-\nu\epsilon)^N$. 
	
	For $0<\delta\leq 1/2$, if the integer $N$ satisfies $k+1 \leq N<k/(\nu\epsilon)$, then Proposition~\ref{prop:kboundsiid} implies that 
	$\bar{\epsilon}^{\,\iid}_{\lambda}(k,N,\delta)>\epsilon$. 
	This fact and the definition in \eref{eq:DefNTConstiid} together confirm \eref{eq:LBiidNfixk}. 
\end{proof}

\begin{proof}[Proof of Proposition \ref{prop:iidMLT}]
	According to the discussions in \ref{app:iidrelation}, Eq.~(90) of 
	our companion paper \cite{Classical22} implies \eref{eq:iiddel0}, and Eq.~(87) of Ref.~\cite{Classical22} implies that 
	\begin{align}\label{eq:iideps0}
		N_{k}^{\iid}(\epsilon,\delta,\lambda)
		&=\frac{t_\mathrm{P}(k,\delta)}{\nu\epsilon}+o\big(\epsilon^{-1}\big) ,
	\end{align}
	where 
	$t_\mathrm{P}(k,\delta)$ is defined as the unique solution of $x$ to the equation
	\begin{align}\label{eq:t1kdelta}
		\rme^{-x} \sum_{j=0}^k \frac{x^j}{j!} = \delta, \qquad x\geq 0.
	\end{align}
	In addition, it is not difficult to verify that  $\lim_{\delta\to 0}\! \left[ t_\mathrm{P}(k,\delta)/(\ln \delta^{-1})\right] =1$. 
	Therefore, we have 
	\begin{align}
		\lim_{\delta\to 0} \lim_{\epsilon\to 0} \frac{\epsilon N_{k}^{\iid}(\epsilon,\delta,\lambda) }{\ln \delta^{-1}}
		=\lim_{\delta\to 0} \frac{t_\mathrm{P}(k,\delta)}{\nu \ln \delta^{-1}} 
		=\frac{1}{\nu}. 
	\end{align}
	This fact and \eref{eq:iiddel0} together imply \eref{eq:iidepsdel0}. 
\end{proof}

\subsection{The number of allowed failures for a given number of tests}\label{app:IIDsetting6E}

As in the adversarial scenario, to  construct a concrete verification protocol, we need to specify the number $k$ of allowed failures in addition to the number $N$ of tests, so that the conditions of soundness and robustness are satisfied simultaneously. A small $k$ is preferred to guarantee soundness, while a larger $k$ is preferred to guarantee robustness. To construct a robust and efficient verification protocol, we need to find a good balance between the two conflicting requirements. 
The following proposition provides a suitable interval for  $k$ that can guarantee soundness;

\begin{proposition}\label{prop:kboundsiid}
	Suppose $0\leq \lambda<1$, $0<\epsilon <1$, $0<\delta\leq1/2$, and $N,k\in\bbZ^{\geq 0}$. 
	If $\nu \epsilon N\leq k\leq N-1$, then  
	$\bar{\epsilon}^{\,\iid}_{\lambda}(k, N, \delta)>\epsilon$. 
	If $0\leq k\leq l_{\iid}(\lambda, N, \epsilon,\delta)$, then $\bar{\epsilon}^{\,\iid}_{\lambda}(k, N, \delta)\leq\epsilon$. Here
	\begin{align}\label{eq:defineliid}
		l_{\iid}(\lambda, N, \epsilon,\delta):=
		\bigl\lfloor\nu \epsilon N - \sqrt{  2 \nu \epsilon N \ln\delta^{-1}} \bigr\rfloor.
	\end{align}
\end{proposition}

\begin{figure}
	\begin{center}
		\includegraphics[width=7.6cm]{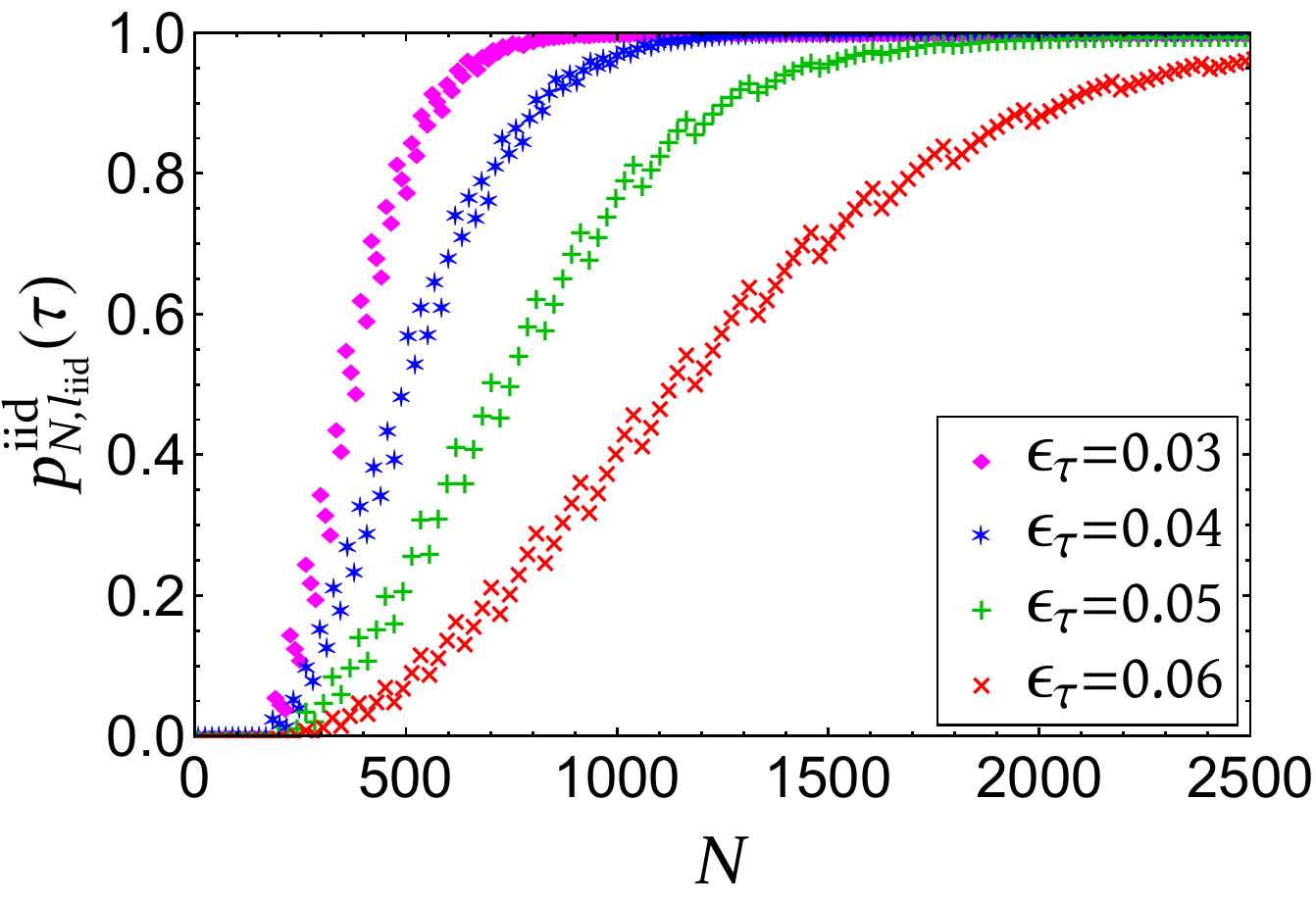}
		\caption{\label{fig:plot_pNkiid}
			The probability $p^{\iid}_{N,l_{\iid}(\lambda, N, \epsilon,\delta)}(\tau)$ that Alice accepts 
			i.i.d.\! quantum states $\tau\in \mathcal{D}(\caH)$ with various infidelities. Here 
			$\lambda=1/2$, $\epsilon=0.1$, $\delta=0.01$;
			$\epsilon_\tau$ denotes the infidelity between $\tau$ and the target state $|\Psi\>$; 
			and $l_{\iid}(\lambda, N, \epsilon,\delta)$ is the number of allowed failures defined in \eref{eq:defineliid}.
			This figure is the counterpart of Fig.~4 in the main text.
		}
	\end{center}
\end{figure}

Next, we turn to the condition of robustness.
When  Bob prepares i.i.d.\! quantum states $\tau\in \mathcal{D}(\caH)$ with $0<\epsilon_\tau<\epsilon$, 
the probability that Alice accepts $\tau$ reads $p^{\iid}_{N,k}(\tau)=B_{N,k}(\nu\epsilon_\tau)$, 
which is strictly increasing in $k$ according to \lref{lem:Bzkmono}.  
Suppose $k=l_{\iid}(\lambda, N, \epsilon,\delta)$; then  we have the following asymptotic relation (in the limit of large $N$) between the acceptance probability and the number $N$ of tests, 
\begin{align}\label{eq:pkLimitIID}
	p^{\iid}_{N,l_{\iid}}(\tau)=
	1-\exp\bigl[ - D( \nu\epsilon \| \nu\epsilon_\tau ) N +O(\sqrt{N}\,) \bigr] \qquad \forall\, 0<\epsilon_\tau<\epsilon,  
\end{align}
where $l_{\iid}$ is a shorthand for $l_{\iid}(\lambda, N, \epsilon,\delta)$. 
Equation \eqref{eq:pkLimitIID} is the counterpart of  \eref{eq:pkLimit} in the main text, and can be 
derived in a similar way. 
This equation shows that 
the probability of acceptance can be arbitrarily close to one as long as $N$ is sufficiently large, as illustrated in Supplementary Figure~\ref{fig:plot_pNkiid}.
Therefore,  our verification protocol can reach any degree of robustness in the i.i.d.\! scenario.

\begin{proof}[Proof of Proposition~\ref{prop:kboundsiid}]
	When $\nu \epsilon N\leq k\leq N-1$ we have 
	\begin{align}\label{eq:lem14Bnk>}
		B_{N,k}(\nu\epsilon)\stackrel{(a)}{\geq} B_{N,k}(k/N)\stackrel{(b)}{>} 1/2 \geq \delta,
	\end{align}
	where $(a)$ follows from \lref{lem:Bzkmono}, and $(b)$ follows from \lref{lem:Bnk1-k/n>1/2}. 
	Equation \eqref{eq:lem14Bnk>} implies that $\bar{\epsilon}^{\,\iid}_{\lambda}(k, N, \delta)>\epsilon$ according to \lref{lem:epsiidANDBnk}, which
	confirms the first statement of Proposition~\ref{prop:kboundsiid}.
	
	Given $0\leq \lambda<1$, $0<\epsilon <1$, $0<\delta\leq1$, and $N,k\in\bbZ^{\geq 0}$, we have
	\begin{align}
		&k \leq l_{\iid}(\lambda, N, \epsilon,\delta)
		&
		&\Rightarrow \quad 
		k \leq \nu \epsilon N - \sqrt{  2 \nu \epsilon N \ln\delta^{-1}}&
		\nonumber \\   \Rightarrow\quad
		&\frac{k}{N}\leq \nu \epsilon ,\ 
		\frac{(k/N-\nu\epsilon)^2}{2\nu\epsilon}N \geq \ln\delta^{-1}
		&
		&\stackrel{(a)}{\Rightarrow} \quad 
		\frac{k}{N}\leq \nu \epsilon ,\ 
		D\!\left( \frac{k}{N}\bigg\| \nu \epsilon\right)N  \geq \ln\delta^{-1}&
		\\   \stackrel{(b)}{\Rightarrow} \quad
		&B_{N,k}( \nu \epsilon) \leq \exp\!\left[ -D\!\left( \frac{k}{N}\bigg\| \nu \epsilon\right)N \right]  \leq \delta
		&
		&\stackrel{(c)}{\Rightarrow} \quad 
		\bar{\epsilon}^{\,\iid}_{\lambda}(k, N, \delta)\leq\epsilon, &\nonumber 
	\end{align}
	which confirm the second statement of Proposition~\ref{prop:kboundsiid}.
	Here the notation $A\Rightarrow B$ means  $A$ is a sufficient condition of $B$; 
	$(a)$ follows from the inequality $D(x\|y)\geq(x-y)^2/(2y)$ for $x\leq y$; 
	$(b)$ follows from the Chernoff bound \eqref{eq:ChernoffB}; and
	$(c)$ follows from \lref{lem:epsiidANDBnk}. 
\end{proof}

\subsection{Explanation of Algorithm~\ref{alg:iidNoptAdv}}\label{app:IIDsetting6F}
For $0\leq\lambda<1$, $0<\delta,\epsilon<1$ and $0\leq r<1$, the  minimum number of tests required for  robust verification  in the i.i.d.\! scenario is denoted by $N_{\rm min}^{\rm iid}(\epsilon,\delta,\lambda,r)$ as defined in \eref{eq:iidOptNadvDef2} of the Methods section.  
Algorithm~\ref{alg:iidNoptAdv} proposed in the Methods section can be used to calculate $N_{\rm min}^{\rm iid}(\epsilon,\delta,\lambda,r)$ and 
the corresponding number of allowed failures.

To understand why Algorithm~\ref{alg:iidNoptAdv} works,
denote by $k^{\rm iid}_{\rm min}(\epsilon,\delta,\lambda,r)$ the minimum nonnegative integer $k$ such that \eref{eq:robustConditionIID} in the Methods section holds 
for some $N\in\bbZ^{\geq k+1}$. 
When $r=0$ for example, we have $k_{\rm min}^{\rm iid}(\epsilon,\delta,\lambda,r)=0$, because $B_{N,0}(0)=1$ and 
$\lim_{N\to\infty}B_{N,0}(\nu\epsilon)=0$. 
To simplify the notation, we shall use $k_{\rm min}^{\rm iid}$ as a shorthand for $k_{\rm min}^{\rm iid}(\epsilon,\delta,\lambda,r)$. 
Then $N_{\rm min}^{\rm iid}(\epsilon,\delta,\lambda,r)$ can be expressed as 
\begin{align}
	N_{\rm min}^{\rm iid}(\epsilon,\delta,\lambda,r)&\;=
	\min_{k\geq k_{\rm min}^{\rm iid}} N_{k}^{\rm iid}(\epsilon,\delta,\lambda,r), \label{eq:iidDefN_kB3}
	\\
	N_{k}^{\rm iid}(\epsilon,\delta,\lambda,r)
	&:= 
	\min \left\{N\in\bbZ^{k+1} \big| B_{N,k}(\nu\epsilon)\leq\delta, B_{N,k}(\nu r\epsilon)\geq1-\delta \right\} 
	\nonumber \\
	&\;= \min \left\{N\in\bbZ^{k+1} \big| B_{N,k}(\nu\epsilon)\leq\delta\right\}  
	\qquad \forall\, k\geq k_{\rm min}^{\rm iid}. 
	\label{eq:iidDefN_kG39} 
\end{align}
The second equality in \eref{eq:iidDefN_kG39} follows from the definition of $k_{\rm min}^{\rm iid}$ and the fact that $B_{N,k}(\nu r\epsilon)$ is nonincreasing in $N$ for $N\geq k$ according to \lref{lem:Bzkmono}. 
Here $N_{k}^{\rm iid}(\epsilon,\delta,\lambda,r)$ reduces to $N_{k}^{\rm iid}(\epsilon,\delta,\lambda)$ 
defined in \eref{eq:DefNTConstiid} when $k\geq k^{\rm iid}_{\rm min}$. 
Note that $B_{N,k}(\nu\epsilon)$ is strictly increasing in $k$ for $0\leq k\leq N$ by \lref{lem:Bzkmono}, 
so $N_{k}^{\rm iid}(\epsilon,\delta,\lambda,r)$ is nondecreasing in $k$ for $k\geq k_{\rm min}^{\rm iid}$, which implies that 
\begin{align}\label{eq:IIDoptNRewriteB5}
	N_{\rm min}^{\rm iid}(\epsilon,\delta,\lambda,r)=N_{k_{\rm min}^{\rm iid}}^{\rm iid}(\epsilon,\delta,\lambda,r).
\end{align} 
By definition \eqref{eq:iidDefN_kG39}, \eref{eq:robustConditionIID} in the Methods section holds when 
$N=N^{\rm iid}_{\rm min}(\epsilon,\delta,\lambda,r)$ and $k=k^{\rm iid}_{\rm min}(\epsilon,\delta,\lambda,r)$. 

In Algorithm~\ref{alg:iidNoptAdv}, the aim of steps 1--11 is to find $k_{\rm min}^{\rm iid}$. 
In particular, steps 1--2 aim to find $k_{\rm min}^{\rm iid}$ in the case $r=0$; 
steps 3--10 aim to find $k_{\rm min}^{\rm iid}$ in the case $r>0$ by virtue of the fact that both 
$B_{M,k}(\nu r\epsilon)$ and $B_{M,k}(\nu\epsilon)$ are strictly decreasing in $M$ for $M\geq k$ according to \lref{lem:Bzkmono}.   
Steps 12--13 aim to find $N_{\rm min}^{\rm iid}$ by virtue of Eqs.~\eqref{eq:iidDefN_kG39} and~\eqref{eq:IIDoptNRewriteB5}.

\subsection{Upper bounds for $N_{\rm min}^{\rm iid}(\epsilon,\delta,\lambda,r)$ in  \eref{eq:iidNoptUB} of the Methods section}\label{app:IIDsetting6G}

\subsubsection{Main proof and illustration}

Equation \eqref{eq:iidNoptUB} in the Methods section can be proved by virtue of Proposition \ref{prop:iidHighProbiid} in the Methods section, which provides a guideline for choosing appropriate parameters $N$ and $k$ to  achieve a given verification precision and robustness.
In practice, we need to  choose a suitable error rate $s$, so that the number of tests in Proposition \ref{prop:iidHighProbiid} is as small as possible.
To this end, we determine the maximum of the denominator in \eref{eq:choosesknIID} in the Methods section. Since $D(\nu s\|\nu r\epsilon)$ is nondecreasing in $s$, while $D(\nu s\|\nu\epsilon)$ is nonincreasing in $s$ for $s\in (r\epsilon,\epsilon)$, 
the above maximum value is attained when $D(\nu s\|\nu r\epsilon)=D(\nu s\|\nu\epsilon)$. Based on this observation we can deduce that 
\begin{gather}
	\max_{s\in (r\epsilon,\epsilon)} \min\{D(\nu s\|\nu r\epsilon), D(\nu s\|\nu\epsilon)\} 
	= 
	D\!\left( \frac{\ln\!\big(\frac{1-\nu\epsilon}{1-\nu r\epsilon}\big) }
	{\ln r + \ln\!\big(\frac{1-\nu\epsilon}{1-\nu r\epsilon}\big)} \Bigg\|\nu\epsilon\right),
	\label{eq:maxDomiN} \\
	s_{\iid}(\lambda,r,\epsilon)
	:= \text{argmax}_{s} \min\{D(\nu s\|\nu r\epsilon), D(\nu s\|\nu\epsilon)\} 
	= \frac{\ln\!\big(\frac{1-\nu\epsilon}{1-\nu r\epsilon}\big) }
	{\nu\ln r + \nu\ln\!\big(\frac{1-\nu\epsilon}{1-\nu r\epsilon}\big)}.\label{eq:optimals}
\end{gather}

\begin{figure}[b]
	\includegraphics[width=7.2cm]{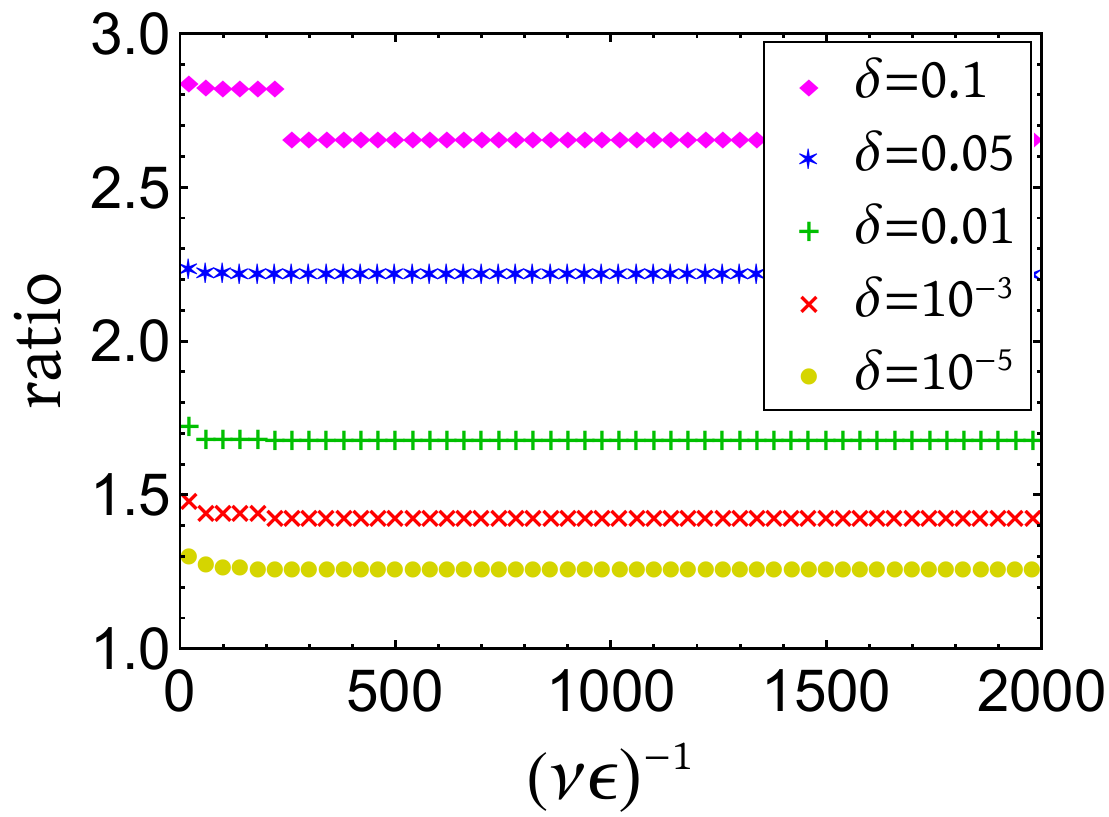}
	\caption{\label{fig:RatioNubNiid}
		The ratio of the second upper bound in \eref{eq:iidNoptUB} of the Methods section over $N_{\rm min}^{\rm iid}(\epsilon,\delta,\lambda,r)$
		with $r=1/2$. 
	}
\end{figure}

\begin{figure}
	\begin{center}
		\includegraphics[width=7.2cm]{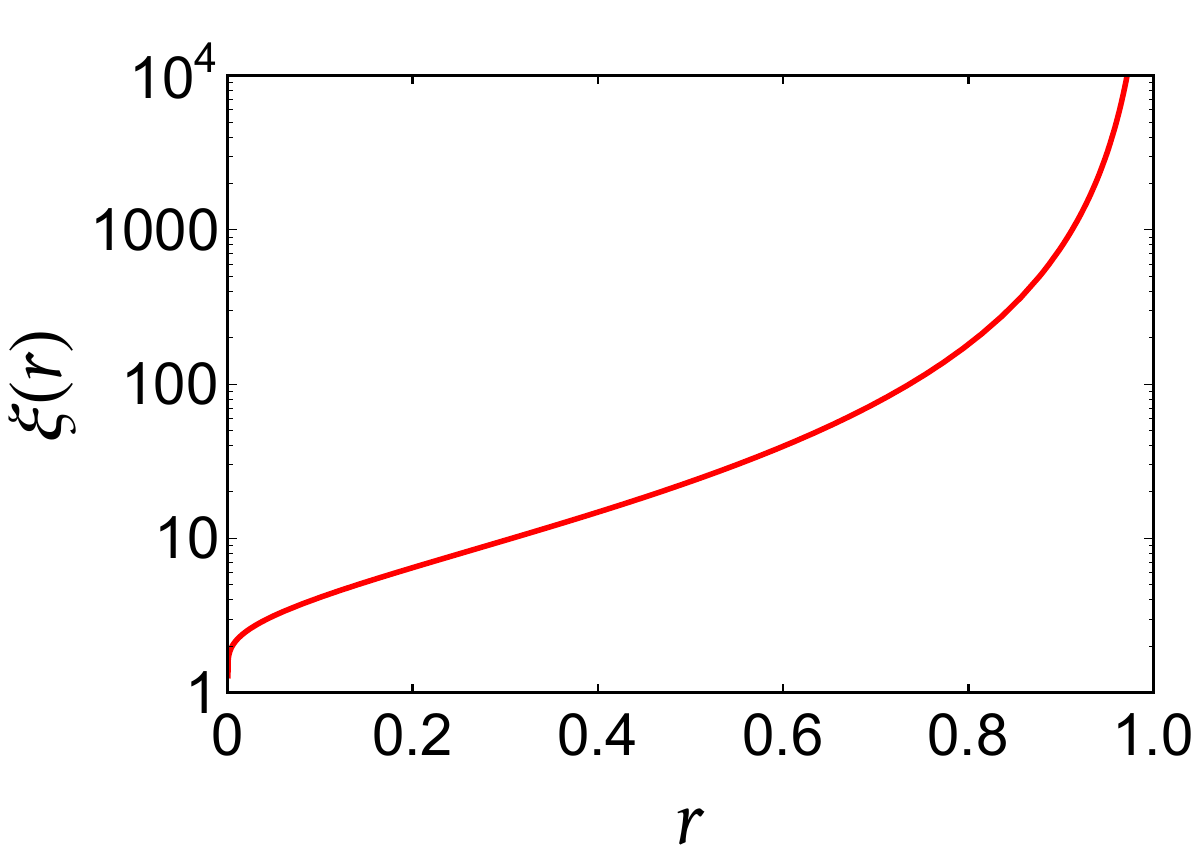}
		\caption{\label{fig:xi_t}
			Variation of the function $\xi(r)$ with $r$. 
		}
	\end{center}
\end{figure}

Thanks to \eref{eq:maxDomiN} and Proposition~\ref{prop:iidHighProbiid}, the conditions of soundness and robustness in \eref{eq:robustConditionIID} of the Methods section hold as long as 
$k=\left\lfloor\nu s_{\iid}(\lambda,r,\epsilon) N\right\rfloor$ and 
\begin{align}\label{eq:Niidt}
	N
	= \left\lceil \frac{\ln\delta^{-1}}{D\!\left( \nu s_{\iid}(\lambda,r,\epsilon) \big\|\nu\epsilon\right)}  \right\rceil  
	= \left\lceil \frac{\ln\delta^{-1}}{\nu\epsilon} \zeta(r,\nu\epsilon) \right\rceil . 
\end{align} 
This fact implies the upper bounds  for $N_{\rm min}^{\rm iid}(\epsilon,\delta,\lambda,r)$ in \eref{eq:iidNoptUB} of the Methods section. 
In particular, the second inequality of \eref{eq:iidNoptUB} in the Methods section follows from \lref{lem:zetaMono} below.

To illustrate the tightness of the bounds in \eref{eq:iidNoptUB} of the Methods section, the ratio of the second upper bound over $N_{\rm min}^{\rm iid}(\epsilon,\delta,\lambda,r)$
is plotted in Supplementary Figure~\ref{fig:RatioNubNiid}. 
In addition, the coefficient $\xi(r)$ of the second upper bound is plotted in Supplementary Figure~\ref{fig:xi_t}.

\subsubsection{Auxiliary lemmas}
\begin{lemma}\label{lem:zetaMono}
	Suppose $0<r<1$; then $\zeta(r,p)$ defined in \eref{eq:zetarpxi} of the Methods section 
	is strictly decreasing in $p$ for $0<p<1$.
\end{lemma}

\begin{proof}[Proof of \lref{lem:zetaMono}]
	For $0<r<1$ and $0<p<1$, define 
	\begin{align}\label{eq:Freqpt}
		\beta(r,p):=\frac{\ln\!\big(\frac{1-p}{1-rp}\big) }{\ln r + \ln\!\big(\frac{1-p}{1-rp}\big)}.   
	\end{align}
	One can easily verify that $D\!\left( \beta(r,p) \big\|p\right) =D\!\left( \beta(r,p) \big\|rp\right)$. 
	This fact and \lref{lem:DpqMonoton} imply that $rp<\beta(r,p)< p$. 
	Therefore, to prove \lref{lem:zetaMono}, it suffices to show that 
	$D\!\left( \beta(r,p) \big\|rp\right)\!/p$ 
	is strictly increasing in $p$ for $0<p<1$. 
	
	For $0<p_1<p_2<1$ we have 
	\begin{align}
		\frac{D\!\left(\beta(r,p_2) \big\|r  p_2\right)}{p_2}
		\stackrel{(a)}{>}
		\frac{1}{p_2} \, D\!\left( p_2 \frac{\beta(r,p_1)}{p_1}  \Bigg\|r  p_2\right) 
		\stackrel{(b)}{\geq}
		\frac{1}{p_1} \, D\!\left( p_1 \frac{\beta(r,p_1)}{p_1}  \Bigg\|r  p_1\right) 
		=
		\frac{D\!\left( \beta(r,p_1) \big\|r  p_1\right)}{p_1} , 
	\end{align}
	which confirms \lref{lem:zetaMono}. 
	Here $(a)$ follows from \lref{lem:DpqMonoton} and the inequality
	$rp_2<\beta(r,p_1)p_2/p_1< \beta(r,p_2)$ by \lref{lem:f/pIncrease} below; 
	$(b)$ follows because $\frac{1}{q}  D\!\left( \left. q \frac{\beta(r,p_1)}{p_1}  \right\|\! r  q\right)$ 
	is nondecreasing in $q$ for $0<q<1$ according to \lref{lem:D(ap,p)/pMono} below. 
\end{proof}

\begin{lemma}\label{lem:D(ap,p)/pMono}
	Suppose $0<r,a,q<1$; then 
	$D\!\left(aq \|rq\right)\!/q$ is nondecreasing in $q$.  
\end{lemma}

\begin{proof}[Proof of \lref{lem:D(ap,p)/pMono}]
	We have 
	\begin{align}
		\frac{{\rm d}[D\!\left(aq \|rq\right)\!/q]}{{\rm d} q} 
		&= \frac{q\big[ \frac{{\rm d}}{{\rm d} q}D\!\left(aq \|rq\right)\!\big] -D\!\left(aq \|rq\right)}{q^2}.  
	\end{align} 
	In the following, we shall prove that the numerator of the RHS satisfies  
	\begin{align}\label{eq:der>0H47}
		q\bigg[ \frac{{\rm d}}{{\rm d} q}D\!\left(aq \|rq\right)\!\bigg] -D\!\left(aq \|rq\right)\geq0 \qquad \forall 0<r,a,q<1,
	\end{align}
	which implies \lref{lem:D(ap,p)/pMono}. 
	First, note that $\lim_{q\to 0^+}\!\big( q\big[ \frac{{\rm d} D(aq \|rq)}{{\rm d} q}\big] -D\!\left(aq \|rq\right)\big) =0$.  
	Therefore, to prove \eref{eq:der>0H47}, it suffices to show that 
	$q\big[ \frac{{\rm d}}{{\rm d} q}D\!\left(aq \|rq\right)\!\big] -D\!\left(aq \|rq\right)$ is nondecreasing in $q$ for $0<q<1$, 
	that is, 
	\begin{align}
		\frac{{\rm d}}{{\rm d} q}  \left\{ q\bigg[ \frac{{\rm d}}{{\rm d} q}D\!\left(aq \|rq\right)\!\bigg] -D\!\left(aq \|rq\right) \right\}
		=q\left[ \frac{{\rm d^2}D\!\left(aq \|rq\right)}{{\rm d} q^2} \right] 
		\geq0. 
	\end{align}
	This inequality can be proved as follows, 
	\begin{align}
		\frac{{\rm d}D\!\left(aq \|rq\right)}{{\rm d} q} 
		&=
		a\ln\left( \frac{a}{r}\right) -a\ln\bigg( \frac{1-aq}{1-rq}\bigg)  +\frac{r-a}{1-rq} , 
		\\
		\frac{{\rm d^2}D\!\left(aq \|rq\right)}{{\rm d} q^2} 
		&=
		\frac{(r-a)^2}{(1-aq)(1-rq)^2} 
		\geq0 , 
	\end{align} 
	which completes the proof of \lref{lem:D(ap,p)/pMono}. 
\end{proof}

\begin{lemma}\label{lem:f/pIncrease}
	Suppose $0<r,p<1$; then 
	$\beta(r,p)/p$ is strictly increasing in $p$, where $\beta(r,p)$ is defined in \eref{eq:Freqpt}. 
\end{lemma}

\begin{proof}[Proof of \lref{lem:f/pIncrease}]
	The proof is composed of three steps.
	
	\noindent \textbf{Step 1:} The aim of this step is to prove the following inequality,
	\begin{align}\label{eq:f/pInStep1Ineq1}
		2(r-1)- (\ln r)(r+1) >0 \qquad \forall \, 0<r<1. 
	\end{align}

	The following relations hold when $0<r<1$:  
	\begin{align}
		\begin{split}
			\text{\eref{eq:f/pInStep1Ineq1}}
			&\quad \stackrel{(a)}{\Leftarrow} \quad
			0> \frac{{\rm d}}{{\rm d} r}	\bigl[ 2(r-1)- (\ln r)(r+1) \bigr]
			=1-\frac{1}{r}-\ln r 
			\\
			&\quad \stackrel{(b)}{\Leftarrow} \quad
			0< \frac{{\rm d}}{{\rm d} r}	\Big( 1-\frac{1}{r}-\ln r \Big)
			= \frac{1-r}{r^2} , 
		\end{split}
	\end{align}
	which confirm \eref{eq:f/pInStep1Ineq1}. 
	Here the notation $A\Leftarrow B$ means  $B$ is a sufficient condition of $A$; 
	$(a)$ holds because the LHS of \eref{eq:f/pInStep1Ineq1} equals 0 when $r=1$; 
	and $(b)$ holds because $1-1/r-\ln r=0$ when $r=1$.

	\noindent \textbf{Step 2:} The aim of this step is to prove the following inequality.
	\begin{align}\label{eq:f/pInStep2Ineq2}
		2(1+r)(1-r)+4\,r(\ln r) >0 \qquad \forall \, 0<r<1. 
	\end{align}
	
	The following relation holds when $0<r<1$:  
	\begin{align}
		\text{\eref{eq:f/pInStep2Ineq2}}
		&\quad \stackrel{(a)}{\Leftarrow} \quad
		0>
		\frac{{\rm d}}{{\rm d} r}\bigl[ 2(1+r)(1-r)+4\,r(\ln r) \bigr]
		=4(1+\ln r-r), 
	\end{align}
	which confirms \eref{eq:f/pInStep2Ineq2}. 
	Here $(a)$ holds because the LHS of \eref{eq:f/pInStep2Ineq2} equals 0 when $r=1$.

	\noindent \textbf{Step 3:} The aim of this step is to prove \lref{lem:f/pIncrease}. 
	Let $\gamma_r(p):=\ln\frac{1-p}{1-rp}$, which satisfies $\gamma_r(0)=0$ and 
	\begin{align}
		\gamma_r'(p) = \frac{r-1}{(1-rp)(1-p)}, 
		\qquad\quad
		\gamma_r''(p) = \frac{(1-r)(2rp-r-1)}{(1-rp)^2(1-p)^2}. 
	\end{align}

	The following relations hold when $0<r<1$ and $0<p<1$:  
	\begin{align}
		\frac{\beta(r,p)}{p} \ &\text{is increasing in}\ p
		\quad \Leftarrow \quad
		\frac{p\ln r}{\gamma_r(p)} + p \ \text{is decreasing in}\ p
		\quad \Leftarrow \quad
		\frac{\partial}{\partial p} \left[ \frac{p\ln r}{\gamma_r(p)} + p \right] <0 
		\nonumber\\
		&\quad \Leftarrow \quad
		(\ln r) \left[ \gamma_r(p) - p\,\gamma_r'(p) \right] + \gamma_r(p)^2 <0
		\label{eq:G43condition}\\
		&\quad \stackrel{(a)}{\Leftarrow} \quad
		\frac{\partial}{\partial p} 
		\bigl[ (\ln r) \left[ \gamma_r(p) - p\,\gamma_r'(p) \right] + \gamma_r(p)^2 \bigr]  <0 
		\nonumber\\
		&\quad \Leftarrow \quad 
		2 \gamma_r(p) \gamma_r'(p) - (\ln r) p \,\gamma_r''(p) <0
		\nonumber\\
		&\quad \Leftarrow \quad
		2 \gamma_r(p) + \frac{(\ln r)(2rp-r-1)p }{(1-rp)(1-p)} >0
		\label{eq:G44condition}\\
		&\quad \stackrel{(b)}{\Leftarrow} \quad 
		\frac{\partial}{\partial p}  \left[ 2 \gamma_r(p) + \frac{(\ln r)(2rp-r-1)p }{(1-rp)(1-p)} \right] > 0
		\nonumber\\ 
		&\quad \Leftarrow \quad
		2(r-1)(rp-1)(p-1) + (\ln r)(4rp -rp^2-r^2p^2-r-1) >0
		\label{eq:G45condition}\\ 
		&\quad \stackrel{(c)}{\Leftarrow} \quad
		\frac{\partial}{\partial p} \left[ 2(r-1)(rp-1)(p-1) + (\ln r)(4rp -rp^2-r^2p^2-r-1) \right]  >0
		\nonumber\\ 
		&\quad \Leftarrow \quad
		2(r-1)(2rp-1-r) + (\ln r)(4r -2rp-2r^2p) >0
		\label{eq:G46condition}\\ 
		&\quad \stackrel{(d)}{\Leftarrow} \quad
		\frac{\partial}{\partial p} \left[ 2(r-1)(2rp-1-r) + (\ln r)(4r -2rp-2r^2p) \right]  >0
		\nonumber\\ 
		&\quad \Leftarrow \quad
		2(r-1)- (\ln r)(r+1) >0. \label{eq:G47condition}
	\end{align}
	Here 
	$(a)$ holds because the LHS of \eref{eq:G43condition} equals 0 when $p=0$; 
	$(b)$ holds because the LHS of \eref{eq:G44condition} equals 0 when $p=0$;
	$(c)$ holds because the LHS of \eref{eq:G45condition} equals $2(r-1)-(\ln r)(r+1)$ when $p=0$, which is positive according to \eref{eq:f/pInStep1Ineq1};
	$(d)$ holds because the LHS of \eref{eq:G46condition} equals $2(1+r)(1-r)+4\,r(\ln r)$ when $p=0$, which is positive according to \eref{eq:f/pInStep2Ineq2}; 
	Eq.~\eqref{eq:G47condition} holds according to \eref{eq:f/pInStep1Ineq1}, which confirms \lref{lem:f/pIncrease}. 
\end{proof}

\end{document}